\documentclass[12pt]{article}
\usepackage{graphicx} 
\usepackage{amsmath}
\usepackage{amsfonts}
\usepackage{amssymb}
\usepackage{amsthm}
\usepackage{pgfplots}
\usepackage[round]{natbib}
\usepackage{appendix}
\usepackage{comment}
\usepackage{authblk}
\usepackage{mathdots}
\usepackage{enumitem}
\usepackage[linesnumbered,ruled,vlined]{algorithm2e}
\let\oldnl\nl
\newcommand{\nonl}{\renewcommand{\nl}{\let\nl\oldnl}}

\def\@centernot#1#2{%
  \mathrel{%
    \rlap{%
      \settowidth\dimen@{$\m@th#1{#2}$}%
      \kern.5\dimen@
      \settowidth\dimen@{$\m@th#1=$}%
      \kern-.5\dimen@
      $\m@th#1\not$%
    }%
    {#2}%
  }%
}
\makeatother

\newcommand{\blind}{1}

\usepackage{tikz}
\usetikzlibrary{arrows.meta,shapes}
\usetikzlibrary{arrows,shapes.arrows,shapes.geometric,shapes.multipart,
decorations.pathmorphing,positioning,shapes.swigs,}

\newtheorem{lemma}{Lemma}
\newtheorem{proposition}{Proposition}
\newtheorem{theorem}{Theorem}
\newtheorem{corollary}{Corollary}
\theoremstyle{definition}
\newtheorem{assumption}{Assumption}
\newtheorem{definition}{Definition}
\theoremstyle{remark}
\newtheorem{remark}{Remark}

\newcommand{\independent}{\perp\mkern-9.5mu\perp}

\DeclareUnicodeCharacter{00A0}{ }

\DeclareMathOperator*{\argmax}{argmax}
\DeclareMathOperator*{\argmin}{argmin}

\pgfplotsset{compat=1.18}

\usepackage{apptools}

\AtAppendix{\counterwithin{lemma}{section}}  
\AtAppendix{\counterwithin{proposition}{section}}
\AtAppendix{\counterwithin{theorem}{section}}
\AtAppendix{\counterwithin{corollary}{section}}
\AtAppendix{\counterwithin{assumption}{section}}
\AtAppendix{\counterwithin{definition}{section}}
\AtAppendix{\counterwithin{remark}{section}}

\begin{document}

\def\spacingset#1{\renewcommand{\baselinestretch}%
{#1}\small\normalsize} \spacingset{1}



\if1\blind
{
  \title{\bf Optimal sequential decision-making with initiation regimes}
  \author[1]{Julien D. Laurendeau \thanks{
    The authors gratefully acknowledge support from the Swiss National Science Foundation.}\hspace{.2cm}\thanks{julien.laurendeau@epfl.ch}}
  \author[2]{Leora Sarvet \thanks{asarvet@umass.edu}}
  \author[1]{Mats J. Stensrud \thanks{mats.stensrud@epfl.ch}}
    \affil[1]{Institute of Mathematics, Ecole Polytechnique Fédérale de Lausanne, Station 8, 1015 Lausanne, Switzerland}
    \affil[2]{School of Public Health \& Health Sciences, University of Massachusetts Amherst, USA}
    \date{}                                                       
  \maketitle
} \fi

\if0\blind
{
  \bigskip
  \bigskip
  \bigskip
  \begin{center}
    {\LARGE\bf Optimal sequential decision-making with initiation regimes}
\end{center}
  \medskip
} \fi

\begin{abstract}
Consider an optimal dynamic treatment regime, $g^{\textbf{opt}}$ correctly identified from a large, perfectly executed sequentially randomized experiment. Even when the experimental results are generalizable to a future target population, there is no guarantee that $g^{\textbf{opt}}$ outperforms human decision-makers; human experts can do better than $g^{\textbf{opt}}$ whenever they have access to relevant information beyond the covariates recorded in the experiment. Motivated by this observation, we derive results on a new class of regimes called initiation regimes, which generalize existing results on superoptimal regimes. These regimes follow human decision-makers up to the point where it becomes more beneficial to initiate a sequential optimal regime, and are guaranteed to outperform both purely human and purely algorithmic decision rules, e.g., based on reinforcement learning algorithms. Furthermore, we present modified experimental designs that identify the best initiation regimes, show how the best initiation regime can be identified from classical observational data under explicit assumptions, and give estimation and statistical inference methodology for these regimes. To illustrate the practical utility of the methods, we consider initiation regimes in a case study on treatment of lower back pain. 
\end{abstract}

\noindent%

\spacingset{1.65} 

\section{Introduction}

There is a flourishing methodological literature on optimal dynamic treatment regimes (\cite{murphy2003optimal, robins2004optimal, moodie2007demystifying, moodie2012q, chakraborty2013statistical, chakraborty2014dynamic,schulte2014q_and_a, tsiatis2019dynamic,clifton2020qlearning, kosorok2021introduction}, to name a few). These methods are now implemented in practice, for example, in AI systems using reinforcement learning \citep{kober_reinforcement_2013, silver2016mastering, vinyals2019grandmaster, kendall2019learning, openai_2024_learning,reid2024gemini,meta2024responsible,aws2024responsible, Lgayhardt_2024}, and in the analysis of Sequential Multiple Assignment Randomized Trials (SMARTs) \citep{murphy_experimental_2005, murphy_customizing_2007, murphy2007developing}, see also Appendix \ref{app: related_literature} for a brief review of relevant literature on optimal regimes. Nevertheless, there is skepticism about the practical implementation of estimated regimes used in AI systems, despite the existence of certain theoretical optimality guarantees \citep{verghese2018computer, matheny2019artificial}. A main concern is that human decision makers have access to relevant information that is not easily encoded in datasets. For example, doctors often receive visual or auditory cues from a patient, which are subsequently used in the process of making decisions \citep{hamerman1999toward}. Leveraging this extra information, human decision makers can outperform algorithmic regimes learned from perfect sequential experiments, see, e.g., \citet[Fine Points 22.7 and 22.8]{hernan_causal_2024} for a discussion on how this can occur,  \citet{zhang2022can,stensrud_optimal_2024, raghavan2025counterfactual} for methodological discussions, and \citet{cabitza2017unintended, verghese2018computer} for concern in the medical literature. In other words, the human decision makers might successfully tailor their treatment decisions based on characteristics that were not used to learn the optimal algorithmic regime.  

Recent work on optimal regimes in a point treatment setting suggests that the existing results can be improved; specifically, \citet{stensrud_optimal_2024} considered regimes that also use an individual's natural treatment value, that is, the treatment value an individual would take in the absence of it being assigned by a regime. 
These regimes are guaranteed to outperform the conventional optimal regimes and also the regime implicitly implemented in the observed data, and have therefore been called superoptimal regimes in point treatment settings. The assumptions needed to identify these regimes in a point treatment setting are subtly different from, but not necessarily stronger than, those needed to identify conventional optimal regimes.   

Here we consider a longitudinal, sequential treatment setting and give new optimality results for regimes that leverage time-varying natural treatment values. The extension from the point treatment setting is not straightforward. One problem is that recording a human's natural treatment preferences over time is often difficult when an algorithm governs the decision process: without the ability to enforce their natural treatment choice, a human has little incentive to articulate or even reflect on it. Moreover, even when asked to report their natural treatment value, the reliability of their response may be compromised, particularly if they have repeatedly experienced decisions that contradict their intentions. Unsurprisingly, sequential analogues of patient preference trials have thus far not been implemented in practice: natural treatment values are not recorded in conventional SMART trials, which assign treatments based on possibly a high-dimensional time-varying vector of covariates, but not natural treatment values. 


Our aim with this article is to construct a sequential algorithmic decision rule that is guaranteed to outperform human decision makers and, at the same time, overcomes the problem of using natural treatment values when an algorithm entirely governs the decision process. The proposed solution is based on creating a class of so-called initiation regimes. The initiation regimes are defined such that an individual initiates a dynamic optimal regime, given their past history of covariates and natural treatment preferences (values), at a time $k$ after follow-up. Until time $k$, however, the individual follows their natural course of action; that is, the individual takes their natural treatment value without any algorithmic intervention. This class includes both the observed regime and the conventional, previously considered algorithmic optimal regime, which uses measured covariate histories alone: the observed regime corresponds to setting $k>K$, where $K$ is the end of the study, and the conventional optimal corresponds to $k=0$, where $k = 1$ is the first treatment time: We call the initiation regime that maximizes the expected reward the optimal initiation regime. Informally, this regime  finds the optimal time $k$ to deviate from an individual's natural treatment decision and, subsequently, for times $j>k$, follows an algorithmic decision rule given time-varying covariates through time $j$ and natural treatment values through time $k$. Under this regime, each individual might deviate from their natural course at different times, depending on their natural covariate and treatment history. A feature of the optimal initiation regime is that it outperforms the conventionally defined optimal regime and the observed regime under a non-trivial class of data-generating mechanisms. We will consider larger classes of regimes that share this feature, but our focus on initiation regimes is also motivated by the fact that these regimes can be identified in a wider range of common data structures. Furthermore, the initiation regimes are easier to estimate in many settings. 

The optimal initiation regime is computable via Bellman equations, similar to the conventional optimal regime in longitudinal settings. We leverage the Bellman equations to derive identification conditions, and describe how these conditions are related to identification results for conventional optimal regimes. Broadly, the initiation regimes are identified under assumptions not necessarily stronger than assumptions used to identify conventional sequential optimal regimes. 


To fix ideas, we will consider a running example on the treatment of chronic back pain, motivated by a large consortium funded by the National Institute of Health (NIH).

\subsection{Optimal treatment of chronic back pain}
\label{sec: chronic back pain ex intro}
Chronic lower back pain is one of the leading causes of disability worldwide, affecting hundreds of millions of individuals \citep{world-health-organization-who-2023}. The population suffering from chronic lower back pain is heterogeneous, and there is interest in finding optimal regimes that are tailored to individual characteristics. The Back Pain Consortium is a NIH-funded research program created to improve diagnosis and treatment of chronic lower back pain. A part of the consortium is the Biomarkers for Evaluating Spine Treatments (BEST) trial, a Sequential Multiple Assignment Randomized Trial (SMART) with 630 patients who suffered from chronic lower back pain. The trial participants were sequentially assigned to be treated with Duloxetine or Enhanced Self Care (ESC) at baseline and to Evidence-Based Exercise and Manual Therapy (EBEM) or Acceptance and Commitment Therapy (ACT) three months thereafter. The outcome of interest was the Pain, Enjoyment of life and General activity (PEG) Score, a measure of pain and pain interference \citep{krebs2009development} ranging from $0$ (no pain) to $10$ (worst pain). Because we aim to maximize a value function, we consider a reversed scale, where $10$ is the best score and $0$ is the worst. Age was recorded at baseline, and opioid usage and depression status were recorded at baseline and 3 months later, and the patient's response to treatment is recorded, see Figure \ref{fig: Batorsky trial illustration} for an illustration adapted from \citet{batorsky2024integrating}. The BEST trial allows identification of the conventional optimal dynamic treatment regime for back pain treatment in the population the trial was conducted in. However, the BEST trial does not record the physician's recommendation for back pain treatment, corresponding to the individual's natural treatment value, that the optimal initiation regime uses to improve over the conventional optimal regime.

\begin{figure}
    \centering
    \begin{tikzpicture}[scale=0.7, transform shape,
    arrow/.style={-{Latex[scale=1.0]}, thick}]

\node[name=R,circle, draw=blue, line width=1mm, text=black, minimum size=1cm] at (0,0) {R};
\node[name=Duloxetine,rectangle, draw=black, fill=gray!50, rounded corners, align=center, text width=4cm, minimum height=1.5cm] at (4,4) {Treatment (1)\\Duloxetine};
\node[name=ESC,rectangle, draw=black, fill=gray!50, rounded corners, align=center, text width=4cm, minimum height=1.5cm] at (4,-4) {Treatment (0)\\ESC};
\node[name=R_h,circle, draw=blue, line width=1mm, text=black, minimum size=1cm] at (7,4) {R};
\node[name=R_l,circle, draw=blue, line width=1mm, text=black, minimum size=1cm] at (7,-4) {R};
\node[name=EBEM_h,rectangle, draw=black, fill=gray!50, rounded corners, align=center, text width=4cm, minimum height=1.5cm] at (12,6) {Treatment (1)\\EBEM};
\node[name=ACT_h,rectangle, draw=black, fill=gray!50, rounded corners, align=center, text width=4cm, minimum height=1.5cm] at (12,2) {Treatment (0)\\ACT};

\node[name=EBEM_l,rectangle, draw=black, fill=gray!50, rounded corners, align=center, text width=4cm, minimum height=1.5cm] at (12,-2) {Treatment (1)\\EBEM};
\node[name=ACT_l,rectangle, draw=black, fill=gray!50, rounded corners, align=center, text width=4cm, minimum height=1.5cm] at (12,-6) {Treatment (0)\\ACT};

\draw[arrow] (R) -- (Duloxetine.197);
\draw[arrow] (R) -- (ESC.163);
\draw[arrow] (R_h) -- (EBEM_h.197);
\draw[arrow] (R_h) -- (ACT_h.163);
\draw[arrow] (R_l) -- (EBEM_l.197);
\draw[arrow] (R_l) -- (ACT_l.163);

\node[rectangle, draw=black, align=center, minimum height=1cm, minimum width=7cm] at (3.5,-8) {Stage 1};
\node[rectangle, draw=black, align=center, minimum height=1cm, minimum width=8cm] at (11,-8) {Stage 2};

\node[name = Cap_1] at (0,-9) {Baseline};
\node[name = Cap_2] at (7,-9) {3 months};
\node[name = Cap_3] at (15,-9) {6 months};

\node[name = Cap_1_1, align = center] at (0,-11.5) {Data collection:\\$L_1$: age, \\opioid
usage, \\depression};
\node[name = Cap_2_2, align = center] at (7,-12) {Data collection:\\$Y_1$: PEG Score\\$L_2$: opioid usage,\\depression,\\responder status};
\node[name = Cap_3_3, align = center] at (15,-10.5) {Data collection:\\ $Y_2$: PEG Score};

\end{tikzpicture}
    \caption{Data Structure of a SMART on treatment of lower back pain, adapted from \citet{batorsky2024integrating}. Here, 'R' represents randomization to Duloxetine or Enhanced Self Care (ESC) at baseline, and to Evidence-Based Exercise and Manual therapy (EBEM) or Acceptance and Commitment Therapy (ACT) 3 months later.}
    \label{fig: Batorsky trial illustration}
\end{figure}
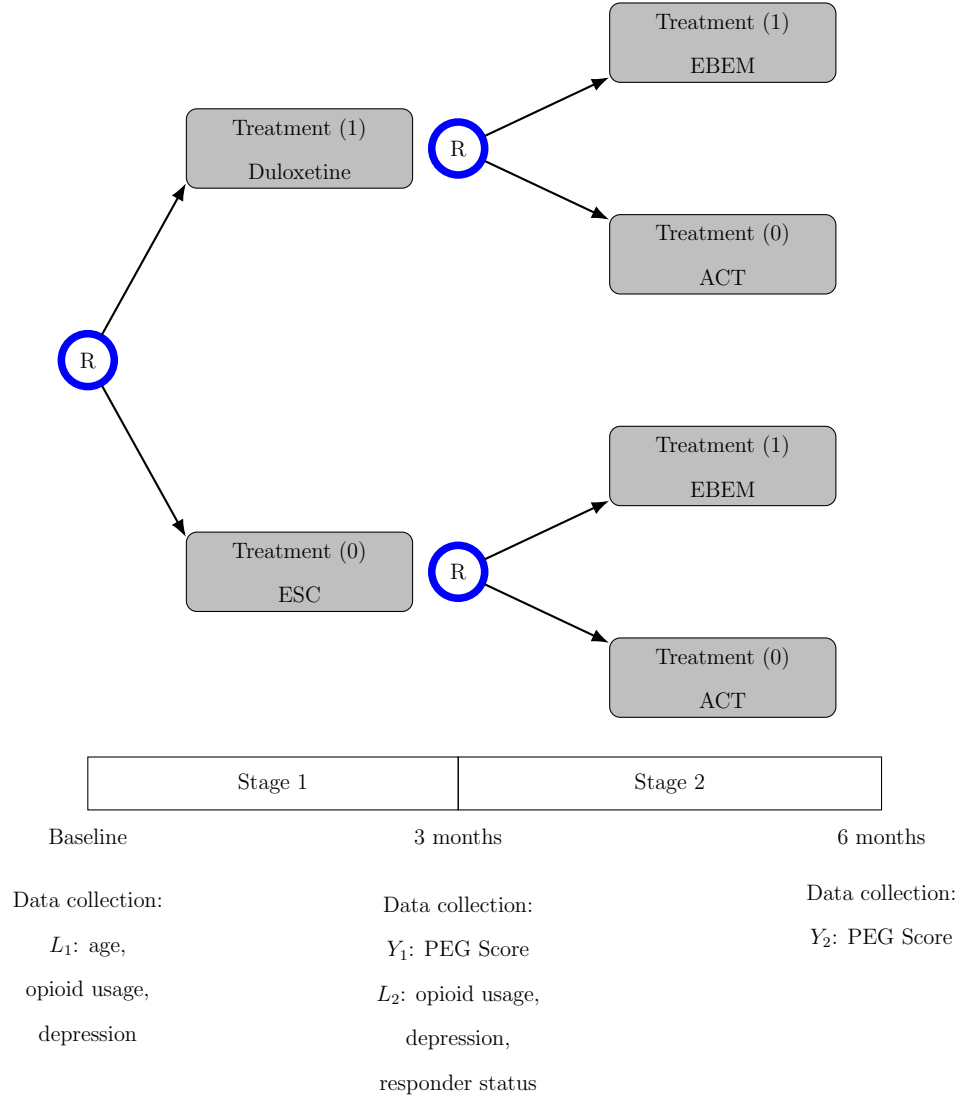

However, following \citet{batorsky2024integrating}, we will consider a hypothetical setting where observational data from 1'000 individuals are additionally available from the same superpopulation as the trial participants. Such observational data could, for example, come from the observational studies that are part of The Back Pain Consortium \citep{BACPAC2019, mauck2023back, batorsky2024integrating}. In the observational data, physicians make decisions for patients based on their own discretion and intentions. We will examine how the trial and observational data help identify regimes that follow the physician's recommendation for assigning Duloxetine or ESC at baseline, and then intervene to give EBEM or ACT depending on covariates and the physician's recommendation for treatment after 3 months.

We use this chronic back pain example from \citet{batorsky2024integrating}, where randomized and observational data are available, as a running example. However, to illustrate when identification is possible and under what assumptions, we also consider other data structures that involve observational or randomized data in Section \ref{sec: data fusion} and thereafter.

\section{Data structure, notation and basic assumptions}
\label{sec: Data structure}
 Suppose that we observe data on time-indexed vectors of covariates $\overline{L}_K = (L_1, \ldots, L_K)$, treatments $\overline{A}_K = (A_1, \ldots, A_K)$ and outcomes, also called rewards, $\overline{R}_K = (R_1, \ldots, R_K)$. There might also exist time-varying unobserved covariates $\overline{U}_K = (U_1, \ldots, U_K)$ that, e.g., exert effects on treatments and outcomes.
 
 In the special case of a sequentially randomized trial, $\overline{U}_k$ will be empty as treatments $A_k$ are randomized, conditional at most on the measured past, $\overline{L}_k$.  

We define a regime $g$ as a sequence of decision rules $(g_1, \ldots, g_K)$ taken at times $1, \ldots, K$. Let superscripts denote counterfactuals. In particular, $(\overline{A}_{k-1}^{\overline{a}_{k-2}}, \overline{R}_{k-1}^{\overline{a}_{k-1}}, \overline{L}_{k}^{\overline{a}_{k-1}})$ are the counterfactual values of $(\overline{A}_{k-1}, \overline{R}_{k-1}, \overline{L}_{k})$ when implementing the regime $(g_1, \ldots, g_{k-1}) = (a_1, \ldots, a_{k-1})$. Furthermore, we use a ``$+$" to denote treatment values that are assigned by the regime. Thus, $A_k^{g+}$ is the value of treatment at time $k$ assigned under the regime $g$. We define the history for times $k \geq 2$ as $H_k^{\Bar{a}_{k-1}} = (\overline{R}_{k-1}^{\overline{a}_{k-1}}, \overline{L}_{k}^{\overline{a}_{k-1}})$, taking values in $\mathcal{H}_k$, and $H_1 = L_1$.

Like \citet{stensrud_optimal_2024, laurendeau2024improved}, we aim to leverage natural treatment values to improve $H$-optimal regimes. 
To be more precise, we give a definition of natural treatment values in time-varying settings, following \citet{richardson2013single, sarvet2025natural}:
\begin{definition}[Natural treatment value]
    The natural treatment value at time $k$ is the value the treatment would take in absence of it being assigned by an intervention $A_k^g = A_k^{\overline{A}_{k-1}^{g+}} = A_k^{\overline{g}_{k-1}}$.
\end{definition}
See Table \ref{tab:notation} in the Appendix for a summary of the terminology introduced in this and the following sections.

\subsection{Conventional regimes and their relation to initiation regimes}
Throughout this article, we will refer to different types of dynamic regimes that intervene on treatments $A_k$ for $k =1,\dots,K$. In particular, it will be useful to express the treatment regime that is implemented naturally in the absence of interventions as a particular dynamic regime:
\begin{definition}[Observed regime]
    The observed regime $g^{\textbf{obs}}$ is the regime that intervenes on treatments $A_k$ only, and assigns $A_k^{g^{\textbf{obs}}+}=A_k^{g^{\textbf{obs}}}$, at every time $k = 1, \ldots, K$, with probability 1.
    \label{def: obs regime}
\end{definition}
For Definition \ref{def: obs regime} to be linked to the observed data, we need a standard consistency assumption:
\begin{assumption}[Consistency]
    For all $k = 1, \ldots, K$, for all $\overline{a}_{k}$,
    \begin{align*}
        (R_k^{\overline{a}_{k}}, A_k^{\overline{a}_{k-1}}, L_k^{\overline{a}_{k - 1}}) = (R_k, A_k, L_k) \mid  \overline{A}_{k} = \overline{a}_{k},
    \end{align*}
    where $\overline{a}_{k-1}$ is assumed to be empty when $k = 1$.
    \label{ass: consistency}
\end{assumption}

By Assumption \ref{ass: consistency}, $g^{\textbf{obs}}$ assigns the observed treatment $A_k$ at every time $k =1,\dots,K$. The observed regime can be described as the ``current standard of care" \citep{stensrud_optimal_2024,laurendeau2024improved}, or ``status-quo regime" \citep{levis2024intervention},\footnote{This definition of ``current standard of care" is (intentionally) different from baseline regimes that sometimes have been used in the literature, like the ``never-treat" regime where $g_k = 0$ for all $k$ \citep{luedtke2016optimal, cui2021individualized, qiu2022individualized} or an arbitrary baseline regime that is a function of observed covariates only \citep{kallus_confounding_2018}. However, our results on optimality also apply to these other baseline regimes, see Appendix \ref{app: other baseline regimes}.} as it corresponds to decisions made by current decision makers in observed data. 

We will contrast our proposed regimes with the conventional optimal regime based on covariate history, hereby the \emph{$H$-optimal regime}, where ``$H$" stands for ``history". For a regime $g = (g_1, \ldots, g_K)$, let $\overline{g}_k = (g_1, \ldots, g_k)$, let $\underline{g}_k = (g_k, \ldots, g_K)$, let $\mathcal{G}$ be the set of regimes such that for each $k$, $g_k$ is a function mapping $\mathcal{H}_k$ to $\{0,1\}$, and let $R$ denote the utility function; we will consider $R = \sum_{k = 1}^K R_k$, such that, under a regime $g$, we have $R^g = \sum_{k = 1}^K R_k^g$. The outcome $R$ can be an arbitrary real function of $R_1, \ldots, R_K$, it does not need to be the sum \citep{bertsekas1996neuro,murphy2003optimal}. 

Then, the $H$-optimal regime, as often described in the literature \citep{murphy2003optimal}, is defined as follows:
\begin{definition}[$H$-optimal regime]
  $  
        g^{\textbf{opt}} := \argmax_{g \in \mathcal{G}} \mathbb{E}[R^g].
 $
    \label{def: optimal regime}
\end{definition}

Let $\mathcal{V}_{P,k}^{g}(H_{k}^g)$ be the expected cumulative reward of implementing regime $g$ given history $H_{k}^g$, 
\begin{align*}
    \mathcal{V}_{P,k}^{g}(H_{k}^g) := \mathbb{E}[\sum_{j = k }^K R_j^{\overline{g}_j(H_j^g)} | H_{k}^g],
\end{align*}
which only depends on $g_{k}, \ldots, g_K$ conditional on $H_{k}^g = H_{k}^{g_1(H_1), \ldots, g_{k-1}(H_{k-1}^{g})}$. We index $\mathcal{V}_{P,k}^{g}$ by the true law $P$ to emphasize that the expected cumulative reward changes with $P$. Using $\mathcal{V}_{P,k}^{g}(H_{k}^g)$, we can express the $H$-optimal regime $g^{\textbf{opt}}$ as the solution of Bellman equations:
\begin{theorem}[\citet{bellman1956dynamic, murphy2003optimal}]
The $H$-optimal regime verifies

   \begin{align}
    g_K^{\textbf{opt}}(h_K) &= \argmax_{a_K'\in \{0,1\}} \mathbb{E}[R_K^{\overline{a}_{K-1}, a_K'} | H_K^{\overline{a}_{K-1}} = h_K],\label{eq: Bellman opt 1}\\
    g^{\textbf{opt}}_k(h_k) &= \argmax_{a_k'\in \{0,1\}} \mathbb{E}[R_k^{\overline{a}_{k-1}, a_k'} + \mathcal{V}_{P,k+1}^{ \overline{a}_{k-1}, a_k', \underline{g}^{\textbf{opt}}_{k+1}}(H_{k+1}^{\overline{a}_{k-1}, a_k'}) | H_k^{\overline{a}_{k-1}} = h_k], \label{eq: Bellman opt 2}
\end{align} for $k = 1, \ldots, K-1$, and with $\overline{a}_{K-1} = \overline{g}_{K-1}(h_{K-1})$.


\label{thm: opt Bellman eqs}
\end{theorem}

Theorem \ref{thm: opt Bellman eqs} motivates a dynamic programming algorithm to find $g^{\textbf{opt}}$ when value functions $\mathbb{E}[R_k^{\overline{a}_k} | H_k^{\overline{a}_{k-1}}]$ are identified. Such an approach to computing the $H$-optimal regime is often referred to as Q-learning \citep{sutton1999reinforcement}. Modified versions of the Bellman equations in Theorem \ref{thm: opt Bellman eqs} can be used to define optimal regimes in a larger class of regimes than in Theorem \ref{thm: opt Bellman eqs}, as we will see in Section \ref{sec: initiation regimes}.


\subsubsection{Initiation regimes}
\label{sec: initiation regimes}
To introduce a richer class of regimes that leverage natural treatment values and include $g^{\textbf{obs}}$ and $g^{\textbf{opt}}$, we first define $H$-optimal regimes conditional on an observed past $(h_j, \overline{a}_{j-1}')$,
\begin{align*}
    g^{\textbf{opt}_j}:= \argmax_{g \in \mathcal{G}_j} \mathbb{E}[ R^g],
\end{align*}
where 
\begin{align*}
    \mathcal{G}_j:= \{ \underline{g}_{j} = (g_{j}, \ldots, g_K) \text{ s.t. } g_k: \mathcal{H}_k \times \{0,1\}^{j-1} \to \{0,1\} \text{ for $k \geq j$}\}.
\end{align*}

The joint distribution of $\overline{L}_j, \overline{A}_{j-1},$ and $\overline{R}_{j-1}$ is identified from observational (non-interventional) data for all $j$. Because, under Assumption \ref{ass: consistency}, $\overline{A}_{j-1}^{g+} = \overline{A}_{j-1}$ for $g \in \mathcal{G}_j$, we have that ${H}_j^g = {H}_j$ and $\overline{A}_{j-1}^g = \overline{A}_{j-1}$. Thus, using the definitions of the $g^{\textbf{opt}_j}$, we can define regimes that align with the observed (factual) regime until a set time $j$, and thereafter aligns with the $H$-optimal regime:
\begin{definition}[$j$th initiation regime]
    The $j$th initiation regime $g^{\textbf{obs}_j}$ is identical to the observed regime through time $j-1$, where $j = 1, \ldots, K+1$, and then switches to an $H$-optimal regime that additionally uses patients' observed treatments through $j-1$,
\begin{align*}
    g^{\textbf{obs}_j}_k( H_k^{g^{\textbf{obs}_j}}, \overline{A}^{g^{\textbf{obs}_j}}_{j-1}) := \begin{cases}
        A^{g^{\textbf{obs}_j}}_k & \text{ if $k < j$}\\
        g^{\textbf{opt}_j}_k(H_k^{g^{\textbf{obs}_j}}, \overline{A}_{j-1}^{g^{\textbf{obs}_j}}) & \text{if $k \geq j$}.
    \end{cases}
\end{align*}
\label{def: initiation regimes}
\end{definition}
In particular, the first initiation regime, $g^{\textbf{obs}_1}$, is equal to the $H$-optimal regime, $g^{\textbf{opt}}$, and the last initiation regime, $g^{\textbf{obs}_{K + 1}}$, is equal to the observed regime $g^{\textbf{obs}}$. 

Consider the optimal $j$th initiation regime, that is, $\argmax_{ j = 1, \ldots, K + 1} \mathbb{E}[R^{g^{\textbf{obs}_j}}]$. As the $j$th initiation regimes shift the time-point from which we start the $H$-optimal regime from time $1$ to time $j$, we call this regime the \emph{optimal shifting regime}, $g^{\textbf{osh}}$, see Appendix \ref{app: osh regime} for more details on properties of these regimes.

However, we will focus on a regime that also outperforms the optimal shifting regime; to further leverage the natural treatment values before initiating the $H$-optimal regime, we consider
\begin{align*}
    g^{\textbf{sup}_j}:= \argmax_{g \in \mathcal{G}^{\textbf{sup}}_j} \mathbb{E}[ R^g],
\end{align*}
where for each $k$, $g_k := g_k(H_k^g, \overline{A}^g_j)$, and 
where 
\begin{align*}
    \mathcal{G}_j^{\textbf{sup}}:= \{ \underline{g}_{j} = (g_{j}, \ldots, g_K) \text{ s.t. } g_k: \mathcal{H}_k \times \{0,1\}^{j} \to \{0,1\} \text{ for $k \geq j$}\},
\end{align*}
which includes $\mathcal{G}_j$.

This motivates the $j$-th superoptimal initiation regime:
\begin{definition}[$j$th superoptimal initiation regime]
    The $j$th initiation regime $g^{\textbf{sup-obs}_j}$ is identical to the observed regime through time $j-1$, where $j = 1, \ldots, K+1$, and then switches to the $H$-optimal regime,
\begin{align*}
    g^{\textbf{sup-obs}_j}_k( H_k^g, \overline{A}^{g^{\textbf{sup-obs}_j}}_j) = \begin{cases}
        A^{g^{\textbf{sup-obs}_j}}_k & \text{ if $k < j$}\\
        g^{\textbf{sup}_j}_k(H_k^{g^{\textbf{sup-obs}_j}}, \overline{A}_{j}^{g^{\textbf{sup-obs}_j}}) & \text{if $k \geq j$}.
    \end{cases}
\end{align*}
\label{def: superopt initiation regimes}
\end{definition}
Definition \ref{def: superopt initiation regimes} is different from Definition \ref{def: initiation regimes} as $g^{\textbf{sup}_j}$ is a function of the natural treatment at time $j$, $A^{g^{\textbf{sup-obs}_j}}_j$, but $g^{\textbf{opt}_j}$ is not. 

Although the $j^{th}$ superoptimal initiation regime can be equal to the observed or the optimal regime given the past observed history, there are data generating mechanisms for which this superoptimal regime at time-point $j$ almost surely equals $1-A_j$, the opposite of the observed treatment, see \citet[Proposition 3]{stensrud_optimal_2024}. Conceptually, the optimal initiation regime, $g^{\textbf{init}}$ can be understood as a strategy that initially adheres to the observed regime, deviates from it once, at some time-point $T^{g^{\textbf{init}}}$,
and subsequently aligns with the $H$-optimal regime, conditional on the observed past,

\begin{definition}[Optimal initiation regime]
    Let 
    \begin{align*}
        \mathcal{G}^{\textbf{init}} := \{g = (g_1, \ldots, g_K) &\text{ s.t. } g_k : \mathcal{H}_k \times \{0,1\}^k \to \{0,1\}, \\
&g_k(H_k^g, \overline{A}_k^g) = I(\overline{A}_{k - 1}^{g+} = \overline{A}_{k-1}^g) g_k^k(H_k^g, \overline{A}_k^g)  \\
       & + \sum\limits_{j=1}^{k-1}I(\overline{A}_{j-1}^{g+} = \overline{A}_{j-1}^g)I({A}_{j}^{g+} \neq {A}_{j}^g) g_k^{j}(H_k^g, \overline{A}_{j}^g)\\
       &\text{for some }  g_k^j : \mathcal{H}_k \times \{0,1\}^j \to \{0,1\}, j\leq k, k = 1,\ldots, K\}. 
    \end{align*}

   The optimal initiation regime $g^{\textbf{init}}$ is given by
    \begin{align*}
        g^{\textbf{init}} &:= \argmax_{g \in \mathcal{G}^{\textbf{init}}} \mathbb{E}[R^{g}].
    \end{align*}
\end{definition}

\begin{remark}
    For $g \in \mathcal{G}^{\textbf{init}}$, let
    \begin{align*}
        T^{g} := \min\left(K + 1, \min\{k = 1, \ldots, K : A_k^{g+} \neq A_k^{g}\}\right)
    \end{align*}
    be the time at which a unit deviates from the observed regime and aligns with the $H$-optimal regime conditional on the observed past. The definition of $\mathcal{G}^{\textbf{init}}$ implies that $T^{g}$ is a random variable that depends on each unit's history and responses to $g$; in particular, units can deviate from the observed regime at different time points. This follows from the observation that the first time $j$ at which $A_j^{g+} \neq A_j^g$ depends on $g_1(H_1, A_1), \ldots, g_j(H_j^g, \overline{A}_j^g)$, which varies with $H_j^g$ and $\overline{A}_j^g$. Let $\mathcal{T}$ be the set of $\{1, \ldots, K+1\}$-valued stopping times with respect to the filtration $\mathcal{F}_k := \sigma(H_k^g, \overline{A}_k^g)$, $k = 1, \ldots, K$, that is, $T \in \mathcal{T}$ if the event $\{T = k\}$ is determined by $(H_k^g, \overline{A}_k^g)$ for every $k \leq K$. By the above, $T^g \in \mathcal{T}$ for every $g \in \mathcal{G}^{\textbf{init}}$, and, for $T \in \mathcal{T}$, $g^{\textbf{sup-obs}_T}$ denotes the regime of Definition
\ref{def: superopt initiation regimes} with $j$ replaced by $T$ pathwise.
    \label{rem: adaptive switch time}
\end{remark}

The optimal initiation regime performs at least as well as the $H$-optimal and observed regimes as it optimizes over a larger class of regimes; that is, $\mathcal{G}^{\textbf{init}}$ contains the observed regime $g^{\textbf{obs}}$ and the $H$-optimal regime $g^{\textbf{opt}}$:

\begin{proposition}
    The optimal initiation regime performs at least as well as the $H$-optimal and observed regimes,
    \begin{align*}
        \mathbb{E}[R^{g^{\textbf{init}}}] \geq \max\left(\mathbb{E}[R^{g^{\textbf{obs}}}], \mathbb{E}[R^{g^{\textbf{opt}}}]\right),
    \end{align*}
    and there exist distributions for which the inequality is strict.
    \label{prop: ssw > max(obs,opt)}
\end{proposition}

Proposition \ref{prop: ssw > max(obs,opt)} follows from the fact that the optimal initiation regime $g^{\textbf{init}}$ uses information on unmeasured confounders $U_1, \ldots, U_K$ from $(A_1, \ldots, A_{T^{g^{\textbf{init}}}})$ that is ignored by the $H$-optimal regime $g^{\textbf{opt}}$, which is only a function of measured confounders $L_1, \ldots, L_k$. In particular, if $\mathbb{E}[R^{g^{\textbf{init}}}] > \mathbb{E}[R^{g^{\textbf{opt}}}]$, we can deduce that there is unmeasured confounding, see Section \ref{sec: DGPs for initiation regimes} for additional discussion.


 Similar to $g^{\textbf{opt}}$, we can use modified Bellman equations to compute $g^{\textbf{init}}$, see Algorithm \ref{algo: superoptimal switching} in Appendix \ref{app: Bellman eqs}. The Bellman equations allow for a search over a small space of value functions in an efficient manner. 




\subsection{Longitudinal superoptimal regimes}
Proposition \ref{prop: ssw > max(obs,opt)} implies that $g^{\textbf{init}}$ outperforms $g^{\textbf{obs}}$ and $g^{\textbf{opt}}$ on average. However, there exist larger classes of regimes that leverage information on natural treatment values at every time. Consider
\begin{align*}
    \mathcal{G}^{\textbf{sup}} := &\{g = (g_1, \ldots, g_K) \text{ s.t. } g_k: \mathcal{H}_k \times \{0,1\}^k \to \{0,1\}, \\
    &\text{ where $g_k$ is a function of $H_k$ and $\overline{A}_k^g$ for every $k = 1, \ldots, K$} \}.
\end{align*}

Then, we can define the time-varying superoptimal regime as follows:
\begin{definition}[Superoptimal regime]
    The superoptimal regime $g^{\textbf{sup}}$ is given by 
    \begin{align*}
        g^{\textbf{sup}} := \argmax_{g \in \mathcal{G}^{\textbf{sup}}} \mathbb{E}[R^g].
    \end{align*}
    \label{def: superoptimal regime}
\end{definition}

The next proposition then follows immediately from $\mathcal{G}^{\textbf{init}} \subseteq \mathcal{G}^{\textbf{sup}}$.
\begin{proposition}
    The superoptimal regime performs at least as well as the initiation regime,
    \begin{align*}
        \mathbb{E}[R^{g^{\textbf{sup}}}] \geq \mathbb{E}[R^{g^{\textbf{init}}}].
    \end{align*}
    \label{prop: sup > init}
\end{proposition}
The following corollary follows immediately from Proposition \ref{prop: sup > init} and Proposition \ref{prop: ssw > max(obs,opt)}.
\begin{corollary}
    The optimal superoptimal regime performs at least as well the $H$-optimal and observed regimes,
    \begin{align*}
        \mathbb{E}[R^{g^{\textbf{sup}}}] \geq \max\left(\mathbb{E}[R^{g^{\textbf{obs}}}], \mathbb{E}[R^{g^{\textbf{opt}}}]\right).
    \end{align*}
\end{corollary}

Although Proposition \ref{prop: sup > init} clarifies that the superoptimal regime $g^{\textbf{sup}}$ outperforms the initiation regime $g^{\textbf{init}}$, the regimes in the class $\mathcal{G}^{\textbf{sup}} \setminus \mathcal{G}^{\textbf{init}}$ depend on counterfactuals $A_k^g$, which often complicates identification, as we discuss in Section \ref{sec: init vs superopt}.

\section{Identification}
\label{sec: Identification}

Identifying, estimating, and implementing $g^{\textbf{init}}$ using data requires consideration of natural treatment values. Many currently used methods rely on sequentially randomized trial data where, unlike observational data, natural treatment values are not recorded. Here we will present sufficient conditions for identification of the regimes introduced in Section \ref{sec: Data structure} from different types of data structures. In particular, we will describe how the regimes can be identified from unconventional experiments that, in principle, can be conducted, aligning with an interventionist, ``single-world" causal inference paradigm  \citep{robins1986new,richardson2013single}. One such experiment assigns units to follow the natural treatment value regime until the unit is randomly selected to enter a conventional randomized trial; an experimental design tailored to identifying the initiation regimes that similarly follow natural treatment value regimes until initiation of a conventional $H$-optimal regime at a time $j$. Another experiment randomly assigns units to treatment, control, or a natural treatment value regime at each time point, and is particularly well suited to identifying the superoptimal regime. We describe these experimental designs and discuss identification properties in Section \ref{sec: data fusion} and Appendix \ref{app: larger trials}. 










In the remainder of the article, we will assume that interventions on treatments $A_k$ are well-defined, such that Assumption \ref{ass: consistency} holds.


We also invoke the usual positivity assumption:

\begin{assumption}[Positivity]
    For all $k = 1, \ldots, K$, $$P(A_k = 1 | H_k, \overline{A}_{k - 1}) \in (0,1)$$ with probability one.
    \label{ass: positivity}
\end{assumption}

At this stage, we do not impose any sequential exchangeability assumptions \citep{richardson2013single}, which hold in sequentially randomized trials, like SMARTs, by design, but can fail in observational data when there is unmeasured confounding. 

To simplify the notation, we define four functionals for $j = 1,\ldots,K$
\begin{align*}
    \gamma_j(\overline{a}_K, \overline{l}_j) &:= \mathbb{E}[R^{a_j,a_{j + 1}, \ldots, a_K} | A_j = 1-a_j, \overline{A}_{j-1} = \overline{a}_{j-1}, \overline{L}_j = \overline{l}_j]\\
    \omega_j(\overline{a}_K, \overline{l}_j) &:= \mathbb{E}[R^{a_j, \ldots, a_K }| \overline{A}_{j-1} = \overline{a}_{j-1}, \overline{L}_j = \overline{l}_j]\\
    \eta_j(\overline{a}_K, \overline{l}_j) &:= \mathbb{E}[R^{a_{j + 1}, \ldots, a_K}| A_j = a_j, \overline{A}_{j-1} = \overline{a}_{j-1}, \overline{L}_j = \overline{l}_j]\\
    \pi_j (\overline{a}_j, \overline{l}_j) &:= P(A_j = a_j | \overline{A}_{j-1} = \overline{a}_{j-1}, \overline{L}_j = \overline{l}_j).
\end{align*}

Then, the following lemma will be convenient when expressing identification results for $g^{\textbf{init}}$:
\begin{lemma}
    Under Assumptions \ref{ass: consistency} and \ref{ass: positivity}, for $j = 1, \ldots, K$,
    \begin{align*}
        \gamma_j(\overline{a}_K, \overline{l}_j) = \frac{\omega_j(\overline{a}_K, \overline{l}_j) - \eta_j(\overline{a}_K, \overline{l}_j)\pi_j(\overline{a}_j, \overline{l}_j)}{1-\pi_j(\overline{a}_j, \overline{l}_j)}.
    \end{align*}
    \label{lemma: lemma 1 analog}
\end{lemma}
In the special case of a point-treatment setting, Lemma 1 is identical to a formula for the superoptimal regime studied in \citet{stensrud_optimal_2024}, and similar to formulas that are well-known in works on treatment effects on the treated \citep{robins2006comment, geneletti2011defining, shpitser2012effects, bareinboim2015bandits, dawid2021can}:

\begin{corollary}
    When $K = 1$, Lemma \ref{lemma: lemma 1 analog} reduces to
    \begin{align*}
        \mathbb{E}[R^{a_1} | A_1 = 1-a_1, L_1 = l_1] = \frac{\mathbb{E}[R^{a_1} | L_1 = l_1] - \mathbb{E}[R | A_1 = a_1, L_1 = l_1]P(A_1 = a_1 | L_1 = l_1)}{P(A_1 = 1-a_1 | L_1 = l_1)}.
    \end{align*}
    \label{cor: Lemma 1 implies lemma 1}
\end{corollary}
In Appendix \ref{app: Bellman eqs} we describe how $g^{\textbf{init}}_j$ is identified via Algorithm \ref{algo: superoptimal switching} in terms of $\gamma_j$ and $\eta_j$. 
Lemma \ref{lemma: lemma 1 analog} shows how classical designs, and observed data structures, can be insufficient to identify initiation regimes. In particular, conventional sequentially randomized trials only identify $\mathbb{E}[R^{a_1, \ldots, a_K} | H_k^{a_1, \ldots, a_{k-1}}]$ for $k = 1, \ldots, K$, whereas observational studies only identify $\mathbb{E}[R | \overline{A}_j, H_k]$ and $P(A_k = 1 | H_k, \overline{A}_j)$ for $j,k = 1, \ldots, K$ without additional identification assumptions. One remedy is to consider trials with staggered entries \citep{pollock1989survival,sun2021estimating, athey2022design,xiong_optimal_2024, wing2024designing}. These trials enroll individuals at different times, and subsequently randomize them to different treatments.

If treatments, covariates, and outcomes prior to trial enrollment are measured, staggered entry trials identify $\mathbb{E}[R^{a_{j+1}, \ldots, a_K} | H_k, \overline{A}_j = \overline{a}_j] = \eta_j(\overline{a}_K, \overline{l}_j)$ for some $j < k$ and $\overline{a}_j$. Hence, if staggered entry trials and available observational data include all possible sequential treatment combinations with positive probability such that Assumption \ref{ass: positivity} is verified, $g^{\textbf{init}}$ is identified. Furthermore, in the so-called nested trial design, where a subset of individuals in an observational cohort study are recruited to a randomized experiment \citep{dahabreh2019generalizing,dahabreh2019generalizing_biometrics}, $g^{\textbf{init}}$ can be identified without additional identification assumptions. Simpler trial designs that aim to identify $g^{\textbf{init}}$ can be deduced from these observations, which we describe in Section \ref{sec: data fusion}. In brief, as the $\pi_j$ are observed, Lemma \ref{lemma: lemma 1 analog} identifies $\gamma_j$ from observational studies when $\omega_j$ and $\eta_j$ are identified. Conditional on past natural treatment values and covariates, $\omega_j$ and $\eta_j$ are expectations of $R$ under longitudinal static regimes that can be identified with additional assumptions such as assumptions used for instrumental variable \citep{chen_estimating_2023} or proximal learning \citep{ying_proximal_2023} identification, or by staggered entry trials, which we will describe in Section \ref{sec: data fusion}.

\subsection{A note on initiation regimes versus superoptimal regimes}
\label{sec: init vs superopt}
As Proposition \ref{prop: sup > init} clarifies, the superoptimal regime $g^{\textbf{sup}}$ outperforms $g^{\textbf{init}}$. However, for $g \in\mathcal{G}^{\textbf{sup}} \setminus \mathcal{G}^{\textbf{init}}$, $g$ will depend on unobserved $A_k^g$ for some $k$. This implies that the superoptimal regime $g^{\textbf{sup}}$ is not necessarily identified under practically relevant data structures that identify  $g^{\textbf{init}}$; conceptually, identifying the distribution of natural treatment values after an intervention is not necessary for identification of $g^{\textbf{init}}$, analogously to the optimal regime $g^{\textbf{opt}}$. 

We explicitly describe relevant data structures in Section \ref{sec: data fusion} and Appendix \ref{app: larger trials}, respectively. In particular, we present an instrumental variable setting where the same set of assumptions can be used to identify $g^{\textbf{opt}}$ and $g^{\textbf{init}}$, but stronger assumptions are required for $g^{\textbf{sup}}$, see Appendices \ref{sec: IV identification for obs data} and \ref{app: sequential IV}.

\subsubsection{Practical motivation for $g^{\textbf{init}}$}
\label{sec: DGPs for initiation regimes}
   Suppose that early after treatment initiation, physicians make decisions based on information that is unavailable in the recorded history $H_k$. This information includes subtle symptoms, clinical appearance, frailty, patient preferences, adherence, early toxicity, informal reports from patients, or other features that are observed by the physician but not captured in the database \citep{hamerman1999toward, verghese2018computer,matheny2019artificial}. If these factors predict both treatment decisions and subsequent outcomes, then an $H$-optimal regime can perform poorly when it replaces observed treatment decisions during this early period.

   Later in follow-up, response to treatment, toxicity, adherence and biomarker trajectories are recorded, so $H_k$ captures more of the information the physician uses and the $H$-optimal regime can improve on observed care. An initiation regime that follows observed care early and switches to the $H$-optimal regime once sufficient additional information is added to $H_k$ therefore improves on both the observed and conventional $H$-optimal regimes.

We emphasize that the variables $(A_1,\ldots,A_K)$ can  correspond to different clinical decisions, such as initiating first-line treatment, switching treatment, adding rescue therapy, or starting maintenance therapy, similarly to the different treatment options at the first and second time points in the motivating example in Figure \ref{fig: Batorsky trial illustration}. 

Regimes with the structure described here appear in practice. \citet{poller1998multicentre} randomized patients to anticoagulant treatment dosing by physicians or by a computer algorithm. For $79\%$ of patients in the algorithmic arm, dosage was chosen by physicians during the first three weeks, until the computer arm had accumulated sufficient history to be trusted by the clinicians, so the treatment regime of patients in the computer arm resembled an initiation regime. The algorithmic arm better targeted International Normalized Ratios (INRs), the outcome of interest for the study, than the control arm with standard care. In childhood leukemia, a protocol-based induction chemotherapy phase, where physicians decide on supportive care and dose modifications using information that is largely unrecorded, is the standard way of assigning treatment initially. Minimal Residual Disease (MRD), which can only be measured at the end of induction, is the strongest recorded predictor of relapse \citep{borowitz2008clinical} and post-remission therapy is therefore assigned by classifying patients into MRD risk groups \citep{vora2013treatment,vora2014augmented}. This sequence of two treatment phases has the form of an initiation regime with $T^g$ at the end of induction, when MRD is measured. The estimated switch time $T^{g^{\textbf{init}}}$ can be used to check whether the usefulness of clinical expertise indeed decreases after MRD is measured.  Another motivation is \citet{guo2015measurement}, who compared algorithmic and physician treatments of major depression and found higher remission under the algorithmic regime, but did not consider regimes that combine the two.


\subsection{Study designs}
\label{sec: data fusion}



We consider two data structures corresponding to data from different types of study designs. We illustrate the data structures using trees with branches representing treatment assignments. In particular, each leaf corresponds to an arm of the trial; hence the total number of leaves represents the complexity of the trial. 

Many of these designs are non-standard but have nevertheless been implemented in practice. We focus on them here to clarify the interpretation of $g^{\textbf{init}}$ and its identification conditions, and also to inspire future trial designs.

    \subsubsection{Observational data + Conventional Sequentially Randomized Trial} 
    Suppose we have data from both a sequentially randomized experiment and an observational study drawn from the same superpopulation. This is an example of a ``data fusion" setting \citep{zhang2019near, batorsky2024integrating, joshi2024towards}, which we call Data Structure \ref{algo: obs + trial }, represented in Figure \ref{fig: obs + trial}. In particular, such data can be obtained from certain types of Patient Preference Trials (PPTs) \citep{rucker1989two,long2008causal,knox2019design}, which randomly assign patients at baseline to follow their preferences -- analogously to an observational study -- or to be assigned treatments through randomization.\footnote{If patients are randomized at each time point to follow their preference or enter a sequentially randomized trial, as in staggered entry trials, this corresponds to Data Structure \ref{algo: randomized trial obs} in Section \ref{sec: sequ entry rand trial}. If patients are randomized at each time point to either follow their preference or receive treatments this corresponds to Data Structure \ref{algo: three choice trial } in Appendix \ref{sec: sequ 3-arm trial}.} Although data following Data Structure \ref{algo: obs + trial } are sometimes readily available, the observational and sequentially randomized data need to satisfy some assumptions to be used in combination \citep{graham2026towards}. In particular, the observational and trial data could be conceived as draws from the same superpopulation, where covariates $L_k$ and rewards $R_k$ are recorded. If observational data and trial data originate from different populations, appropriate generalization methods need to be used to identify the value functions of $g^{\textbf{opt}}$ on the target population \citep{rothwell2005external, bareinboim2016causal,brantner2023methods,shi2023data}.
  \begin{algorithm}
    \spacingset{1}
    \SetAlgoRefName{1}
    \SetAlgorithmName{Data Structure}
\;
    
    \For{$k = 1,\ldots, K$}{
    
    \For{each individual in trial arm}{
    Record covariate value $L_k^g$\;
    Randomize individual to receive treatment $A_k^{g+} = 1$ with fixed and known probability $p(H_k^g) \in (0,1)$\;
    Record reward $R_k^g$\;
    }
    \For{each individual in observational arm}{
    Record $(L_k, A_k, R_k)$\;

    }
    }
    \caption{Observational data + Trial.}
    \label{algo: obs + trial }
    \end{algorithm}
    
    \begin{proposition}
         Under Assumptions \ref{ass: consistency} and \ref{ass: positivity}, Data Structure \ref{algo: obs + trial } identifies the $H$-optimal regime $g^{\textbf{opt}}$ and the observed regime $g^{\textbf{obs}}$, but not the optimal initiation regime $g^{\textbf{init}}$ or the superoptimal regime $g^{\textbf{sup}}$. 
        \label{prop: usual data fusion impossibility}
    \end{proposition}
    The proof of Proposition \ref{prop: usual data fusion impossibility} is in Appendix \ref{app: proofs}. The identification formula for $g^{\textbf{opt}}$ based on the Bellman equations was described by \citet{murphy2003optimal}:
    \begin{align}
        g_K^{\textbf{opt}}(h_K) &= \argmax_{a_K \in \{0,1\}} \mathbb{E}[R_K | H_K = h_K, \overline{A}_{K-1}^{g^{\textbf{opt}}+} = \overline{g}^{\textbf{opt}}_{K-1}(h_{K-1}), A_K^{g^{\textbf{opt}}+} = a_K] \label{eq: id Bellman opt 1}\\
        g_k^{\textbf{opt}}(h_k) &= \argmax_{a_k \in \{0,1\}} \mathbb{E}[R_k + \mathcal{V}^{\underline{g}_{k + 1}^{\textbf{opt}}}_{k+1}(H_{k+1}) | H_k = h_k, \overline{A}_{k-1}^{g^{\textbf{opt}}+} = \overline{g}^{\textbf{opt}}_{k-1}(h_{k-1}), A_k^{g^{\textbf{opt}}+} = a_k], \label{eq: id Bellman opt 2}
    \end{align}
    where $\overline{A}_K^{g^{\textbf{opt}}+}$ is observed in the trial arm and hence identified and $\mathcal{V}^{\underline{g}_{k + 1}^{\textbf{opt}}}_k(H_{k+1})$ is identified recursively by the trial data alone.

        Data Structure \ref{algo: obs + trial } only requires one additional arm compared to the conventional sequentially randomized trial. This additional arm corresponds to the observed regime, see Figure \ref{fig: obs + trial}. Data Structure \ref{algo: obs + trial } is the design used to generate data in the running example by \citet{batorsky2024integrating}, see Section \ref{sec: chronic back pain ex intro}, and is not sufficient to identify the optimal initiation regime by Proposition \ref{prop: usual data fusion impossibility}. However, next we show that these additional regimes can be identified with some simple modifications to the trial design.

        Under additional assumptions about the data generating process, for example instrumental variable assumptions, we can identify $g^{\textbf{init}}$ and even $g^{\textbf{sup}}$ from observational data only, as we describe explicitly in Appendices \ref{sec: IV identification for obs data} and \ref{app: sequential IV}. More broadly, in Proposition \ref{prop: initiation regime id} of Section \ref{app: general estimation result}, we explain that given identification of a set of counterfactual quantities, then $g^{\textbf{init}}$ is identified. This motivates the use of proximal inference \citep{ying_proximal_2023} or frontdoor identification \citep{pearl1993mediating} in potential future work.
    
    \begin{figure}
        \centering
        \begin{tikzpicture}[scale = .9]
            \node (start) at (0,0) {Trial entry};
            \node (A1=1) at (2,-2) {};
            \node (3DL) at (-3,-3.5) {$\cdots$};
            \node (3DR) at (3,-3.5) {$\cdots$};
            \node (3DML) at (-1,-3.5) {$\cdots$};
            \node (3DMR) at (1,-3.5) {$\cdots$};

            \draw[-] (start) -- node[rotate = 45, above, scale = .85] {$A_1^{g+} = 0$} (-2,-2);
            \draw[-] (start) -- node[rotate = 315, above, scale = .85] {$A_1^{g+} = 1$} (2,-2);
            \draw[-] (-2,-2) -- node[rotate = 55, above, scale = .85]{$A_2^{g+} = 0$} (3DL);
            \draw[-] (-2,-2) -- node[rotate = 305, above, scale = .85] {$A_2^{g+} = 1$} (3DML);
            \draw[-] (2,-2) -- node[rotate = 55, above, scale = .85]{$A_2^{g+} = 0$} (3DMR);
            \draw[-] (2,-2) -- node[rotate = 305, above, scale = .85] {$A_2^{g+} = 1$} (3DR);

            \node (start obs) at (10,0) {Observation start};
            \node (A1) at (10,-2) {Record $A_1$};
            \node (A2) at (10, -3.5) {$\cdots$};

            \draw[-] (start obs) -- (A1);
            \draw[-] (A1) -- (A2);
        \end{tikzpicture}
        \caption{Tree describing Data Structure \ref{algo: obs + trial }.}
        \label{fig: obs + trial}
    \end{figure}
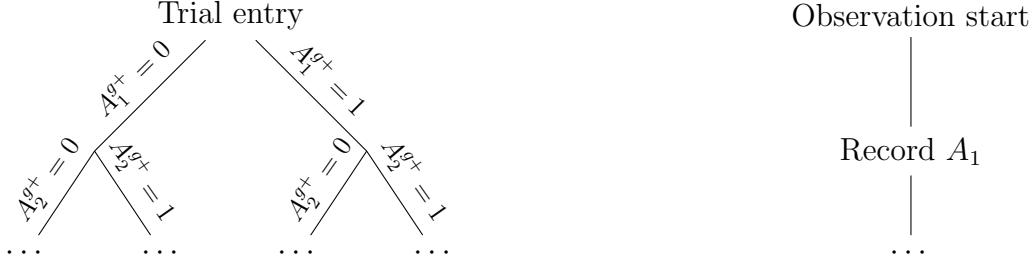
    
    \subsubsection{Sequential entry randomized trial} 
    \label{sec: sequ entry rand trial}
   Consider a second design, represented in Figure \ref{fig: randomized trial obs + recorded nat val} as  Data Structure \ref{algo: randomized trial obs}, that randomizes units to enter a sequentially randomized trial or follow their natural treatment value. The following proposition follows immediately from the properties of Data Structure \ref{algo: randomized trial obs} and the definition of $g^{\textbf{init}}$.
    \begin{proposition}
    Under Assumptions \ref{ass: consistency} and \ref{ass: positivity}, data generated by Data Structure \ref{algo: randomized trial obs} can be used to identify $g^{\textbf{obs}}$, $g^{\textbf{opt}}$, $g^{\textbf{init}}$, but not $g^{\textbf{sup}}$.
    \label{prop: randomized trial obs id}
    \end{proposition}
    A proof of Proposition \ref{prop: randomized trial obs id} can be found in Appendix \ref{app: proofs}.
    Under Data Structure \ref{algo: randomized trial obs}, $g^{\textbf{opt}}$ is identified by the Bellman equations \eqref{eq: id Bellman opt 1} and \eqref{eq: id Bellman opt 2}. The $j$th superoptimal initiation regime is identified by 
    \begin{align*}
        &g_K^{\textbf{sup-obs}_j}(h_K, \overline{a}'_{j}) = \argmax_{a_K \in \{0,1\}} \\
        &\mathbb{E}[R_K | H_K = h_K, \overline{A}_{K-1}^{g^{\textbf{sup-obs}_j}+} = \overline{g}^{\textbf{sup-obs}_j}_{K-1}(h_{K-1}, \overline{a}'_{j}), A_K^{g^{\textbf{sup-obs}_j}+} = a_K], \\
        \\
        &g^{\textbf{sup-obs}_j}_k(h_k, \overline{a}'_{j}) = \argmax_{a_k \in \{0,1\}}\\
        &\mathbb{E}[R_k + \mathcal{V}^{\underline{g}^{\textbf{sup-obs}_j}_{k + 1}}_{k + 1}(H_{k+1}, \overline{a}'_{j}) | H_k = h_k, \overline{A}_{k-1}^{g^{\textbf{sup-obs}_j}+} = \overline{g}^{\textbf{sup-obs}_j}_{k-1}(h_{k-1}, \overline{a}'_{j}), A_k^{g^{\textbf{sup-obs}_j}+} = a_k],\\
        &\text{ for $j \leq k < K$,}\\
        \\
        &g^{\textbf{sup-obs}_j}_k(h_k, \overline{a}'_k)= a'_k \text{ for $k < j$},
    \end{align*}
    where $A_{k}^{g^{\textbf{sup-obs}_j}+}$ is identified by the randomized data for $k \geq j$ and by the natural treatment value $A_k$ for $k < j$, and $\mathcal{V}^{\underline{g}^{\textbf{sup-obs}_j}_{k + 1}}_k(H_{k+1}, \overline{a}'_{j})$ is identified recursively.

     Then, the initiation regime is identified by
    \begin{align}
        g^{\textbf{init}} = g^{\textbf{sup-obs}_{T^{\ast}}}, \qquad
        T^{\ast} := \argmax_{T \in \mathcal{T}} \mathbb{E}\big[R^{g^{\textbf{sup-obs}_{T}}}\big] = T^{g^{\textbf{init}}},
        \label{eq: init id formula}
    \end{align}
    where $\mathcal{T}$ is the set of stopping times of Remark \ref{rem: adaptive switch time}. The maximum in \eqref{eq: init id formula} is computed by backward induction, see Algorithm \ref{algo: superoptimal switching} in Appendix \ref{app: Bellman eqs}.

    Data Structure \ref{algo: randomized trial obs} represents the data-generating mechanisms of staggered entry trials \citep{pollock1989survival, sun2021estimating, athey2022design, wing2024designing, xiong_optimal_2024}.

    \begin{algorithm}
    \spacingset{1}
    \SetAlgoRefName{2}
    \SetAlgorithmName{Data Structure}
    \;
    
        \For{$k = 1, \ldots, K$}{
        \For{each individual not in a trial}{
        Randomize individual to be in $k$th trial with known and fixed probability $p(H_k^g, k) \in (0,1)$\;
        \If{individual not in trial}{
        Record $(L_k, A_k, R_k)$\;
        }
        }
        \For{each individual in a trial}{
        Record $L_k$\;
    Randomize individual to receive treatment $A_k^{g+} = 1$ with fixed and known probability $p_k(H_k^g) \in (0,1)$\;
    Record $R_k$\;
        }
        }
        \caption{Sequential entry randomized trial.}
        \label{algo: randomized trial obs}
    \end{algorithm}

    To illustrate the complexity of the trial, suppose that $L_k = \emptyset$ for all $k$. Then, Data Structure \ref{algo: randomized trial obs} has $$1 + \sum_{k = 1}^K 2^{K - k + 1} = 2^{K + 1} - 1$$ arms. This is twice the order of magnitude of the usual sequentially randomized trial, which has $2^K$ arms, but is in the order of magnitude of the class of initiation regimes, $O(2^{K + 1})$, and is smaller than the magnitude of the class of all regimes that use the natural treatment values $O(2^{2K})$. 

     Data Structure \ref{algo: randomized trial obs} describes staggered entry trials where natural treatment values prior to treatment randomization are recorded, such as nested trial designs, also often referred to as trials within cohort studies (TwiCs, \citet{relton2010rethinking}). TwiCs are used in practice \citep{nickolls2024randomised}, where a subset of individuals in a cohort study are randomized to receive an experimental treatment or not. For example, \citet{relton2010rethinking} conducted a trial where they recruited 856 women aged 45-64 and recorded their outcomes; some eligible women in this cohort that experienced frequent or severe menopausal hot flashes were subsequently randomly selected to be offered treatment. If previous treatment and covariate history were recorded, this design matches Data Structure \ref{algo: randomized trial obs}.

    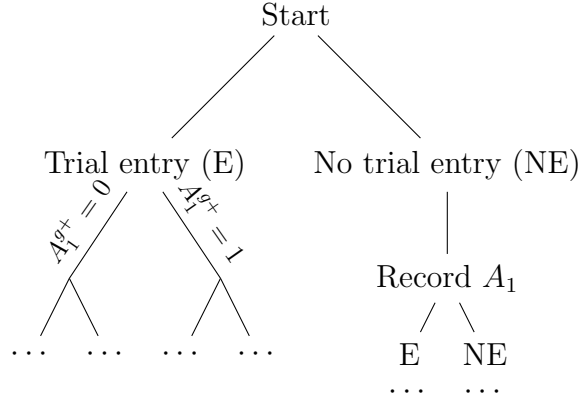
\begin{figure}
        \centering
        \begin{tikzpicture}
            \node (start) at (0,0) {Start};
            \node (entry) at (-2,-2) {Trial entry (E)};
            \node (obs) at (2,-2) {No trial entry (NE)};
            \node (A1) at (2,-3.5) {Record $A_1$};
            \node (3DL) at (-3.5,-4.5) {$\cdots$};
            \node (3DR) at (-0.5,-4.5) {$\cdots$};
            \node (3DMR) at (-1.5,-4.5) {$\cdots$};
            \node (3DML) at (-2.5,-4.5) {$\cdots$};
            \node (3DAL) at (1.5,-4.5) {E};
            \node (3D1) at (1.5,-5) {$\cdots$};
            \node (3DAR) at (2.5, -4.5) {NE};
            \node (3D2) at (2.5,-5) {$\cdots$};

            \draw[-] (start) --  (entry);
            \draw[-] (start) -- (obs);
            \draw[-] (entry) -- node[rotate = 55, above, scale = 0.85] {$A_1^{g+} = 0$} (-3,-3.5);
            \draw[-] (entry) -- node[rotate = 305, above, scale = 0.85] {$A_1^{g+} = 1$} (-1,-3.5);
            \draw[-] (obs) -- (A1);
            \draw[-] (-3,-3.5) -- (3DL);
            \draw[-] (-3,-3.5) -- (3DML);
            \draw[-] (-1,-3.5) -- (3DMR);
            \draw[-] (-1,-3.5) -- (3DR);
            \draw[-] (A1) -- (3DAL);
            \draw[-] (A1) -- (3DAR);
        \end{tikzpicture}
        \caption{Tree describing Data Structure \ref{algo: randomized trial obs}.}
        \label{fig: randomized trial obs + recorded nat val}
    \end{figure}

\begin{remark}
We also describe two further designs, Data Structures \ref{algo: recorded nat value trial} and \ref{algo: three choice trial } in Appendix \ref{app: larger trials}, that identify both $g^{\textbf{init}}$ and $g^{\textbf{sup}}$, but require measurements of natural treatment values at every time point. Another common strategy to identify the $H$-optimal regime using observational or experimental data is to impose a Markov Decision Process (MDP) structure \citep{bellman_mdp_1957, howard1960dynamic,sutton1999reinforcement,wang_provably_2021, shi_off_policy_2024}. We consider identification assumptions for point identification of $g^{\textbf{init}}$ from observational data in Appendix \ref{app: other id ass}. In particular, we can use less restrictive identification assumptions than the ones used for classical MDPs; despite their frequent use, classical MDPs strongly restrict the class of possible data-generating mechanisms, see Appendix \ref{app: other id ass} for a discussion. We also consider partial identification methods in Appendix \ref{sec: Bounds}, and identification using Data Structures \ref{algo: obs + trial }-\ref{algo: three choice trial } with more than two treatments in Appendix \ref{app: multiple treatments}.
\end{remark}

\section{On estimation of expected outcomes under initiation regimes}
\label{app: general estimation result}
Here we give conditions for convergence of estimators of value function under the estimated regime. First, if the propensity score and outcome regressions can be estimated at $o_P(n^{-1/4})$ rates or if the propensity score is known, then under additional regularity conditions we obtain a $\sqrt{n}$-asymptotically normal estimator, see Proposition \ref{prop: estimation convergence}. Furthermore, in Section \ref{app: IF estimation}, we outline a strategy to leverage influence-function-based estimators, \citep{van1988estimating,newey1990semiparametric,robins1992recovery,robins1994estimation,rotnitzky1998semiparametric,van_der_laan2015targeted, luedtke2016super, michael2024instrumental}, to estimate the optimal initiation regime and its value functions. Then, we describe how to modify existing influence-function-based estimators of regimes and value functions. In some settings, these type of estimators can be used to estimate optimal regimes when the value function is only partially identified (bounded), see also \citet{laurendeau2024improved} for a discussion on the difference between regime and value function identification. As a specific example, Appendix \ref{app: chen and zhang time-varing IV estimation} presents algorithms for regime estimation in time-varying instrumental variable settings, including settings where the value function is not point identified. For completeness, we also include a convergence result for the one-step estimator in Proposition \ref{prop: IF asymptotic normality}.

\subsection{A strategy for estimation}
\label{app: estimation propositon}
Before considering estimation, we formalize a useful identification condition for $g^{\textbf{init}}$:
\begin{proposition}
    When $\mathbb{E}[R_k^{a_k} | \overline{L}_k, \overline{R}_{k-1}, \overline{A}_{k-1}]$ and value functions $\mathbb{E}[R_k^{\overline{a}_k} | H_k^{\overline{a}_{k-1}}, \overline{A}_j = \overline{A}_j^{g+} = \overline{a}_j]$,  $j < k$ are identified for all $j,k$, $\overline{a}_k \in \{0,1\}^k$, then $g^{\textbf{init}}$ is identified.
    \label{prop: initiation regime id}
\end{proposition}
It is sufficient for $\mathbb{E}[R_k^{\overline{a}_k} | H_k^{\overline{a}_{k-1}}, \overline{A}_j = \overline{A}_j^{g+} = \overline{a}_j]$, to be identified for all $j<k$, $\overline{a}_k \in \{0,1\}^k$ for Proposition \ref{prop: initiation regime id} to hold, as $\mathbb{E}[R_k^{a_k} | \overline{L}_k, \overline{R}_{k-1}, \overline{A}_{k-1}]$ is the special case where $j = k-1$. Then we can write $\mathbb{E}[R^{\overline{a}_k} | H_k^{\overline{a}_{k-1}}]$ as a function of $H_k^{\overline{a}_{k-1}}$, $A_k^{\overline{a}_{k-1}}$ and $R^{\overline{a}_{k-1}}$ for all $k = 1, \ldots, K$. 
   We say that the value functions $\mathbb{E}[R^{\overline{a}_k} | H_k^{\overline{a}_{k-1}}]$ are \emph{sequentially identified} when we know functions $f_k$ such that $\mathbb{E}[Y_k^{\overline{a}_k} | H_k^{\overline{a}_{k-1}}] = f_k(H_k^{\overline{a}_{k-1}}, A_k^{\overline{a}_{k-1}}, Y_k^{\overline{a}_{k-1}})$ for all $k = 1, \ldots, K$ and outcomes $Y_k = f_{k,Y}(H_k, \overline{A}_k, \overline{R}_k)$ for some real-valued $f_{k,Y}$ realized after $A_K$. Such sequential identification will be useful for $Y_k = R_kI(\overline{A}_k = \overline{a}_k')$ in Appendix \ref{app: Dirac identification}, where we discuss identification of $g^{\textbf{sup}}$, which generalizes $g^{\textbf{init}}$.

We denote estimators with hats. Suppose that we estimate $\mathbb{E}[R^g]$ by averaging over the possible values of the history of covariates, rewards, and treatments and estimating expectations of counterfactual outcomes that are identified by assumption in Proposition \ref{prop: initiation regime id}. If the propensity score and the outcome regression can be estimated at $o_P(n^{-1/4})$ rates, or the propensity score is known, then $\hat{\mathbb{E}}[R^g]$ is an asymptotically normal consistent estimator:

\begin{proposition}
For $k = 1, \ldots, K$, let $j_k := \max\{i < k : \overline{A}_i^{g+} = \overline{A}_i^g\}$,
$\mu_k^g := \mathbb{E}[R_k^g \mid H_k^g, \overline{A}_{j_k}^g]$, and let $f_k^g$ be a
function of the observed data with $P f_k^g = \mathbb{E}[R_k^g]$, or $f_k^g = \mu_k^g + \mathbb{IF}_k$ as in Section
\ref{app: IF estimation}. Let $\hat{f}_k^g$ be an estimator of $f_k^g$ and
$\hat{\mathbb{E}}(R^g) := \mathbb{P}_n \sum_{k = 1}^K \hat{f}_k^g$. If for all $k$,
\begin{enumerate}[label=(\roman*)]
  \item $\|\hat{f}_k^g - f_k^g\|_{L^2(P)} = o_P(1)$, \label{assL2: 1}
  \item $P(\hat{f}_k^g - f_k^g) = o_P(n^{-1/2})$, \label{assL2: 6}
  \item $\hat{f}_k^g$ is $P$-Donsker (or estimated in a separate sample), \label{assL2: 7}
  \item $|R_k^{a_k}| \leq C$ and $|\hat{f}_k^g| \leq C$ with probability one for some $C > 0$, \label{assL2: 3}
\end{enumerate}
then $\sqrt{n}(\hat{\mathbb{E}}(R^g) - \mathbb{E}(R^g)) \to^d
N(0, \sigma_g^2)$ with $\sigma_g^2 = \mathrm{Var}(\sum_{k = 1}^K f_k^g)$.
Furthermore, if
\begin{enumerate}[label=(\roman*)]\setcounter{enumi}{4}
  \item $\hat{g}$ is $P$-Donsker (or estimated in a separate sample), \label{assL2: 8}
  \item $\|\hat{g}_k - g_k\|^2_{L^2(P)} = P(\hat{g}_k \neq g_k) = o_P(n^{-1/2})$ for all $k$, \label{assL2: 9}
\end{enumerate}
then $\sqrt{n}\,(\hat{\mathbb{E}}(R^{\hat{g}}) - \mathbb{E}(R^g)) \to^d N(0, \sigma_g^2)$.
\label{prop: estimation convergence}
\end{proposition}

The value functions $\mathbb{E}[R_k^g |H_k^g, \overline{A}_{j}^g = \overline{A}_j^{g+}]$ can be identified based on sequential IVs, for example using the strategy suggested by \citet{chen_estimating_2023}. \citet{chen_estimating_2023} also gave estimators of the value functions $\mathbb{E}[R_k^g |H_k, \overline{A}_{j} = \overline{A}_j^{g+}]$ in sequential IV settings. In particular, influence-function based estimators such as the one-step estimator can be used to estimate $\mathbb{E}[R_k^g |H_k^g, \overline{A}_{j}^g = \overline{A}_j^{g+}]$ at the desired rates, see Section \ref{app: IF estimation} for details on the required assumptions.

\begin{remark}[On estimation of regimes]
    In practice, optimal regimes are usually unknown and need to be estimated from data. This is different from estimating the value function under a known regime $g$, $\mathbb{E}[R^g]$, it requires an additional set of assumptions on convergence of $\hat{g}$. In particular, let 
    \begin{align*}
        \mathcal{C}_k(H_k^g, \overline{A}_j) &:= \mathbb{E}[R_k^{\overline{a}_{k-1}, 1} + \mathcal{V}_{P,k+1}^{\overline{a}_{k-1}, 1, \underline{g}_{k+1}}(H_{k+1}^{\overline{a}_{k-1}, 1}) | H_k^{\overline{a}_{k-1}}, \overline{A}_j = \overline{A}_j^{g+} = \overline{a}_j]\\
        &- \mathbb{E}[R_k^{\overline{a}_{k-1}, 0} + \mathcal{V}_{P,k+1}^{\overline{a}_{k-1}, 0, \underline{g}_{k+1}}(H_{k+1}^{\overline{a}_{k-1}, 0}) | H_k^{\overline{a}_{k-1}}, \overline{A}_j = \overline{A}_j^{g+} = \overline{a}_j].
    \end{align*}
    Then, assuming $||\hat{g}-g||_{L^2(P)}^2 = P(\hat{g} \neq g) = o_P(n^{-1/2})$ implies that there is no exceptional law, that is, $P(\mathcal{C}_k(H_k^g, \overline{A}_j) = 0) = 0$. Furthermore, it requires that the signs of $\mathcal{C}_k(H_k^g, \overline{A}_j)$ are consistently estimated, for example by consistently estimating the contrast functions directly. 

    
\end{remark}

\begin{remark}[On the variance of estimators for different regimes]
    The observed regime and $g^{\textbf{init}}$ can coincide at several time-points. Usually estimation of value functions under the observed regime is easier than estimation under, say, conventional static and dynamic regimes. This gives intuition why estimators for the value function under $g^{\textbf{init}}$  often has lower variance than those for the value of $g^{\textbf{opt}}$ and $g^{\textbf{sup}}$. This is particularly interesting if we only partially identify counterfactual expectations, as the observed regime is always point identified, see \citet{laurendeau2024improved} for more arguments in the $K = 1$ setting.
    \label{rem: Lower variance of init regime}
\end{remark}

\begin{remark}[Augmenting the analysis with data from a randomized trial]
    Proposition \ref{prop: estimation convergence} clarifies how we can use randomized data to supplement our estimation procedure for $g^{\textbf{init}}$, e.g. if we have data in the form of Data Structure \ref{algo: obs + trial }. Indeed, data from a randomized trial can be used to estimate $\mathbb{E}[R^{a_1} | H_1]$ and $\mathbb{E}[R^{g} | H_1]$ for a dynamic regime $g$ that does not depend on natural treatment values.
\end{remark}

\subsection{Generic strategy for deriving influence-function based estimators}
\label{app: IF estimation}

Influence-function based estimators of $\mathbb{E}[R_k^{\overline{a}_k} | H_k^{\overline{a}_{k-1}}, \overline{A}_j = \overline{a}_j]$ for $0 \leq j < k$ include the one-step estimator
\begin{align*}
    \hat{\mathbb{E}}[R_k^{\overline{a}_k} | H_k^{\overline{a}_{k-1}}, \overline{A}_j = \overline{a}_j] + \hat{\mathbb{IF}}(\mathbb{E}[R_k^{\overline{a}_k} | H_k^{\overline{a}_{k-1}}, \overline{A}_j = \overline{a}_j]),
\end{align*}
where we use hats to denote estimators and $\mathbb{IF}(\mathbb{E}[R_k^{\overline{a}_k} | H_k^{\overline{a}_{k-1}}, \overline{A}_j = \overline{a}_j])$ is the influence function of 
$\mathbb{E}[R_k^{\overline{a}_k} | H_k^{\overline{a}_{k-1}}, \overline{A}_j = \overline{a}_j]$.

If $j = k-1$, as for the initiation time-point for initiation regimes, $$\mathbb{E}[R_k^{\overline{a}_k} | H_k^{\overline{a}_{k-1}}, \overline{A}_j = \overline{a}_j] = \mathbb{E}[R_k^{a_k} | H_k, \overline{A}_{k-1} = \overline{a}_{k-1}].$$ 

Then, 
\begin{align*}
    &\mathbb{IF}(\mathbb{E}[R_k^{a_k} | H_k, \overline{A}_{k-1} = \overline{a}_{k-1}, A_k = a_k'])\\
    &= \begin{cases}
        \mathbb{IF}(\mathbb{E}[R_k | H_k, \overline{A}_k = \overline{a}_k]) & \text{if $a_k' = a_k$,}\\
        \mathbb{IF}(\frac{\mathbb{E}[R_k^{a_k} | H_k, \overline{A}_{k-1} = \overline{a}_{k-1}] - \mathbb{E}[R_k | H_k, \overline{A}_k = \overline{a}_k]P(A_k = a_k | H_k, \overline{A}_{k-1} = \overline{a}_{k-1}) }{P(A_k = a_k' | H_k, \overline{A}_{k-1} = \overline{a}_{k-1})}) & \text{if $a_k' \neq a_k$,}
    \end{cases}\\
    &= \begin{cases}
        \frac{I(\overline{A}_k = \overline{a}_k)}{P(\overline{A}_k = \overline{a}_k \mid H_k)}(R_k - \mathbb{E}[R_k^{a_k} | H_k, \overline{A}_{k} = \overline{a}_{k}]) & \text{if $a_k' = a_k$,}\\
          \Psi(\mathbb{IF}(\mathbb{E}[R_k^{a_k} | H_k, \overline{A}_{k-1} = \overline{a}_{k-1}]), \overline{A}_k, H_k) & \text{if $a_k' \neq a_k$,}
    \end{cases}\\
\end{align*}
where $\Psi(\mathbb{IF}(\mathbb{E}[R_k^{a_k} | H_k, \overline{A}_{k-1} = \overline{a}_{k-1}]), \overline{A}_k, H_k)$ is derived from the analogous formula in \citet{stensrud_optimal_2024}, see Appendix \ref{app: IF formula}.

Conceptually, under suitable regularity conditions, influence-function based estimators can be used to consistently estimate the value functions at a fast rate, which is required for convergence of the estimators as described in Section \ref{app: estimation propositon}. More specifically, consider the following result on estimation of $\mathbb{E}[R^g]$:
\begin{proposition}
For $k = 1, \ldots, K$, let $\hat{f}_k^g$ be the one-step estimator of Section \ref{app: IF estimation} at
stage $k$ built from an outcome regression $\hat{\mu}_k$ where $\mu_k = \omega_k(H_k)$ and an estimated
propensity score $\hat{\pi}_k$ for the true propensity $\pi_k$.

Suppose that for all $k$,
\begin{enumerate}[label=(\roman*)]
  \item $\|\hat{\mu}_k - \mu_k\|_{L^2(P)}\, \|\hat{\pi}_k - \pi_k\|_{L^2(P)} = o_P(n^{-1/2})$,
  \item $\|\hat{f}_k^g - f_k^g\|_{L^2(P)} = o_P(1)$,
  \item $\hat{f}_k^g$ is $P$-Donsker (or estimated in a separate sample),
  \item $\hat{\pi}_k \in (\epsilon, 1-\epsilon)$, $|R_k| \leq C$ and $|\hat{f}_k^g| \leq C$ with probability one.
\end{enumerate}
Then $\sqrt{n}\,(\mathbb{P}_n \sum_{k=1}^K \hat{f}_k^g - \mathbb{E}[R^g]) \to^dN(0, \sigma_g^2)$,
$\sigma_g^2 = \mathrm{Var}(\sum_{k=1}^K f_k^g)$. 

If in addition conditions \ref{assL2: 8} and \ref{assL2: 9} of Proposition
\ref{prop: estimation convergence} hold, then
$\sqrt{n}\,(\mathbb{P}_n \sum_{k=1}^K \hat{f}_k^{\hat{g}} - \mathbb{E}[R^g]) \to^dN(0, \sigma_g^2)$.
\label{prop: IF asymptotic normality}
\end{proposition}
The value functions $\mathbb{E}[R_k^g |H_k^g, \overline{A}_{j}^g = \overline{A}_j^{g+}]$ and the propensities $P(A_k^g  = 1 |H_k^g, \overline{A}_{j}^g = \overline{A}_j^{g+})$ that we need to estimate do not appear in conventional estimators of optimal regimes. However, analogous expressions where $\overline{A}_j^g$ is absent from the conditioning often appear in these estimators. Thus, if we regard the natural treatment values $\overline{A}_j^g$ as additional time-varying covariates, then the value functions we consider correspond to the usual value functions and propensities estimated in Q-Learning algorithms \citep{murphy2003optimal, chakraborty2013statistical, chen_estimating_2023}. The estimator given in Proposition \ref{prop: IF asymptotic normality} then corresponds to the time-varying one-step estimator.

We give the full expression of the influence function estimator in Appendix \ref{app: IF formula}.

\section{Example: Optimal treatment of chronic back pain}
\label{sec: Sim}
Consider the chronic back pain example described in Section \ref{sec: chronic back pain ex intro} and illustrated in Figure \ref{fig: Batorsky trial illustration}. We use the data generating mechanism from \citet{batorsky2024integrating}.  The data sample consists of $1'630$ units, $630$ from the synthetic randomized trial and $1'000$ from the synthetic observational study by \citet{batorsky2024integrating}, see Appendix \ref{app: Data generating mechs} for details. 

To combine the randomized and observational data, we use the data structures described in Section \ref{sec: data fusion} and correctly specified estimators for $g^{\textbf{opt}}$, $g^{\textbf{obs}}$ and $g^{\textbf{init}}$. Then, we estimate $\mathbb{E}[R^{\hat{g}^{\textbf{opt}}}],$ $\mathbb{E}[R^{\hat{g}^{\textbf{obs}}}]$, and $\mathbb{E}[R^{\hat{g}^{\textbf{init}}}]$ on a new test set of $20'000$ units. 

In this setting, the $H$-optimal regime performs only slightly better than the observed regime (Table \ref{tab: trial sim results}). The optimal initiation regime slightly outperforms the $H$-optimal regime and the observed regime. 

\begin{table}
    \centering
    \spacingset{1}
    \begin{tabular}{c | c | c | c}
         \textbf{Regime} & \textbf{DS \ref{algo: randomized trial obs}} & \textbf{DS \ref{algo: recorded nat value trial}} & \textbf{DS \ref{algo: three choice trial }}\\
         \hline
         $\hat{g}^{\textbf{obs}}$& 9.797 (9.747, 9.847) & 9.789 (9.739, 9.840) & 9.778 (9.728, 9.828)\\
         $\hat{g}^{\textbf{opt}}$& 9.875 (9.829, 9.921) & 9.857 (9.811, 9.903) & 9.867 (9.820, 9.914)\\
         $\hat{g}^{\textbf{init}}$& 10.099 (10.053, 10.144) & 10.194 (10.149, 10.239) & 10.171 (10.125, 10.217)\\
    \end{tabular}
    \caption{Expected values and $95\%$ confidence intervals of the regimes obtained by Data Structures (DS) \ref{algo: randomized trial obs}, \ref{algo: recorded nat value trial}, and \ref{algo: three choice trial } on the test set with $n_{\text{eval}} = 20'000$ samples in the chronic back pain example from \citet{batorsky2024integrating}. Higher values are better.}
    \label{tab: trial sim results}
\end{table}

With a slight modification of the original data mechanism, where only the interaction effects of $U$ and the treatments on the outcomes of \citet{batorsky2024integrating} are changed, the optimal initiation regime can considerably improve outcomes of the observed and $H$-optimal regimes, see Table \ref{tab: modified trial sim results} and also Appendix \ref{app: Data generating mechs} for more details. This illustrates that both doctors and the $H$-optimal regimes can be suboptimal, compared to the optimal initiation regime.

\begin{table}
    \centering
    \spacingset{1}
    \begin{tabular}{c | c | c | c}
         \textbf{Regime} & \textbf{DS \ref{algo: randomized trial obs}} & \textbf{DS \ref{algo: recorded nat value trial}} & \textbf{DS \ref{algo: three choice trial }}\\
         \hline
         $\hat{g}^{\textbf{obs}}$ & 8.636 (8.590, 8.682) & 8.613 (8.566, 8.660) & 8.613 (8.566, 8.659)\\
         $\hat{g}^{\textbf{opt}}$ & 10.553 (10.502, 10.603) & 10.563 (10.513, 10.612) & 10.472 (10.421, 10.523)\\
         $\hat{g}^{\textbf{init}}$ & 11.229 (11.176, 11.283) & 11.275 (11.221, 11.328) & 11.177 (11.122, 11.232)\\
    \end{tabular}
    \caption{Expected values and $95\%$ confidence intervals of the regimes obtained by Data Structures (DS) \ref{algo: randomized trial obs}, \ref{algo: recorded nat value trial}, and \ref{algo: three choice trial } on the test set with $n_{\text{eval}} = 20'000$ samples in the modified example of \citet{batorsky2024integrating}.  Higher values are better.}
    \label{tab: modified trial sim results}
\end{table}

We provide simulations based on the \citet{batorsky2024integrating} data-generating mechanism for $K = 3$ and $K = 4$ in Appendix \ref{app: Data generating mechs}. Furthermore, to illustrate how our results are relevant also when only observational data are available, we give a simulation in a sequential instrumental variable setting in Appendix \ref{sec: sequential IV data analysis}.


\section{Future directions}
We have considered inference on initiation regimes in an offline setting. There is substantial interest in optimal regimes in online settings too, in particular in the reinforcement learning literature \citep{silver2016mastering,zhang2019near}. For example, previous data on expert human decision makers can be used to evaluate possible actions, as implemented by Google DeepMind's AlphaGo \citep{silver2016mastering}. Alternatively, human decision-making can be directly included as a covariate during the learning and decision-making process, as used for autonomous driving agents \citep{wu2021human, wu2022prioritized, wu2023toward}. Optimally combining human and AI knowledge to achieve the best expected outcomes is a challenge in many areas where AI is now used to make decisions traditionally made by humans \citep{verghese2018computer, matheny2019artificial}.

\if1\blind
{
\subsection*{Acknowledgments}
The authors were supported by the Swiss National Science Foundation, grant $200021\_207436$.

\subsection*{Tool and computational resource disclosure} 
The authors used large language models (Claude Fable 5.1 and Opus 5 and ChatGPT 5.6 Sol) to read the statements and proofs critically and make edits in response, to suggest papers and arguments relevant to the use of initiation regimes in practice, and to make changes and suggestions for improved replication and performance of the code, which was initially written without the use of large language models. All suggestions, text and references produced with these tools were verified and edited by the authors.

\subsection*{Disclosure statement} The authors have no competing interests to declare.
} \fi

\clearpage
\appendix
\setcounter{figure}{0}\renewcommand{\thefigure}{S\arabic{figure}}
\setcounter{table}{0}\renewcommand{\thetable}{S\arabic{table}}
\setcounter{algocf}{0}\renewcommand{\thealgocf}{S\arabic{algocf}}
\setcounter{equation}{0}\renewcommand{\theequation}{S\arabic{equation}}


\section*{Notation Table}
\label{app: notation table}
See Table \ref{tab:notation}.

\begin{table}
    \spacingset{1}
    \centering
    \begin{tabular}{c | c}
         $K$ & Number of time (decision) points, \\
         $A_k$ & (Observed) treatment at time $k$, \\
         $L_k$ & (Observed) covariates at time $k$, \\
         $R_k$ & (Observed) reward at time $k$, \\
         $R$ & Utility function, taken to be $\sum_{k = 1}^K R_k$,\\
         $U_k$ & Unobserved confounders at time $k$, \\
         $Z_k$ & (Observed) instrumental variable at time $k$,\\
         $g$ &  (Arbitrary) regime, \\
         $g_k$ & Regime decision at time $k$, \\
         $A_k^g$ & Natural treatment value at time $k$ under regime $g$, \\
         $A_k^{g+}$ & Value of regime $g$ at time $k$, equal to $g_k$, \\
         
         $H_k^g$ & History of previous rewards and covariates\\ 
         & and current covariates  
          under regime $g$, equal to $(\overline{R}_{k-1}^g, \overline{L}_{k}^g)$,\\
         
         $\mathcal{H}_k$ & Domain of $H_k^g$,\\
         
         $\mathcal{G}$ & Set of regimes that depend only on history $H_k^g$ at each time $k$, \\
         &equal to $\{g = (g_1, \ldots, g_K) : g_k : \mathcal{H}_k \to \{0,1\}\}$,\\

         $\mathcal{G}_j$ & Set of regimes that depend only on history $H_k^g$,\\
         & conditional on an observed past $H_j = h_j, \overline{A}_{j-1} = \overline{a}_{j-1}'$,\\

         $\mathcal{G}_j^{\textbf{sup}}$ & Set of regimes that depend only on history $H_k^g$ and $A_j$,\\
         & conditional on an observed past $H_j = h_j, \overline{A}_{j-1} = \overline{a}_{j-1}'$,\\

         $\mathcal{G}^{\textbf{init}}$ & Set of initiation regimes,\\

         $\mathcal{G}^{\textbf{sup}}$ & Set of regimes that depend on history $H_k^g$ \\
         & and natural treatment values,\\
         $g^{\textbf{opt}_j}$ & $H$-optimal regimes conditional on an observed past,\\
         & $H_j = h_j, \overline{A}_{j-1} = \overline{a}_{j-1}'$,\\
         $g^{\textbf{sup}_j}$ & $H$-optimal regimes conditional on an observed past,\\
         & $H_j = h_j, \overline{A}_{j} = \overline{a}_{j}'$,\\
         $g^{\textbf{obs}_j}$ & $j$th initiation regime,\\
         $g^{\textbf{sup-obs}_j}$ & $j$th superoptimal initiation regime,\\
         $g^{\textbf{opt}}$ & $H$-optimal regime, equal to $\argmax_{g \in \mathcal{G}} \mathbb{E}[R^g]$,\\
         $g^{\textbf{obs}}$ & Observed regime,\\
         $g^{\textbf{osh}}$ & Optimal shifting regime,\\
         $g^{\textbf{init}}$ & Optimal initiation regime,\\
         $g^{\textbf{sup}}$ & Superoptimal regime,\\

         $\mathcal{V}_{P,k}^g$ & Value function of $g$ at time $k$,\\
         $Q_k(h_k, a_k)$ & Counterfactual value of intervention $A_k = a_k$ given history $H_k^g = h_k$,\\
         $Q_k^{\textbf{init}}(h_k, \overline{a}_k,a_k')$ & Counterfactual value of intervention $A_k = a_k'$, given \\
         &history $H_k^g = h_k$ and observed regime $\overline{A}_k = \overline{a}_k$,\\
         $\mathcal{C}_k$ & Contrast function,\\
         $\mathbb{P}_n$ & Empirical distribution operator,\\
         $P$ & Factual distribution operator,\\
         $(\mathbf{L}_k, \mathbf{U}_k)$ & Lower and upper bounds of value function at time $k$, respectively,\\
         $w_k$ & Weight function at time $k$ of decision criteria from \citet{cui2021individualized},\\
         $\mathbb{IF}$ & Influence function operator\\
    \end{tabular}
    \caption{Table of notation.}
    \label{tab:notation}
\end{table}

\section{A brief note on related literature}
\label{app: related_literature}
Most of the existing results in the literature on optimal regimes, particularly in longitudinal settings, rely on assumptions of no unmeasured confounding \citep{murphy2003optimal,robins2004optimal, chakraborty2014dynamic,schulte2014q_and_a,clifton2020qlearning}, which can be ensured by design in SMART trials. However, the recent growth of work on data fusion and data integration reflects that additional observational data often are available, which can be used in combination with trial data to improve inference \citep{zhang2019near,batorsky2024integrating}. Combining trial and observational data can have useful implications, such as increased statistical power \citep{batorsky2024integrating}, in particular as trials often lack power to detect heterogeneous causal effects \citep{rothwell2005external,bareinboim2016causal,frieden2017evidence,shi2023data,brantner2023methods}, and, as we will discuss, to identify regimes with improved guarantees. The new guarantees are also relevant to the increasing body of work on learning optimal regimes from observational data, where there might be unmeasured confounding \citep{kallus_confounding_2018, kallus2020confounding}. Some of these results rely on theory for Markov Decision Processes (MDPs), and some consider more general settings \citep{kallus2020confounding, tennenholtz2020off, wang_provably_2021, shi_off_policy_2024}. In particular, some results have used instrumental variables (IVs) or other proxies to identify causal effects in the presence of unmeasured confounding \citep{tennenholtz2020off, chen_estimating_2023}. Like the setting without unmeasured confounding, these works have considered regimes that are function of (time-varying) measured covariates but, to our knowledge, not natural treatment values.

\citet{wang2022blessing} also considered extensions of the superoptimal regime for longitudinal data, in particular longitudinal observational data, using proximal inference for identification. While we also study the longitudinal superoptimal regime in observational settings, we look at the related instrumental variables framework. Moreover, we also consider potential randomized study designs for identification and target the initiation regime, which we believe is more practically useful in practice, as it requires weaker identification assumptions while leveraging the natural treatment value.

Other authors have considered using natural treatment values for longitudinal regimes, see \citet{cao2024hr, maiti1}. Unlike \citet{cao2024hr}, our approach is fully nonparametric and we allow for arbitrary changes in the covariate distribution in response to the treatments, and we use natural treatment values not as a ``warm-start" to identify the optimal regime but as a way to get an improved initiation or superoptimal regime. Unlike \citet{maiti1}, we also focus on identification, in particular in randomized study designs, and estimation of the superoptimal regime and target the initiation regime.

\section{Baseline regimes}
\label{app: other baseline regimes}
The ``never-treat" regime $g^0$ \citep{luedtke2016optimal, cui2021individualized, qiu2022individualized} and the arbitrary baseline regime $g^b$, where $b$ stands for "baseline", from \citet{kallus_confounding_2018} are functions of the history $H_k$ only.\footnote{The "never-treat" regime is a trivial function of the history that assigns $g_k(H_k) = 0$ regardless of the history $H_k$.}

Hence, by the definition of $g^{\textbf{opt}}$,
\begin{align*}
    \max(\mathbb{E}[R^{g^0}], \mathbb{E}[R^{g^b}]) \leq \mathbb{E}[R^{g^{\textbf{opt}}}].
\end{align*}

Then, as $\mathbb{E}[R^{g^{\textbf{init}}}] \geq \mathbb{E}[R^{g^{\textbf{opt}}}]$, $g^{\textbf{init}}$ can do no-worse in expectation than $g^0$ and $g^b$, and hence is still preferable to other baseline regimes.

\section{Optimal shifting regime}
\label{app: osh regime}
\begin{definition}[Optimal shifting regime]
The optimal shifting regime is 
    \begin{align*}
        g^{\textbf{osh}} := \argmax_{ j = 1, \ldots, K + 1} \mathbb{E}[R^{g^{\textbf{obs}_j}}].
    \end{align*}
\end{definition}
A naive procedure to find $g^{\textbf{osh}}$ computes the value function of each $g^{\textbf{obs}_j}$ in $O(K-j)$ time and then selects the maximum $j$ in $O(K^2)$ total time, see Algorithm \ref{algo: naive optimal switching}.
\begin{algorithm}
    \KwData{Value functions $\mathbb{E}[R_k^{\overline{a}_k} | H_k^{\overline{a}_{k-1}}, \overline{A}_{\max\{j = 1, \ldots, k: \overline{a}_j = \overline{A}_j\}}]$.}
    \KwResult{Optimal shifting regime $g^{\textbf{osh}}$.}
    Compute $g^{\textbf{obs}_1}$ and $V_1^{g^{\textbf{obs}_1}}$ using the Bellman equations\;
    $g \leftarrow g^{\textbf{obs}_1}$\;
    \For{$j = 2,\ldots,K+1$}{
    Compute $g^{\textbf{obs}_j}$ and $V_1^{g^{\textbf{obs}_j}}$ using the Bellman equations\;
    \If{$V_1^{g^{\textbf{obs}_j}} > V_1^{g}$}{
    $g \leftarrow g^{\textbf{obs}_j}$\;
    }
    }
    Return $g$\;
    \caption{Naive algorithm to identify $g^{\textbf{osh}}$.}
    \label{algo: naive optimal switching}
\end{algorithm}

However, we can optimize this procedure to run in linear time using backward induction \citep{bellman1956dynamic}. In particular, it is optimal to deviate from the observed regime at time step $j$ if
\begin{align*}
    &\mathbb{E}[R_j^{1-A_j} + V_{j + 1}^{g^{\textbf{opt}_j}}(H_{j + 1}, \overline{L}_j, \overline{A}_{j-1}) | \overline{L}_j, \overline{A}_{j-1}]\\
    &\geq \mathbb{E}[R_j^{A_j} + V_{j + 1}^{g^{\textbf{osh}}}(H_{j + 1},  \overline{L}_j, \overline{A}_{j-1}) | \overline{L}_j, \overline{A}_{j-1}].
\end{align*}

As $g^{\textbf{osh}}$ optimizes over a class of regimes containing the optimal and the observed regime, the optimal shifting regime $g^{\textbf{osh}}$ outperforms the $H$-optimal and the observed regime,
    \begin{align*}
        \mathbb{E}[R^{g^{\textbf{osh}}}] \geq \max\left(\mathbb{E}[R^{g^{\textbf{obs}}}], \mathbb{E}[R^{g^{\textbf{opt}}}]\right).
    \end{align*}

\section{Bellman equations and regime set sizes}
\label{app: Bellman eqs}

Algorithm \ref{algo: superoptimal switching} derives $g^{\textbf{init}}$ from its Bellman equations. 
\begin{algorithm}
\spacingset{1}
    \KwData{Value functions $\mathbb{E}[R_k^{a_k}| H_k^{\overline{a}_{k-1}}, \color{red}\overline{A}_k \color{black}]$ and $\mathbb{E}[R_k^{\overline{a}_{k-1}, a_k} | H_k^{\overline{a}_{k-1}}, \color{red}\overline{A}_{\max\{j = 1, \ldots, k: \overline{a}_j = \overline{A}_j\} } \color{black}]$.}
    \KwResult{Optimal initiation regime $g^{\textbf{init}}$.}
    $g_K^{\textbf{init}}(h_K, \overline{a}_K') \leftarrow \argmax_{a_K \in \{0,1\}} I[\overline{a}_{K-1} = \overline{a}_{K - 1}'] \mathbb{E}[R_K^{a_K} | H_K^{\overline{a}_{K-1}} = h_K, \color{red}\overline{A}_K = \overline{a}_K'\color{black}]$
    \nonl$+ I[\overline{a}_{K-1} \neq \overline{a}_{K - 1}']\mathbb{E}[R_K^{\overline{a}_{K-1}, a_K} | H_K^{\overline{a}_{K-1}} = h_K, \color{red}\overline{A}_{\max\{j = 1, \ldots, K: \overline{a}_j = \overline{A}_j = \overline{a}_j'\} } \color{black}]$\;
    \For{$k = K-1,\ldots,1$}{
    $g_k^{\textbf{init}}(h_k, \overline{a}_k') \leftarrow \argmax_{a_k \in \{0,1\}}I[\overline{a}_{k-1} = \overline{a}_{k - 1}']$\\
    \nonl$\cdot \mathbb{E}[R_k^{a_k} + \mathcal{V}_{P,k+1}^{\overline{a}_{k}}(H_{k+1}^{\overline{a}_{k}})| H_k^{\overline{a}_{k-1}} = h_k, \color{red}\overline{A}_k = \overline{a}_k'\color{black}]$\\
    \nonl$ + I[\overline{a}_{k-1} \neq \overline{a}_{k - 1}']$\\
    \nonl$\cdot\mathbb{E}[R_k^{\overline{a}_{k-1}, a_k} + \mathcal{V}_{P,k + 1}^{\overline{a}_{k}}(H_{k+1}^{\overline{a}_{k}}) | H_k^{\overline{a}_{k-1}} = h_k, \color{red}\overline{A}_{\max\{j = 1, \ldots, k: \overline{a}_j = \overline{A}_j = \overline{a}_j'\} } \color{black}]$\;
    }
    Return $g^{\textbf{init}}$\;
    \caption{Optimal initiation algorithm based on Bellman equations.}
    \label{algo: superoptimal switching}
\end{algorithm}

A naive method for computing $g^{\textbf{init}}$ would compute the value for all possible regimes in $\mathcal{G}^{\textbf{init}}$ and take the maximum. As there are at least two treatment choices at each time $k$, there would be $ \Omega(2^K)$ evaluations required. When using backward induction with the Bellman equations, we compute a constant amount (that depends on the size of the space of covariates and rewards) of expectations at each time step $k$; one evaluation per current history $H_k^g$, action $a_k$, and next possible history $H_{k+1}$. Hence, the number of value function evaluations scales linearly in $K = o(2^K)$.

\subsection{Regime set sizes}
\label{app: regime set size}
We want to understand the sizes of the classes of regimes $\mathcal{G}$, $\mathcal{G}^{\textbf{init}}$, and $\mathcal{G}^{\textbf{sup}}$. 
To get a sense of the possible number of regimes we consider, suppose that the set of covariates $L$ is empty and consider a pre-specified trajectory of rewards and treatments such that we do not consider prior rewards $R_{k-1}^{\overline{a}_{k-1}}$ or treatments $\overline{A}_{k-1}^{\overline{a}_{k-1}}$ at each time $k$, within each class, the regimes whose decision rule at each time $k$ is one of $0$, $1$, $A_k$ or $1 - A_k$. Then the number of such regimes in each class is finite.

Specifically, at each time-point $j$, given the past there are $2^{K -j + 1} + 2^{K - j}$ regimes if we do not follow the observed regime; either $g_k = 0$ or $g_k = 1$ for any value $A_k$ at all times $k > j$ ($2^{K-j + 1}$ possibilities) or do the opposite of the observed regime, $1-A_j$, and $g_k = 0$ or $g_k = 1$ at each of the $K -j$ time-points in the future ($2^{K-j}$ possibilities). Then, one must count the observed regime (one possibility).

Thus, the optimal initiation regime $g^{\textbf{init}}$ optimizes over a class of regimes of size $$|\mathcal{G}^{\textbf{init}}| := \sum_{j = 1}^K (2^{K -j + 1} + 2^{K-j}) +1 = 2^{K + 1} + 2^K - 2 \in [2^K, 2^{2K}]. $$ 

In contrast, the $H$-optimal regime $g^{\textbf{opt}}$ optimizes over a class of $2^K$ regimes, and the superoptimal regime optimizes over a class of $2^{2K} = 4^K$ regimes. This confirms that the computational complexity of a method that naively compares all possible regimes in the space requires an exponentially increasing number of evaluations. Conceptually, backward induction circumvents this burden by ignoring regimes that are known to be suboptimal \citep{bertsekas1996neuro}.

\section{Additional trial designs}
\label{app: larger trials}
We describe two further designs that, unlike Data Structures \ref{algo: obs + trial } and \ref{algo: randomized trial obs} of Section \ref{sec: data fusion}, also identify $g^{\textbf{sup}}$.

   \subsection{Trial with recordings of natural treatment values}
    We introduce a third design, a sequentially randomized trial where we record the natural treatment values at every time-point (Data Structure \ref{algo: recorded nat value trial}). This trial was previously suggested by \citet{Forney_Bareinboim_2019} in single-time-point settings, but can be extended to multiple time-points. The following proposition follows immediately from the description of Data Structure \ref{algo: recorded nat value trial} and the definition of $g^{\textbf{init}}$.
    \begin{proposition}
    Under Assumptions \ref{ass: consistency} and \ref{ass: positivity}, Data Structure \ref{algo: recorded nat value trial}, represented in Figure \ref{fig: recorded nat val}, identifies $g^{\textbf{obs}}$, $g^{\textbf{opt}}$, $g^{\textbf{init}}$ and $g^{\textbf{sup}}$.
    \label{prop: recorded nat value trial id}
    \end{proposition}

        \begin{figure}
        \centering
        \begin{tikzpicture}
            \node (start) at (0,0) {Record $A_1^g$};
            \node (A1=0) at (-2,-2) {Record $A_2^g$};
            \node (A1=1) at (2,-2) {Record $A_2^{g}$};
            \node (3DL) at (-3,-4) {$\cdots$};
            \node (3DR) at (3,-4) {$\cdots$};
            \node (3DML) at (-1,-4) {$\cdots$};
            \node (3DMR) at (1,-4) {$\cdots$};

            \draw[-] (start) -- node[rotate = 45, above, scale = .85] {$A_1^{g+} = 0$} (A1=0);
            \draw[-] (start) -- node[rotate = 315, above, scale = .85] {$A_1^{g+} = 1$} (A1=1);
            \draw[-] (A1=0) -- node[rotate = 60, above, scale = .85] {$A_2^{g+} = 0$} (3DL);
            \draw[-] (A1=0) -- node[rotate = 300, above, scale = .85] {$A_2^{g+} = 1$} (3DML);
            \draw[-] (A1=1) -- node[rotate = 60, above, scale = .85] {$A_2^{g+} = 0$} (3DMR);
            \draw[-] (A1=1) --  node[rotate = 300, above, scale = .85] {$A_2^{g+} = 1$} (3DR);
        \end{tikzpicture}
        \caption{Tree describing Data Structure \ref{algo: recorded nat value trial}.}
        \label{fig: recorded nat val}
    \end{figure}
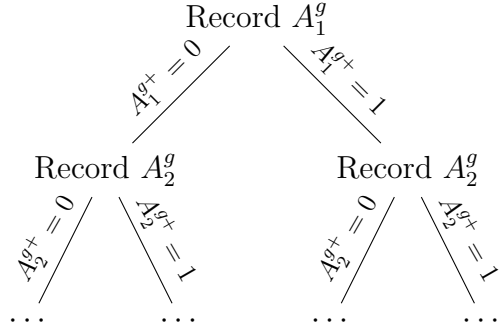

     The $H$-optimal regime is identified using Equations \eqref{eq: id Bellman opt 1} and \eqref{eq: id Bellman opt 2}, and the initiation regime is identified using Equation \eqref{eq: init id formula} as in Section \ref{sec: sequ entry rand trial}. The superoptimal regime is identified by
     \begin{align}
         &g_K^{\textbf{sup}}(h_K, \overline{a}_K') = \argmax_{a_K \in \{0,1\}} \nonumber \\
         &\mathbb{E}[R_K | H_K = h_K, \overline{A}_{K}^{g^{\textbf{sup}}} = \overline{a}_K', \overline{A}_{K-1}^{g^{\textbf{sup}}+} = \overline{g}^{\textbf{sup}}_{K-1}(h_{K-1}, \overline{a}_{K-1}'), A_K^{g^{\textbf{sup}}+} = a_K] \label{eq: id Bellman sup 1}\\
         \nonumber \\
         &g_k^{\textbf{sup}}(h_k, \overline{a}_{k-1}') = \argmax_{a_k \in \{0,1\}} \nonumber \\
         & \mathbb{E}[R_k + \mathcal{V}^{\underline{g}_{k + 1}^{\textbf{sup}}}_{k+1}(H_{k+1}^{\underline{g}_{k + 1}^{\textbf{sup}}}, \underline{A}_{k+1}^{\underline{g}_{k + 1}^{\textbf{sup}}}) | H_k = h_k, \overline{A}_{k}^{g^{\textbf{sup}}} = \overline{a}_k', \overline{A}_{k-1}^{g^{\textbf{sup}}+} = \overline{g}^{\textbf{sup}}_{k-1}(h_{k-1}, \overline{a}_{k-1}'), A_k^{g^{\textbf{sup}}+} = a_k], \label{eq: id Bellman sup 2}
     \end{align}
     where $\overline{A}_K^{g^{\textbf{sup}}+}$ is observed in the trial arm and hence identified and $\mathcal{V}^{\underline{g}_{k + 1}^{\textbf{sup}}}_{k+1}(H_{k+1}^{\underline{g}_{k + 1}^{\textbf{sup}}}, \underline{A}_{k+1}'^{\underline{g}_{k + 1}^{\textbf{sup}}})$ is identified recursively using data following the regime $g^{\textbf{sup}}$.

    \begin{algorithm}
    \spacingset{1}
    \SetAlgoRefName{3}
    \SetAlgorithmName{Data Structure}
    \;
    
    \For{$k = 1,\ldots, K$}{
    \For{each individual in trial arm}{
    Record covariate value $L_k^g$\;
    Record natural treatment value $A_k^g$\;
    Randomize individual to receive treatment $A_k^{g+} = 1$ with fixed and known probability $p(H_k^g, \overline{A}_k^g) \in (0,1)$\;
    Record reward $R_k^g$\;
    }
    }
        \caption{Recorded natural treatment value trial.}
        \label{algo: recorded nat value trial}
    \end{algorithm}
    Data Structure \ref{algo: recorded nat value trial} has $2^K$ arms, the same number as the conventional sequentially randomized trial. 

   \subsection{Sequential three arm trial} 
    \label{sec: sequ 3-arm trial}
    We introduce a fourth design that randomizes units to three options at every time-point. The natural treatment value is recorded and units are then randomized to treatment ($A_k^{g+} = 1$), no-treatment ($A_k^{g+} = 0$), or natural treatment ($A_k^{g+} = A_k^g$). The following proposition follows immediately from the design of Data Structure \ref{algo: three choice trial } and the definition of $g^{\textbf{init}}$.
    \begin{proposition}
    Under Assumptions \ref{ass: consistency} and \ref{ass: positivity}, and using observed data generated as in Data Structure \ref{algo: three choice trial }, represented in Figure \ref{fig: three choice trial}, we can identify $g^{\textbf{obs}}$, $g^{\textbf{opt}}$, $g^{\textbf{init}}$ and $g^{\textbf{sup}}$.
    \label{prop: three choice trial id}
    \end{proposition}
    The $H$-optimal regime is identified using Equations \eqref{eq: id Bellman opt 1} and \eqref{eq: id Bellman opt 2}, the optimal initiation regime is identified using Equation \eqref{eq: init id formula} and the superoptimal regime is identified using Equations \eqref{eq: id Bellman sup 1} and \eqref{eq: id Bellman sup 2}.
    
    \begin{algorithm}
    \spacingset{1}
    \SetAlgoRefName{4}
    \SetAlgorithmName{Data Structure}
    \;
    
    \For{$k = 1,\ldots, K$}{
    \For{each individual in trial arm}{
    Record covariate value $L_k^g$\;
    Randomize individual to receive treatment $A_k^{g+} \in \{0,1, A_k^g \}$ with fixed and known multinomial distribution that can depend on $H_k^g$ and $k$\;
    Record reward $R_k^g$\;
    }
    }
    \caption{Sequential three arm trial.}
    \label{algo: three choice trial }
    \end{algorithm}

    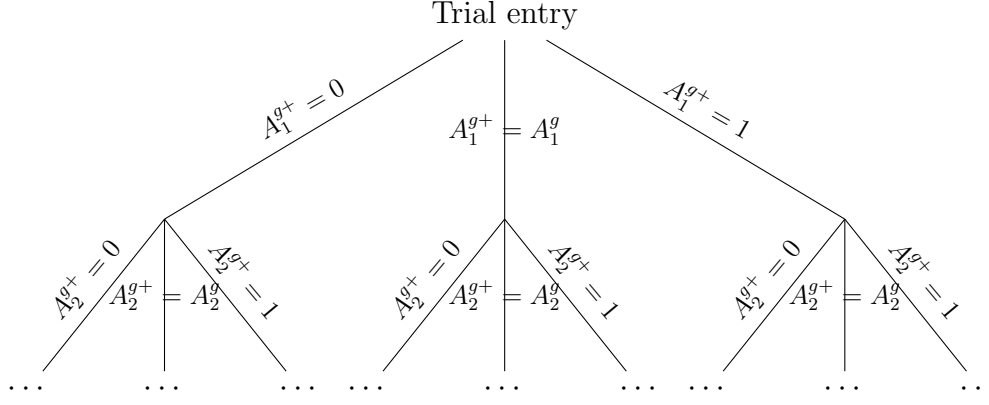
\begin{figure}
        \centering
        \begin{tikzpicture}[scale = .9]
            \node (start) at (0,0) {Trial entry};
            
            \node (3DL) at (-7,-5.5) {$\cdots$};
            \node (3DMML) at (-5,-5.5) {$\cdots$};
            \node (3DML) at (-3,-5.5) {$\cdots$};

            \node (3DLM) at (-2,-5.5) {$\cdots$};
            \node (3DMM) at (0, -5.5) {$\cdots$};
            \node (3DRM) at (2, -5.5) {$\cdots$};
            
            \node (3DMR) at (3,-5.5) {$\cdots$};
            \node (3DMMR) at (5,-5.5) {$\cdots$};
            \node (3DR) at (7,-5.5) {$\cdots$};

            \draw[-] (start) -- node[rotate = 30, above, scale = .85] {$A_1^{g+} = 0$} (-5,-3);
            \draw[-] (start) -- node[scale = .85] {$A_1^{g+} = A_1^g$} (0,-3);
            \draw[-] (start) -- node[rotate = 330, above, scale = .85] {$A_1^{g+} = 1$} (5,-3);
            
            \draw[-] (-5,-3) -- node[rotate = 50, above, scale = .85]{$A_2^{g+} = 0$} (3DL);
            \draw[-] (-5,-3) -- node[scale = .85]{$A_2^{g+} = A_2^{g}$} (3DMML);
            \draw[-] (-5,-3) -- node[rotate = 310, above, scale = .85] {$A_2^{g+} = 1$} (3DML);

            \draw[-] (0,-3) -- node[rotate = 50, above, scale = .85]{$A_2^{g+} = 0$} (3DLM);
            \draw[-] (0,-3) -- node[scale = .85]{$A_2^{g+} = A_2^{g}$} (3DMM);
            \draw[-] (0,-3) -- node[rotate = 310, above, scale = .85] {$A_2^{g+} = 1$} (3DRM);

            \draw[-] (5,-3) -- node[rotate = 50, above, scale = .85]{$A_2^{g+} = 0$} (3DMR);
            \draw[-] (5,-3) -- node[scale = .85]{$A_2^{g+} = A_2^{g}$} (3DMMR);
            \draw[-] (5,-3) -- node[rotate = 310, above, scale = .85] {$A_2^{g+} = 1$} (3DR);
        \end{tikzpicture}
        \caption{Tree describing Data Structure \ref{algo: three choice trial }.}
        \label{fig: three choice trial}
    \end{figure}
Data Structure \ref{algo: three choice trial } has $3^K$ trial arms, significantly more than the $2^K$ of the conventional sequentially randomized trial.

\section{Identification assumptions in experimental and observational data}
\label{app: other id ass}

Existing methods to identify the $H$-optimal regime often focus on identification of the value functions $\mathbb{E}[R_k^{\overline{a}_{k - 1}, a_k} | H_k^{\overline{a}_{k-1}}]$, or the contrast or ``blip" functions $\mathbb{E}[R_k^{\overline{a}_{k - 1}, 1} | H_k^{\overline{a}_{k-1}}] - \mathbb{E}[R_k^{\overline{a}_{k - 1}, 0} | H_k^{\overline{a}_{k-1}}]$. Sequentially randomized trials allow for identification of the $H$-optimal regime by design. 

However, sequentially randomized trial data can be unavailable for several reasons, e.g., because trials are costly, time-consuming and unethical. Additionally, it is often infeasible to conduct trials on populations large enough to detect heterogeneous treatment effects \citep{brookes2004subgroup, schmidt2014exploring, dahabreh2016using}. Thus, a substantial body of work concerns identification and estimation of $H$-optimal regimes from observational data \citep{murphy2003optimal,robins2004optimal,moodie2007demystifying,moodie2012q,chakraborty2013statistical,schulte2014q_and_a,chakraborty2014dynamic,luedtke2016super,tsiatis2019dynamic, zhang2019near,clifton2020qlearning,wang_provably_2021,chen_estimating_2023}. However, unmeasured confounding might complicate causal inference from observational data. Furthermore, as shown in Section \ref{sec: data fusion}, observations of the natural treatment regime are not sufficient for identification of $g^{\textbf{init}}$, $g^{\textbf{sup}}$, or $g^{\textbf{opt}}$ if the observational data do not verify Assumption \ref{ass: sequential exch}; additional assumptions are needed. Here, we consider some identification strategies for initiation regimes under different assumptions and data structures. 

\subsection{Unconfoundedness assumption}
Most existing results on $H$-optimal regimes require \emph{sequential exchangeability}, often called unconfoundedness:
\begin{assumption}[Sequential exchangeability]
    For $j = 1, \ldots, K$,
    \begin{align*}
        (R_j^{a_1, \ldots, a_j}, (R_{k}^{a_1, \ldots, a_k}, A_{k}^{a_1, \ldots a_{k-1}}, L_{k}^{a_1, \ldots, a_{k-1}})_{k = j + 1, \ldots, K}) \independent A_j^{a_1, \ldots, a_{j-1}} | H_j^{a_1, \ldots, a_{j-1}}.
    \end{align*}
    \label{ass: sequential exch}
\end{assumption}
Under Assumption \ref{ass: sequential exch}, which can be enforced by design in sequentially randomized trials, the following value functions are identified:
\begin{align*}
    \mathbb{E}[R_k^{\overline{a}_{k - 1}, a_k} | H_k^{\overline{a}_{k-1}} = (\overline{r}_{k-1}, \overline{l}_k)] = \mathbb{E}[R_k | \overline{L}_k = \overline{l}_k, \overline{R}_{k-1} = \overline{r}_{k-1}, \overline{A}_k = (\overline{a}_{k-1}, a_k)] \text{ for all $k$.}
\end{align*}
While unconfoundedness can be enforced by design in sequentially randomized trials, it might fail in studies of observational data.

\begin{remark}
Assumption \ref{ass: sequential exch} is often implicit in the reinforcement learning literature when, e.g., applying Q-Learning or A-Learning algorithms to derive the $H$-optimal regime \citep{murphy2003optimal}. However, Assumption \ref{ass: sequential exch} can fail in settings where reinforcement learning methods are used, e.g., due to unmeasured confounding. In particular, Q-Learning algorithms sometimes use data from human decision makers to improve their performance; for example, reinforcement-learning-based driving algorithms that use human driving instincts as inputs can outperform fully autonomous driving systems to improve driving performance \citep{wu2021human,wu2022prioritized, wu2023toward}, suggesting that there are unmeasured confounders in the observational data used for training \citep{stensrud_optimal_2024}. Indeed, under Assumption \ref{ass: sequential exch}, $\mathbb{E}[R_k^{\overline{a}_k} | H_k^{\overline{a}_{k-1}}, \overline{A}_k^{\overline{a}_{k-1}}] = \mathbb{E}[R_k^{\overline{a}_k} | H_k^{\overline{a}_{k-1}}]$, hence the natural treatment values $\overline{A}_k^g$ are not informative to the value function, $g^{\textbf{opt}} = g^{\textbf{init}} = g^{\textbf{sup}}$, which implies that the natural treatment values, i.e. the expert opinions, do not improve decision-making.    
\end{remark} 

\begin{remark}
    Sequentially randomized trials have one less arm than Data Structure \ref{algo: obs + trial }, and hence do not identify $g^{\textbf{init}}$ and $g^{\textbf{sup}}$. If Assumption \ref{ass: sequential exch} holds in observational data, $g^{\textbf{init}}$ and $g^{\textbf{sup}}$ are identified, but $g^{\textbf{opt}} = g^{\textbf{init}} = g^{\textbf{sup}}$.
\end{remark}

\subsection{Markov assumptions}
Markov-type assumptions are commonly used assumptions in the reinforcement learning literature \citep{monahan_state_1982,sutton1999reinforcement, kober_reinforcement_2013,kallus2020confounding,wang_provably_2021,shi_off_policy_2024}. These assumptions reduce time-varying decision making to a sequence of point treatment decisions that can be analyzed separately, see for example \citep{bellman_mdp_1957, howard1960dynamic,sutton1999reinforcement,wang_provably_2021, shi_off_policy_2024}.
\begin{assumption}[Markov]
For $k = 2, \ldots, K$
    \begin{align*}
        &(R_k^{\overline{a}_k}, A_k^{\overline{a}_{k-1}}, L_k^{\overline{a}_{k-1}}, U_k^{\overline{a}_{k-1}}) \independent (R_j^{\overline{a}_j}, A_j^{\overline{a}_{j-1}}, L_j^{\overline{a}_{j-1}}, U_j^{\overline{a}_{j-1}})_{j = 1}^{k-1}, \text{ and }\\
        &(R_k^{\overline{a}_k}, A_k^{\overline{a}_{k-1}}, L_k^{\overline{a}_{k-1}}, U_k^{\overline{a}_{k-1}}) = (R_k^{a_k}, A_k, L_k, U_k).
    \end{align*}
    \label{ass: Markov general}
\end{assumption}
For the example of Section \ref{sec: chronic back pain ex intro}, Assumption \ref{ass: Markov general} implies that the treatment chosen at the second time-point only depends on the medical characteristics used at time $k = 2$, that is, the opioid usage and depression statuses, and response to treatment at baseline. 

Under Assumption \ref{ass: Markov general}, the data-generating mechanism is reduced to $K$ point treatment data-generating mechanisms $(L_i, A_i, R_i)_{i = 1}^K$.

Under Assumption \ref{ass: Markov general}, 
\begin{align*}
    \mathbb{E}[R_k^{\overline{a}_k}|H_k^{\overline{a}_{k-1}}] = \mathbb{E}[R_k^{a_k} | L_k]
\end{align*}
hence reducing the problem to $K$ point treatment settings.

When Assumption \ref{ass: Markov general} is reasonable, our exposition clarifies that existing results on point or partial identification \citep{stensrud_optimal_2024, laurendeau2024improved} can be directly applied to improve $H$-optimal regimes. To be explicit, let $\Tilde{g}_k^{\textbf{sup}}$ be the $(L_k,A_k)$-optimal regime, that is the superoptimal regime in the point treatment setting, of the data generating mechanism given by $(L_k,U_k,A_k,R_k^0, R_k^1)$,
    \begin{align*}
        \Tilde{g}_k^{\textbf{sup}}(L,A) := \argmax_{a \in \{0,1\}} \mathbb{E}[R_k^a | L_k,A_k].
    \end{align*}
Then the following proposition holds:
\begin{proposition}
    Suppose Assumption \ref{ass: Markov general} holds and that  $\mathbb{E}[R_k^{\overline{a}_k} | H_k^{\overline{a}_{k-1}}]$ is identified for all $\overline{a}_k$. Then, the regime $\underline{g}^{\textbf{sup}}$ defined by
    \begin{align*}
        \underline{g}^{\textbf{sup}}_k(H_k^g, \overline{A}_k^{g^{\textbf{sup}}}) := \Tilde{g}_k^{\textbf{sup}}( L_k^{\Tilde{g}^{\textbf{sup}}}, A_k^{\Tilde{g}^{\textbf{sup}}})
    \end{align*}
    is identified and equal to the time-varying superoptimal regime $g^{\textbf{sup}}$.
    \label{prop: MDP sup better than opt}
\end{proposition}
See Definition \ref{def: superoptimal regime} for a formal definition of $g^{\textbf{sup}}$.

\subsection{A weaker Markov assumption: The Forgetfulness assumption}
\label{sec: Forgetfulness assumption}

To address limitations of Assumption \ref{ass: Markov general}, several relaxations have been used in the MDP literature \citep{bellman_mdp_1957, howard1960dynamic, sutton1999reinforcement, kallus2020confounding,wang_provably_2021, shi_off_policy_2024}. For example, \citet{wang_provably_2021, shi_off_policy_2024} made the following assumption.
\begin{assumption}[\cite{wang_provably_2021, shi_off_policy_2024}]
    For all $k > 1$, $(U_k, A_k, R_k) \independent \{(U_j, A_j, L_j, R_j)\}_{j = 1}^{k -1} | L_k$.
    \label{ass: shi mdp}
\end{assumption}
 
 We introduce a memoryless unmeasured confounding, or ``forgetfulness", assumption, a relaxation of Assumptions \ref{ass: Markov general} and \ref{ass: shi mdp}, that allows us to identify $g^{\textbf{sup}}$ by assuming that the expected outcome at time $k$ only varies with treatment $A_k$, conditional on past history.
 \begin{assumption}[Forgetfulness]
For all $k$ and $\overline{a}_k \in \{0,1\}^k$, 
    \begin{align*}
        \mathbb{E}[R_k^{\overline{a}_k} |H_k^{\overline{a}_{k-1}}, A_k^{\overline{a}_{k-1}}] = \mathbb{E}[R_k^{\overline{a}_k} |H_k^{\overline{a}_{k-1}}, \overline{A}_k^{\overline{a}_{k-1}}] \text{ a.s.,}
    \end{align*}
    \label{Ass: Forgetfulness}
\end{assumption}

\begin{proposition}
    Assumption \ref{ass: shi mdp} and $R_k^g \independent \overline{A}_k^g | H_k^g, U_k$ implies Assumption \ref{Ass: Forgetfulness}.
    \label{prop: shi mdp implies forgetfulness}
\end{proposition}

Assumption \ref{Ass: Forgetfulness} fails when there are unmeasured confounders $U$ such that $\mathbb{E}[R_k^g | H_k^g, A_k^g, U = u_1] \neq \mathbb{E}[R_k^g | H_k^g, A_k^g, U = u_2]$ for some $u_1 \neq u_2$ and $\text{Cov}(U, A_j^g | H_k^g) \neq 0$ for some $j < k$, which is also ruled out by assumptions on MDPs. These assumptions essentially aim to reduce longitudinal data to a sequence of observations of similar data-generating mechanisms that are linked in a manageable manner for identification. 


Now, we can define the following regime $g^{\textbf{s}}$,
\begin{align*}
    g_K^{\textbf{s}}(h_K, \overline{a}_{K-1}, a_K') &= \argmax_{a_K \in \{0,1\}} \mathbb{E}[R_K^{\overline{a}_{K-1}, a_K} | H_K^{\overline{a}_{K-1}} = h_K, A_K^{\overline{a}_{K-1}} = a_K'],\\
    g^{\textbf{s}}_k(h_k, \overline{a}_{k-1}, a_k') &= \argmax_{a_k\in \{0,1\}} \mathbb{E}[R_k^{\overline{a}_{k-1}, a_k} + \mathcal{V}_{P,k+1}^{\overline{a}_k}(H_{k+1}^{\overline{a}_{k}}, \overline{A}_{k + 1}^{\overline{a}_k})|H_k^{\overline{a}_{k-1}} = h_k, A_k^{\overline{a}_{k-1}} = a_k']
\end{align*}
with associated value functions:
\begin{align*}
     \mathcal{V}_{P,K}^{\overline{a}_{K-1}}(H_K^{\overline{a}_{K-1}}, \overline{A}_K^{\overline{a}_{K-1}})&= \max_{a_K}  \mathbb{E}[R_K^{\overline{a}_{K-1}, a_K} | H_K^{\overline{a}_{K-1}}, A_K^{\overline{a}_{K-1}}],\\
     \mathcal{V}_{P,k}^{\overline{a}_{k-1}}(H_k^{\overline{a}_{k-1}}, \overline{A}_k^{\overline{a}_{k-1}}) &= \max_{a_k} \mathbb{E}[R_k^{\overline{a}_{k-1}, a_k} + \mathcal{V}_{P,k+1}^{\overline{a}_k}(H_{k+1}^{\overline{a}_k}, \overline{A}_{k + 1}^{\overline{a}_k})|H_k^{\overline{a}_{k-1}}, A_k^{\overline{a}_{k-1}}].
\end{align*}

Although Assumption \ref{Ass: Forgetfulness} is weaker than Assumption \ref{ass: Markov general}, both assumptions fail when treatments are confounded across time-points in observational data. 

Furthermore, Assumption \ref{Ass: Forgetfulness} is implied by
\begin{equation}
    R_k^g \independent \overline{A}_{k-1}^g | H_k^g, A_k^g,
    \label{eq: Forgetfulness independence}
\end{equation}
which can be read off Single World Intervention Graphs (SWIGs). 

Equation \eqref{eq: Forgetfulness independence} rules out any unmeasured confounding between $A_k^g$ and $\overline{A}_{k-1}^g$. If such unmeasured confounding exists, Assumption \ref{Ass: Forgetfulness} will fail if the $U_k$s verify $\mathbb{E}[R_k^g | H_k^g, A_k^g, U_k = u_1] \neq \mathbb{E}[R_k^g | H_k^g, A_k^g, U_k = u_2]$ for some $u_1 \neq u_2$ and $\text{Cov}(U_k, A_j^g) \neq 0$ for some $j < k$, see Figure \ref{fig: Graph implications of forgetfulness assumption}. 

\begin{figure}
    \centering
    \begin{tikzpicture}[scale = 0.8]
                \tikzset{line width=1.5pt, outer sep=0pt,
                ell/.style={draw,fill=white, inner sep=2pt,
                line width=1.5pt},
                swig vsplit={gap=5pt,
                inner line width right=0.5pt},
                swig hsplit={gap=5pt}
                };
                      \node[name=L1,ell,shape=ellipse] at (6,-2){$L_{k-1}^g$};
                    \node[name=A1,shape=swig vsplit] at (3,0){
                                                      \nodepart{left}{$A_{k-1}^{g}$}
                                                      \nodepart{right}{$A_{k-1}^{g+}$} };
                    \node[name=Y1, ell, shape=ellipse] at (9,0){$R_{k-1}^g$};
                    \node[name=U1, ell, shape = ellipse] at (6,2){$U_{k-1}^g$};
                    \node[name=U2, ell, shape = ellipse] at (15,2){$U_{k}^g$};
                    \node[name=A2, shape = swig vsplit] at (12,0){ 
                                                      \nodepart{left}{$A_{k}^{g}$}
                                                      \nodepart{right}{$A_{k}^{g+}$} };
                    \node[name=Y2, ell, shape = ellipse] at (18,0){$R_k^g$};
                    \node[name=L2, ell, shape = ellipse] at (15,-2){$L_k^g$};
                \begin{scope}[>={Stealth},
                  every node/.style={fill=white,circle},
                  every edge/.style={draw=gray}]
                    \path[->] (U1) edge (Y1);
                    \path[->] (U1) edge (A1.135);
                    \path[->] (L1) edge (A1.225);
                    \path[->] (A1) edge (Y1);
                    \path [->] (L1) edge (Y1);
                    \path[->] (Y1) edge[bend left = 15] (Y2);
                    \path[->] (L1) edge (Y2);
                    \path[->] (L2) edge (A2.225);
                    \path[->] (A2) edge (Y2);
                    \path[->] (U2) edge (A2.135);
                    \path[->] (U2) edge (Y2);
                    \path[->] (L2) edge (Y2);
                    \path[->] (Y1) edge (A2);
                    \path[->] (A1) edge[bend left = 15] (A2);
                    \path[->] (L1) edge (L2);
                    \path[->] (L1) edge (A2);
                    \path[->] (A1) edge (L2);
                    \path[->] (Y1) edge (L2);
                    \path[->] (A1) edge[bend left = 15] (Y2);
                    \path[->] (U1) edge[dotted, color = red] (U2);
                    \path[->] (A1) edge (U2);
                    \path[->] (Y1) edge (U2);
                \end{scope}
            \end{tikzpicture}
    \caption{Illustration of the edges not allowed by Assumption \ref{Ass: Forgetfulness} by considering time-point $k$ and $k-1$. The red dotted edges violate Equation \eqref{eq: Forgetfulness independence} under a faithfulness assumption.}
    \label{fig: Graph implications of forgetfulness assumption}
\end{figure}
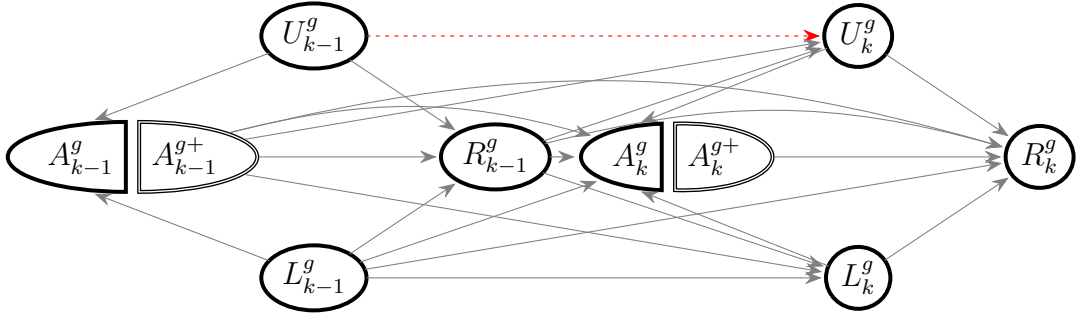

Assumption \ref{Ass: Forgetfulness} restricts the conditional mean of $R_k^{\overline{a}_k}$ only. We need the following assumption to evaluate the effect of treatment $a_k$ on future outcomes.
\begin{assumption}[Future forgetfulness]
For all $k = 2, \ldots, K$ and $\overline{a}_k \in \{0,1\}^k$,
    \begin{align*}
        (R_k^{\overline{a}_k}, L_{k+1}^{\overline{a}_k}, A_{k+1}^{\overline{a}_k}) \independent \overline{A}_{k-1}^{\overline{a}_{k-2}} \mid H_k^{\overline{a}_{k-1}}, A_k^{\overline{a}_{k-1}},
    \end{align*}
    where $A_{K+1}^{\overline{a}_K}$ and $L_{K + 1}^{\overline{a}_K}$ are empty.
    \label{Ass: Forgetfulness future}
\end{assumption}
Assumption \ref{Ass: Forgetfulness future} implies Assumption \ref{Ass: Forgetfulness}.
\begin{proposition}
    Under Assumption \ref{Ass: Forgetfulness future}, $g^{\textbf{s}}$ is the superoptimal regime. Furthermore, if $\mathbb{E}[R_k^{\overline{a}_k} +  \mathcal{V}_{P,k+1}^{\overline{a}_k}(H_{k+1}^{\overline{a}_k}, \overline{A}_{k + 1}^{\overline{a}_k})| \overline{L}_k, \overline{R}_{k-1}, \overline{A}_{k-1}]$ is identified, $g^{\textbf{sup}} = g^{\textbf{s}}$ is identified.
    \label{prop: forgetfulness superopt identification}
\end{proposition}

\subsection{Markov Decision Processes}
In the reinforcement learning literature, it is often assumed that the data can be modeled as a  Markov Decision Process (MDP): 
\begin{definition}
A process verifying Assumptions \ref{ass: sequential exch} and \ref{ass: Markov general} is called a \emph{Markov Decision Process} (MDP) \citep{bellman_mdp_1957, howard1960dynamic, sutton1999reinforcement}.
\end{definition}
Assumptions \ref{ass: sequential exch} and \ref{ass: Markov general} are restrictive and implausible in many applications. Consider our running example on back pain, Assumption \ref{ass: sequential exch} rules out that doctors have access to outcome-relevant unobserved or unmeasured information not encoded in the data, including visual cues about the pain levels of the patients and additional medical history. Furthermore, Assumption \ref{ass: Markov general} assumes that the first treatment decision, outcome and covariates do not affect the treatment, outcome and covariates at the second time-point. For example, Assumption \ref{ass: Markov general} implies that the opioid usage and depression status are independent at baseline and 3 months later, which is seemingly unlikely given the persistence of depression and opioid addiction. Assumption \ref{ass: Markov general} also requires that patients treated with Duloxetine, an antidepressant, do not experience a change in depression status except through changes in their response to treatment. Although Assumption \ref{ass: Markov general} can be modified to use historical information \citep{sutton1999reinforcement}, it is often implicitly assumed when using the results of MDPs in black-box procedures.

In their widely cited book, \citet{sutton1999reinforcement} acknowledge some of the issues related to unused relevant information (violating Assumption \ref{ass: sequential exch}), and history (violating Assumption \ref{ass: Markov general}) when discussing optimal decision-making in poker: 

\begin{quote}
\emph{``Does Ellen like to bluff, or does she play conservatively? Does her face or demeanor provide clues to the strength of her hand? How does Joe’s play change when it is late at night, or when he has already won a lot of money? Although everything ever observed about the other players may have an
effect on the probabilities that they are holding various kinds of hands, in practice this is far too much to remember and analyze, and most of it will have no clear effect on one’s predictions and decisions. Very good poker players are
adept at remembering just the key clues, and at sizing up new players quickly, but no one remembers everything that is relevant. As a result, the state representations people use to make their poker decisions are undoubtedly non-Markov, and the decisions themselves are presumably imperfect. Nevertheless, people still make very good decisions in such tasks. We conclude that the
inability to have access to a perfect Markov state representation is probably not a severe problem for a reinforcement learning agent."} \citep[Example 3.6]{sutton1999reinforcement}
\end{quote}

\citet{sutton1999reinforcement} argue that the unused information in poker agents may not be a severe problem, and indeed very successful poker regimes based on MDPs that ignore player history and visual cues have been successful in poker tournaments \citep{moravcik_deepstack_2017}. However, as recognized by \citet{sutton1999reinforcement}, their decisions may not be optimal, which becomes a more severe problem in health and public policy settings, where using all available information to inform an optimal decision is more important.

Proposition \ref{prop: MDP sup better than opt} confirms that single-time-point methods can be extended to processes verifying Assumption \ref{ass: Markov general} to improve regimes, but Assumption \ref{ass: sequential exch} reduces to exchangeability for a point treatment setting when Assumption \ref{ass: Markov general} holds, $R^a \independent A |L$. This implies that $g^{\textbf{sup}} = g^{\textbf{init}} = g^{\textbf{opt}}$ \citep{stensrud_optimal_2024}.

However, there are many cases where unmeasured confounding is suspected, e.g. when human decision makers outperform reinforcement learning regimes. Thus, recent work has focused on extending the current literature on point-identification of value functions in single-time-point settings in combination with assumptions on the distribution of potential unmeasured confounders, see for example \citet{monahan_state_1982, zhang2019near, kallus2020confounding, joshi2024towards}. Moreover, it can be valuable to change the design to one of the alternative designs of Section \ref{sec: data fusion} to leverage natural treatment values in regimes; reinforcement learning regimes can learn from observing other agents interact in the same environment or solicit advice from human decision makers with access to unmeasured information to improve decision-making.

\subsection{Bounds}
\label{sec: Bounds}
When point-identification is infeasible, obtaining (sharp) bounds on $\mathbb{E}[R_k^g | H_k^g]$ has practical interest. In particular, we could maximize a chosen function of the bounds that reflects the decision maker's risk preferences. Let
\begin{align*}
    \mathbf{L}_k(H_k^g, a_k) \leq \mathbb{E}[R_k^{\overline{g}_{k-1}, a_k} | H_k^g] \leq \mathbf{U}_k(H_k^g, a_k).
\end{align*}
The \citet{balke_pearl_bounds_1997} or \citet{manski1998monotone} bounds are examples of $\mathbf{L}_k$ and $\mathbf{U}_k$ in instrumental variable settings with $K = 1$. Then, instead of maximizing $\mathbb{E}[R_k^g | H_k^g]$, which is unidentified, we can maximize a chosen utility function reflecting the planner's risk preferences. In particular, \citet{cui2021individualized} suggests using lower bounds of $\mathbb{E}[R_k^g | H_k^g]$ of the form
\begin{align*}
    \mathcal{V}_k(H_k^g,a_k)&:= \mathbb{E}\{(1-w_k(H_k^g))[(\mathbf{L}_k(H_k^g, 1) - \mathbf{U}_k(H_k^g, 0))I(a_k = 1) + \mathbf{L}_k(H_k^g, 0)]\\
    &+ w_k(H_k^g)[\mathbf{L}_k(H_k^g, 1) - (\mathbf{U}_k(H_k^g, 1) - \mathbf{L}_k(H_k^g, 0))I(a_k = 0)]\},
\end{align*}
where $0 \leq w_k(H_k^g) \leq 1$. As shown by \citet{cui2021individualized}, this includes the \citet{wald1949statistical} maximin criterion (pessimist), the \citet{savage1951theory} minimax regret criterion (opportunist), and the \citet{arrow1972optimality} criteria, if the appropriate values of $w_k(H_k^g)$ are chosen.

The results in our work can be extended to optimize functions of $\mathcal{V}_k$, instead of a partially identified value function $\mathbb{E}[R_k^{\overline{g}_{k-1}, a_k} | H_k^g]$ , similarly to results in \citet{chen_estimating_2023}. In particular, we can aim to maximize $\mathbb{E}[\sum_{k = 1}^K \mathcal{V}_k(H_k^g, g_k)]$, which is the aggregate utility function based on the choices of $w_k$. Similarly, we can extend results for the optimal initiation and superoptimal regimes by bounding their analogous value functions in the Bellman equations and choosing a preferred utility function. In particular, we can extend performance guarantees described in the point treatment setting by \citet{laurendeau2024improved}, see also related results \citep{ zhang2019near}. For example, if the planner chooses a utility function that corresponds to the conventionality criteria, the planner can be guaranteed to do no-worse than the previous standard-of-care $g^{\textbf{obs}}$.

\subsection{Case study: Instrumental variables for identification in observational data}
\label{sec: IV identification for obs data}

The results in Section \ref{sec: data fusion} rely on identification of counterfactual distributions using randomization, but hold more generally, provided an identification method is given to replace the randomization.

In this section, we illustrate how to identify $g^{\textbf{init}}$ and $g^{\textbf{sup}}$ using longitudinal instrumental variables (IVs) as in \citet{chen_estimating_2023}.

For this section, assume that in addition to $(\overline{U}_K,\overline{L}_K, \overline{A}_K, \overline{R}_K)$, we have access to $\overline{Z}_K$, which we assume are not directly caused by the other variables, that verify the  longitudinal IV assumptions (see e.g. \citet{swanson2018partial,chen_estimating_2023, michael2024instrumental}): 
\begin{assumption}[Longitudinal IV assumptions]
    The instrumental variables $\overline{Z}_K$ are valid instruments for the effect of $\overline{A}_K$ on $R$ if:
    \begin{enumerate}
        \item IV relevance: $\text{Cov} (Z_k, A_k | H_k, \overline{A}_{k-1}) \neq 0$ for all $k$.
        \item Exclusion-restriction: $R_k^{\overline{z}_k, \overline{a}_k} = R_k^{\overline{a}_k}$ for all $k, \overline{a}_k, \overline{z}_k$.
        \item IV unconfoundedness: $Z_k \independent (R_k^{\overline{z}_k, \overline{a}_k}, \overline{A}_k^{\overline{a}_{k-1}, \overline{z}_k}) | H_k^{\overline{z}_{k-1}, \overline{a}_{k-1}}, \overline{A}_{k-1}^{\overline{z}_{k-1}, \overline{a}_{k-2}}$. 
    \end{enumerate}
    \label{ass: IV assumptions}
\end{assumption}

Assumption \ref{ass: IV assumptions} is not sufficient for identification but with additional assumptions, which allow some heterogeneity necessary for the natural treatment values to improve the performance of the $H$-optimal regime \citep{cui2021necessary,stensrud_optimal_2024}.

Then, following \citet{chen_estimating_2023}, we have the following proposition:
\begin{proposition}
    Under Assumptions \ref{ass: consistency}, \ref{ass: positivity}, \ref{ass: IV assumptions} and an additional IV point-identification assumption, $g^{\textbf{opt}}$ and $g^{\textbf{init}}$ are identified.
    \label{prop: IV identification of opt and init}
\end{proposition}
Proposition \ref{prop: IV identification of opt and init} follows immediately from \citet{chen_estimating_2023} by adding natural treatment values as covariates until time $j$ where $g^{\textbf{init}}$ initiates the counterfactual regime.

Proposition \ref{prop: IV identification of opt and init} does not say that $g^{\textbf{sup}}$ is identified, and that is because additional assumptions are required to identify longitudinal regimes that leverage counterfactual natural treatment values \citep{richardson2013single, young2014identification}. We clarify in Appendix \ref{app: sequential IV} what these conditions are, and detail how to identify $g^{\textbf{sup}}$ using Bellman equations. Again, this clarifies that identification of $g^{\textbf{sup}}$ requires non-standard assumptions relative to $g^{\textbf{init}}$, which is clearly illustrated in the need for larger study designs in Section \ref{sec: data fusion}.

We detail how to identify $g^{\textbf{init}}$ with instrumental variables that satisfy Assumption \ref{ass: IV assumptions} and additional IV identification assumptions or IV bounds in Appendix \ref{app: Dirac identification}. Then, we can use influence-function based estimators as in Section \ref{app: general estimation result} and the modified Q-learning Algorithm \ref{algo: ssw Q-Learning estimation} in Appendix \ref{app: chen and zhang time-varing IV estimation} to estimate $g^{\textbf{init}}$ and $\mathbb{E}[R^{g^{\textbf{init}}}]$. As an illustrative example of the use of instrumental variables to identify and estimate $g^{\textbf{init}}$ and $\mathbb{E}[R^{g^{\textbf{init}}}]$, we give a simulation analysis in Appendix \ref{sec: sequential IV data analysis}.

\section{Bellman equation and $g^\textbf{sup}$}
\label{app: sequential IV}
Analogously to Lemma 1 in \citet{stensrud_optimal_2024}, we derive the following equality for $a_k' \neq a_k$:
\begin{align*}
    \mathbb{E}[R_k^{\overline{a}_k}| H_k^{\overline{a}_{k-1}}, A_k^{\overline{a}_{k-1}} = a_k'] = \frac{\mathbb{E}[R_k^{\overline{a}_k}| H_k^{\overline{a}_{k-1}}] - \mathbb{E}[R_k^{\overline{a}_{k-1}}I(A_k^{\overline{a}_{k-1}} = a_k)| H_k^g]}{P(A_k^{\overline{a}_{k-1}} = a_k' | H_k^{\overline{a}_{k-1}})}.
\end{align*}
 If we can identify $\mathbb{E}[R_k^{\overline{a}_{k-1}}I(A_k^{\overline{a}_{k-1}} = a_k)| H_k^{\overline{a}_{k-1}}]$ and $\mathbb{E}[I(A_k^{\overline{a}_{k-1}} = a_k') | H_k^{\overline{a}_{k-1}}]$, for example with some IV methods that have $R_k^{\overline{a}_{k-1}}I(A_k^{\overline{a}_{k-1}} = a_k)$ and $I(A_k^{\overline{a}_{k-1}} = a_k')$ as outcomes for treatment $A_k$, then we can leverage the natural treatment value $A_k$ for our regimes. 

 Indeed, 
 \begin{align*}
     \mathbb{E}[R_k^{\overline{a}_{k-1}}I(A_k^{\overline{a}_{k-1}} = a_k)| H_k^{\overline{a}_{k-1}}] &= \mathbb{E}[R_k^{\overline{a}_{k-1}}I(A_k^{\overline{a}_{k-1}} = a_k)| H_{k-1}, L_{k}^{a_{k-1}}, R_{k-1}^{a_{k-1}}],\\
     \mathbb{E}[I(A_k^{\overline{a}_{k-1}} = a_k')| H_k^{\overline{a}_{k -1}}] &= \mathbb{E}[I(A_k^{\overline{a}_{k-1}} = a_k')| H_{k-1}^{\overline{a}_{k-2}}, L_{k}^{a_{k-1}}, R_{k-1}^{a_{k-1}}].\\
 \end{align*}
 To clarify the identification structure, we represent variables at time $k$ and $k-1$ in a sequential IV setting using a SWIG in Figure \ref{fig:Sequential IV SWIG Id}.
 
 \begin{figure}
            \centering
            \begin{tikzpicture}[scale = 0.85, transform shape]
                \tikzset{line width=1.5pt, outer sep=0pt,
                ell/.style={draw,fill=white, inner sep=2pt,
                line width=1.5pt},
                swig vsplit={gap=5pt,
                inner line width right=0.5pt},
                swig hsplit={gap=5pt}
                };
                      \node[name=L,ell,shape=ellipse] at (6,-2){$H_{k-1}$};
                    \node[name=A,shape=swig vsplit] at (3,0) {
                                                      \nodepart{left}{$A_{k-1}^{\overline{a}_{k-2}}$}
                                                      \nodepart{right}{$a_{k-1}$} };
                    \node[name=Y, ell, shape=ellipse] at (9,0){$R_{k-1}^{\overline{a}_{k-1}}$};
                    \node[name=U, ell, shape=ellipse] at (10.5,2){$U$};
                    \node[name=Z, ell, shape=ellipse] at (3,-2){$Z_{k-1}$};
                    \node[name=A1, ell, shape = ellipse] at (12,0){$A_{k}^{\overline{a}_{k-1}}$};
                    \node[name=Z1, ell, shape = ellipse] at (12,-2){$Z_{k}$};
                    \node[name=Y1, ell, shape = ellipse] at (18,0){$R_k^{\overline{a}_{k-1}}$};
                    \node[name=L1, ell, shape = ellipse] at (15,-2){$L_{k}^{\overline{a}_{k-1}}$};
                \begin{scope}[>={Stealth[black]},
                  every node/.style={fill=white,circle},
                  every edge/.style={draw=black,very thick}]
                    \path[->] (U) edge (Y);
                    \path[->] (U) edge[bend right=5] (3,0.6);
                    \path[->] (L) edge (2.9,-0.6);
                    \path[->] (Z) edge (A);
                    \path[->] (A) edge (Y);
                    \path [->] (L) edge (Y);
                    \path[->] (Y) edge (A1);
                    \path[->] (Z1) edge (A1);
                    \path[->] (A1) edge (Y1);
                    \path[->] (U) edge (A1);
                    \path[->] (U) edge (Y1);
                    \path[->] (L1) edge (A1);
                    \path[->] (L1) edge (Y1);
                    \path[->] (A) edge[bend left=15] (A1);
                    \path[->] (A) edge[bend left=15] (Y1);
                    \path[->] (A) edge (L1);
                \end{scope}
            \end{tikzpicture}
            \caption{SWIG of variables at times $k-1$ and $k$.}
            \label{fig:Sequential IV SWIG Id}
\end{figure}
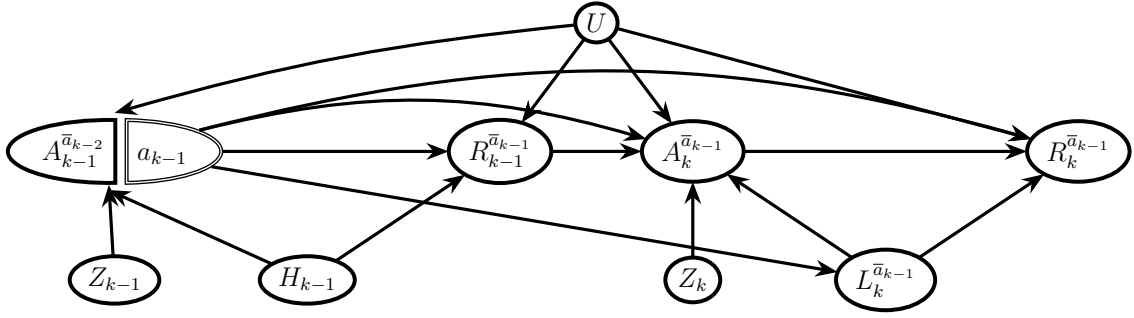
To identify $\mathbb{E}[R_k^{\overline{a}_{k-1}}I(A_k^{\overline{a}_{k-1}} = a_k)| H_k^{\overline{a}_{k-1}}]$ and $\mathbb{E}[I(A_k^{\overline{a}_{k-1}} = a_k') | H_k^{\overline{a}_{k-1}}]$ using IV methods, we need, at least, the following conditions to hold at time $k$:
\begin{enumerate}
    \item $\text{Cov}[Z_{k-1}, A_{k-1}^{\overline{a}_{k-2}} | H_{k-1}] \neq 0$, \label{IV cond 1}
    \item $(A_k^{\overline{a}_{k-1}, z_{k-1}}, R_k^{\overline{a}_{k-1}, z_{k-1}}) = (A_k^{\overline{a}_{k-1}}, R_k^{\overline{a}_{k-1}})$ for all $z_{k-1} \in \{0,1\}$, \label{IV cond 2}
    \item $Z_{k-1} \independent A_{k-1}^{\overline{a}_{k-2},z_{k-1}}, (A_k^{\overline{a}_{k-1}, z_{k-1}}, R_k^{\overline{a}_{k-1}, z_{k-1}}) | H_{k-1}$. \label{IV cond 3}
\end{enumerate}
Condition \ref{IV cond 1} is not excluded by the SWIG in Figure \ref{fig:Sequential IV SWIG Id}. Condition \ref{IV cond 2} holds by the structure of the SWIG when intervening on $Z_{k-1}$, see Figure \ref{fig:Sequential IV SWIG with Z_0 intervention}.
 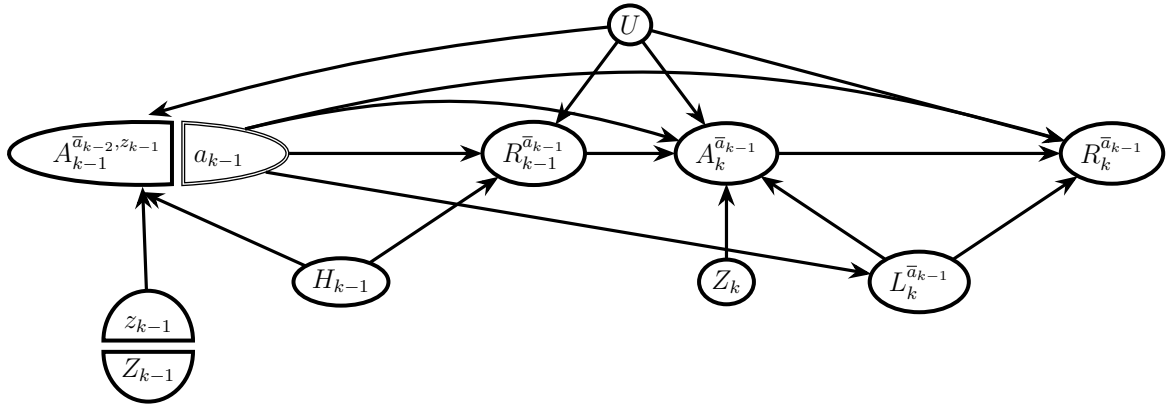
\begin{figure}
            \centering
            \begin{tikzpicture}[scale = 0.85, transform shape]
                \tikzset{line width=1.5pt, outer sep=0pt,
                ell/.style={draw,fill=white, inner sep=2pt,
                line width=1.5pt},
                swig vsplit={gap=5pt,
                inner line width right=0.5pt},
                swig hsplit={gap=5pt}
                };
                      \node[name=L,ell,shape=ellipse] at (6,-2){$H_{k-1}$};
                    \node[name=A,shape=swig vsplit] at (3,0) {
                                                      \nodepart{left}{$A_{k-1}^{\overline{a}_{k-2}, z_{k-1}}$}
                                                      \nodepart{right}{$a_{k-1}$} };
                    \node[name=Y, ell, shape=ellipse] at (9,0){$R_{k-1}^{\overline{a}_{k-1}}$};
                    \node[name=U, ell, shape=ellipse] at (10.5,2){$U$};
                    \node[name=Z, ell, shape= swig hsplit] at (3,-3){\nodepart{lower}{$Z_{k-1}$}
                                                      \nodepart{upper}{$z_{k-1}$} };
                    \node[name=A1, ell, shape = ellipse] at (12,0){$A_{k}^{\overline{a}_{k-1}}$};
                    \node[name=Z1, ell, shape = ellipse] at (12,-2){$Z_{k}$};
                    \node[name=Y1, ell, shape = ellipse] at (18,0){$R_k^{\overline{a}_{k-1}}$};
                    \node[name=L1, ell, shape = ellipse] at (15,-2){$L_{k}^{\overline{a}_{k-1}}$};
                \begin{scope}[>={Stealth[black]},
                  every node/.style={fill=white,circle},
                  every edge/.style={draw=black,very thick}]
                    \path[->] (U) edge (Y);
                    \path[->] (U) edge[bend right=5] (3,0.6);
                    \path[->] (L) edge (2.9,-0.6);
                    \path[->] (Z) edge (A);
                    \path[->] (A) edge (Y);
                    \path [->] (L) edge (Y);
                    \path[->] (Y) edge (A1);
                    \path[->] (Z1) edge (A1);
                    \path[->] (A1) edge (Y1);
                    \path[->] (U) edge (A1);
                    \path[->] (U) edge (Y1);
                    \path[->] (L1) edge (A1);
                    \path[->] (L1) edge (Y1);
                    \path[->] (A) edge[bend left=15] (A1);
                    \path[->] (A) edge[bend left=15] (Y1);
                    \path[->] (A) edge (L1);
                \end{scope}
            \end{tikzpicture}
            \caption{SWIG of variables relevant to sequential IV identification at step $k$, with intervention on $Z_{k-1}$.}
                            \label{fig:Sequential IV SWIG with Z_0 intervention}
\end{figure}
As $Z_{k-1}$ is at least exogenous conditional on $H_{k-1}$, Condition \ref{IV cond 3} also holds, see Figure \ref{fig:Sequential IV SWIG with Z_0 intervention}.

However, when we have a dynamic regime that depends on history and natural treatment values, the IV conditions change, as seen in Figure \ref{fig:Sequential IV SWIG with dynamic regime}.

\begin{figure}
            \centering
            \begin{tikzpicture}[scale = 0.85, transform shape]
                \tikzset{line width=1.5pt, outer sep=0pt,
                ell/.style={draw,fill=white, inner sep=2pt,
                line width=1.5pt},
                swig vsplit={gap=5pt,
                inner line width right=0.5pt},
                swig hsplit={gap=5pt}
                };
                      \node[name=L,ell,shape=ellipse] at (6,-2){$H_{k-1}$};
                    \node[name=A,shape=swig vsplit] at (3,0) {
                                                      \nodepart{left}{$A_{k-1}^{g, z_{k-1}}$}
                                                      \nodepart{right}{$A_{k-1}^{(g, z_{k-1})+}$} };
                    \node[name=Y, ell, shape=ellipse] at (10,0){$R_{k-1}^{g, z_{k-1}}$};
                    \node[name=U, ell, shape=ellipse] at (10.5,3){$U$};
                    \node[name=Z, ell, shape= swig hsplit] at (3,-3){\nodepart{lower}{$Z_{k-1}$}
                                                      \nodepart{upper}{$z_{k-1}$} };
                    \node[name=A1, ell, shape = ellipse] at (14,0){$A_{k}^{g, z_{k-1}}$};
                    \node[name=Z1, ell, shape = ellipse] at (14,-3){$Z_{k}$};
                    \node[name=Y1, ell, shape = ellipse] at (18,0){$R_k^{g, z_{k-1}}$};
                    \node[name=L1, ell, shape = ellipse] at (16,-2){$L_{k}^{g, z_{k-1}}$};
                \begin{scope}[>={Stealth[black]},
                  every node/.style={fill=white,circle},
                  every edge/.style={draw=black,very thick}]
                    \path[->] (U) edge (Y);
                    \path[->] (U) edge[bend right=5] (2.5,0.6);
                    \path[->] (L) edge (2.6,-0.6);
                    \path[->] (Z) edge (2.5,-0.6);
                    \path[->] (A) edge (Y);
                    \path [->] (L) edge (Y);
                    \path[->] (Y) edge (A1);
                    \path[->] (Z1) edge (A1);
                    \path[->] (A1) edge (Y1);
                    \path[->] (U) edge (A1);
                    \path[->] (U) edge (Y1);
                    \path[->] (L1) edge (A1);
                    \path[->] (L1) edge (Y1);
                    \path[->] (A) edge[bend left=15] (A1);
                    \path[->] (A) edge[bend left=15] (Y1);
                    \path[->] (A) edge (L1);
                    \path[->] (A.150) edge[color=blue,dotted,  out=120, in=130, bend left = 90] (A.60);
                    \path[->] (L) edge[color = blue, dotted] (A.300);
                \end{scope}
            \end{tikzpicture}
            \caption{SWIG of variables relevant to sequential IV identification at step $k$ with dynamic intervention regime $g$.}
                            \label{fig:Sequential IV SWIG with dynamic regime}
\end{figure}
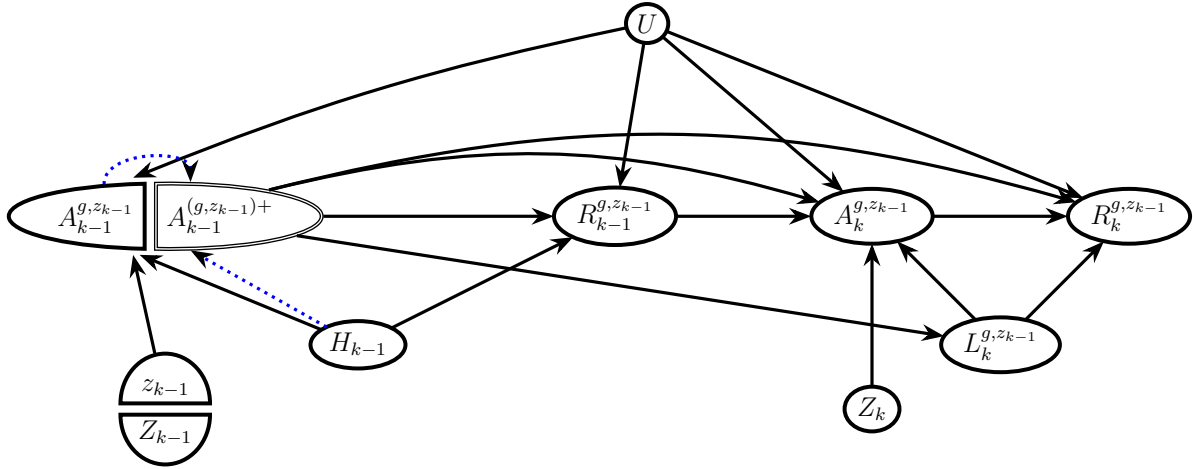
Then, the IV conditions are
\begin{enumerate}
    \item $\text{Cov}[Z_{k-1}, A_{k-1}^{g} | H_{k-1}, A_{k-1}^{g+}] \neq 0$, \label{dyn-IV cond 1}
    \item $(A_k^{g, z_{k-1}}, R_k^{g, z_{k-1}}) = (A_k^{g}, R_k^{g})$ for all $z_{k-1} \in \{0,1\}$, \label{dyn-IV cond 2}
    \item $Z_{k-1} \independent  (A_k^{g, z_{k-1}}, R_k^{g, z_{k-1}}) | H_{k-1}$. \label{dyn-IV cond 3}
\end{enumerate}
 $(A_k^{g, z_{k-1}}, R_k^{g, z_{k-1}}) = (A_k^{\overline{g}_{k-2}, A_{k-1}^{(g, z_{k-1})+}}, R_k^{\overline{g}_{k-2}, A_{k-1}^{(g, z_{k-1})+}})$ can depend on $z_{k-1}$, so Condition \ref{dyn-IV cond 2} can fail in the setting of Figure \ref{fig:Sequential IV SWIG with dynamic regime} where $g^{\textbf{sup}}$ may depend on current and future natural treatment values. 

 

\begin{remark}
    We need to use a more computationally intensive method to identify the superoptimal regime because $g^{\textbf{sup}}_k$ depends on future natural treatment values $A_{k +1}^{g^{\textbf{sup}}}, \ldots, A_{K}^{g^{\textbf{sup}}}$ for $k = 1, \ldots, K-1$, see Appendix \ref{app: Dirac identification}. This highlights the benefits of using initiation regimes such as $g^{\textbf{init}}$, which only uses future natural treatment values if it follows the identified observed regime.
\end{remark}

\section{Beyond binary treatments and regime identification in trials}
\label{app: multiple treatments}
We redefine the superoptimal regime $g^{\textbf{sup}}$ when treatments may take a finite, but larger than two, number of values. Let $\mathcal{A}_1, \ldots, \mathcal{A}_K$ be the supports of $A_1, \ldots, A_K$, respectively, and $|\mathcal{A}_k| < \infty$ for all $k = 1, \ldots, K$. The superoptimal regime $g^{\textbf{sup}}$ is then defined as 
\begin{align*}
    g^{\textbf{sup}} := \argmax_{g \text{ s.t. } g_k: \mathcal{H}_k \times \mathcal{A}_1 \times \cdots \times \mathcal{A}_k \to \mathcal{A}_k} \mathbb{E}[\sum_{k = 1}^K R_k^g].
\end{align*}
\begin{proposition}
    Data Structure \ref{algo: recorded nat value trial} identifies $g^{\textbf{sup}}$ for any $|\mathcal{A}_k| \geq 2$. However, Data Structures \ref{algo: obs + trial }, \ref{algo: randomized trial obs}, and \ref{algo: three choice trial } do not.
    \label{prop: multiple treatments}
\end{proposition}

The proofs of Propositions \ref{prop: randomized trial obs id}, \ref{prop: recorded nat value trial id} and
\ref{prop: three choice trial id} for binary treatments follow the same steps. In particular, Data Structure
\ref{algo: randomized trial obs} does not identify $g^{\textbf{sup}}$ for any $|\mathcal{A}_k| \geq 2$: natural
treatment values are not recorded after trial entry, so the sequential identification used in Appendix
\ref{app: Dirac identification} is not available.

\section{Example: \citet{chen_estimating_2023} time-varying IV}
\label{app: chen and zhang time-varing IV estimation}

As an explicit example, we consider the time-varying IV setting in \citet{chen_estimating_2023}. We show how to modify their algorithm to estimate the optimal initiation regime; estimation results of Section \ref{app: general estimation result} can then be applied to obtain convergence guarantees. \citet{chen_estimating_2023} defined the $H$-optimal regime as the regime minimizing the regret function
\begin{align*}
    &g_k^{\textbf{opt}}(H_k^g)=\argmin_{a_k \in \{0,1\}} |Q_k(H_k^g, 1) - Q_k(H_k^g,0)|I[\text{sign}(Q_k(H_k^g, 1) - Q_k(H_k^g,0)) \neq 2a_k - 1],
\end{align*}
where
\begin{align*}
    Q_k^{\overline{a}_{k-1}}(H_k^{\overline{a}_{k-1}}, a_k) := \mathbb{E}[R_k^{\overline{a}_{k -1}, a_k} + \mathcal{V}_{P,k + 1}^{\overline{a}_{k}}(H_k^{\overline{a}_{k-1}}, R_k^{\overline{a}_{k-1}, a_k}, L_{k+1}^{\overline{a}_{k-1}, a_k}) | H_k^{\overline{a}_{k-1}}]
\end{align*}
is the value function of intervention $a_k$ at time $k$.

Then, a modification of Q-learning can be applied to estimate the $H$-optimal regime, see Algorithm \ref{algo: Q-Learning estimation}.

\begin{algorithm}
\spacingset{1}
    \KwData{Observed data $(\overline{L}_K, \overline{A}_K, \overline{R}_K, \overline{Z}_K)$.}
    \KwResult{Estimated $H$-optimal regime $\hat{g}^{\textbf{opt}}$.}
    Estimate $Q_K$ using observations of $(H_K^g, R_K, Z_K)$\;
    Estimate $\hat{C}_K(H_K^g) := \hat{Q}_K(H_K^g, 1) - \hat{Q}_K(H_K^g, 0)$\;
    Set $\hat{g}^{\textbf{opt}}_K(H_K^g) := I(\hat{C}_K(H_K^g)>0)$\;
    Set $\hat{V}_K(H_K^g):= \hat{Q}_K(H_K^g, \hat{g}_K^{\textbf{opt}}(H_K^g))$\;
    \For{$k = K-1, \ldots, 1$}{
    Construct pseudo-outcomes $\hat{\text{PO}}_k := r_k + \hat{V}_{k + 1}(H_{k+1}^g)$\;
    Estimate $Q_k$ using observations of $(H_k^g, \hat{\text{PO}}_k, Z_k)$\;
    Estimate $\hat{C}_k(H_k^g) := \hat{Q}_k(H_k^g, 1) - \hat{Q}_k(H_k^g, 0)$\;
     Set $\hat{g}^{\textbf{opt}}_k(H_k^g) :=  I(\hat{C}_k(H_k^g)>0)$\;
    Set $\hat{V}_k(H_k^g):= \hat{Q}_k(H_k^g, \hat{g}_k^{\textbf{opt}}(H_k^g))$\;
    }
    Return $\hat{g}^{\textbf{opt}}$\;
    \caption{$Q$-learning estimation from \citet{chen_estimating_2023}}
    \label{algo: Q-Learning estimation}
\end{algorithm}

The algorithm suggested by \citet{chen_estimating_2023} needs to be adapted for estimation of $g^{\textbf{init}}$. As seen in Algorithm \ref{algo: superoptimal switching} in Appendix \ref{app: Bellman eqs}, $g^{\textbf{init}}$ is an optimal regime for modified value functions. Then, let $Q_k^{\textbf{init}}$ be defined as
\begin{align*}
    &Q_k^{\textbf{init}}(h_k, \overline{a}_k')\\
    &:= \max_{a_k \in \{0,1\}} \bigg\{I[\overline{A}_{k-1}^{g+} = \overline{a}_{k - 1}'] \\
    &\cdot\mathbb{E}[R_k^{a_k} + \mathcal{V}_{P,k+1}^g(H_{k+1}^g, I[a_k = A_k]\overline{A}_{k + 1} + I[a_k \neq A_k]\overline{A}_{k})| H_k^g = h_k, \overline{A}_k = \overline{a}_k']\\
    & + I[\overline{A}_{k-1}^{g+} \neq \overline{a}_{k - 1}']\\
    &\cdot\mathbb{E}[R_k^{A_{k-1}^{g+}, a_k} + \mathcal{V}_{P,k + 1}^g(H_{k+1}^g, \overline{A}_{\max\{j = 1, \ldots, k: \overline{A}_j^{g^+} = \overline{A}_j = \overline{a}_j'\} }) | H_k^g = h_k, \overline{A}_{\max\{j = 1, \ldots, k: \overline{A}_j^{g^+} = \overline{A}_j = \overline{a}_j'\} }] \bigg\}.
\end{align*}
An estimation procedure for $g^{\textbf{init}}$ can be derived from the algorithm in \citet{chen_estimating_2023}, see Algorithm \ref{algo: ssw Q-Learning estimation}.
\begin{algorithm}
\spacingset{1}
    \KwData{Observed data $(\overline{L}_K, \overline{A}_K, \overline{R}_K, \overline{Z}_K)$.}
    \KwResult{Estimated optimal initiation regime $\hat{g}^{\textbf{init}}$.}
    Estimate $Q_K^{\textbf{init}}$ using observations of $(H_K^g, R_K, \overline{A}_K, Z_K)$\;
    Estimate $\hat{C}_K^{\textbf{init}}(H_K^g, \overline{a}_K) := \hat{Q}_K^{\textbf{init}}(H_K^g, \overline{a}_K, 1) - \hat{Q}_K^{\textbf{init}}(H_K^g, \overline{a}_K, 0)$\;
    Set $\hat{g}^{\textbf{init}}_K(H_K^g, \overline{a}_K) := I(\hat{C}_K^{\textbf{init}}(H_K^g, \overline{a}_K) > 0)$\;
    Set $\hat{V}^{\textbf{init}}_K(H_K^g, \overline{a}_K):= \hat{Q}^{\textbf{init}}_K(H_K^g, \overline{a}_K, \hat{g}_K^{\textbf{init}}(H_K^g, \overline{a}_K))$\;
    \For{$k = K-1, \ldots, 1$}{
    Construct pseudo-outcomes $\hat{\text{PO}}_k := r_k + \hat{V}_{k + 1}^{\textbf{init}}(H_{k+1}^g, \overline{a}_{k+1})$\;
    Estimate $Q_k^{\textbf{init}}$ using observations of $(H_k^g, \hat{\text{PO}}_k, \overline{A}_k, Z_k)$\;
    Estimate $\hat{C}_k^{\textbf{init}}(H_k^g, \overline{a}_k) := \hat{Q}^{\textbf{init}}_k(H_k^g, \overline{a}_k, 1) - \hat{Q}^{\textbf{init}}_k(H_k^g, \overline{a}_k, 0)$\;
     Set $\hat{g}^{\textbf{init}}_k(H_k^g, \overline{a}_k) := I(\hat{C}^{\textbf{init}}_k(H_k^g, \overline{a}_k) > 0)$\;
    Set $\hat{V}^{\textbf{init}}_k(H_k^g, \overline{a}_k):= \hat{Q}^{\textbf{init}}_k(H_k^g, \overline{a}_k, \hat{g}_k^{\textbf{init}}(H_k^g, \overline{a}_k))$\;
    }
    Return $\hat{g}^{\textbf{init}}$\;
    \caption{$Q$-learning estimation for $g^{\textbf{init}}$.}
    \label{algo: ssw Q-Learning estimation}
\end{algorithm}

\subsection{Simulation of an IV setting}
\label{sec: sequential IV data analysis}

Using Assumption \ref{ass: IV assumptions}, we use \citet{manski1998monotone} at the first time point $k = 1$ and \citet{balke_pearl_bounds_1997} bounds at $k = 2 =:K$, and the IV-optimal decision criterion from \citet{chen_estimating_2023} as an outcome $R$. We use Algorithms \ref{algo: Q-Learning estimation} and \ref{algo: ssw Q-Learning estimation}, with influence-function based estimators for $\mathbb{E}[R_k^g | H_k^g, \overline{A}_j^g]$ as in Section \ref{app: general estimation result}, to estimate $g^{\textbf{opt}}$ and $g^{\textbf{init}}$, respectively. We subsequently estimate $\mathbb{E}[R^{\hat{g}^{\textbf{init}}}]$, $\mathbb{E}[R^{\hat{g}^{\textbf{opt}}}]$, and $\mathbb{E}[R^{\hat{g}^{\textbf{obs}}}]$ on $n_{\text{eval}}$ new samples empirically, keeping the estimated regimes $\hat{g}^{\textbf{init}}$ and $\hat{g}^{\textbf{opt}}$ fixed, and compute their respective $95\%$ Gaussian confidence intervals on the new samples.

We first start with the data-generating mechanism used by \citet[Section 6.1]{chen_estimating_2023}, where the observed covariates $L_1 \sim (\text{Bernoulli}(1/2), \text{Bernoulli}(1/2))$, the unobserved confounders $U_1, U_2 \sim \text{Rademacher}(1/2)$, the instrumental variables $Z_1, Z_2 \sim \text{Rademacher}(1/2)$, the treatments $A_i \sim \text{Rademacher}(\text{expit}(-3+ C(1+ Z_i)/2+ \xi(1+ U_i)/2))$ for $i = 1,2$, and outcomes $R_1 \sim \text{Bernoulli}(\text{expit}(0.2L_{1,1} + 0.05L_{1,2}(1+A_1) - \xi(1+U_1)/2))$ and $R_2 \sim \text{Bernoulli}(\text{expit}(0.2L_{1,1} + 0.05R_1(1+A_2) - \xi(1+U_2)/2))$ with $R = R_1 + R_2.$

The data-generating mechanism in the simulation section of \citet[Section 6.1]{chen_estimating_2023} does not have any interaction between $U$ and $A$ for the outcome $R$, hence the initiation and superoptimal regimes will match the optimal regime, as it outperforms the observed regime. For completeness, we include Table \ref{tab: chen and zhang sim results}, with the outcome of applying Algorithm \ref{algo: ssw Q-Learning estimation} for $g^{\textbf{init}}$ compared with $g^{\textbf{opt}}$ and $g^{\textbf{obs}}$.

\begin{table}
    \centering
    \spacingset{1}
    \begin{tabular}{c | c}
         \textbf{Regime} & \textbf{Expected Value and CI}\\
         \hline
         $g^{\textbf{obs}}$& $0.936$ $(0.934, 0.937)$\\
         $\hat{g}^{\textbf{opt}}$& $0.947$ $(0.946, 0.949)$\\
         $\hat{g}^{\textbf{init}}$& $0.946$ $(0.945, 0.947)$\\
    \end{tabular}
    \caption{Simulation results of implementation of Algorithms \ref{algo: Q-Learning estimation} and \ref{algo: ssw Q-Learning estimation} on $n_{\text{eval}} = 1'000'000$ samples for the \citet{chen_estimating_2023} data-generating mechanism.}
    \label{tab: chen and zhang sim results}
\end{table}

We also simulate a data-generating mechanism representing a sequential IV setting using a time-varying version of a simulation in \citet{laurendeau2024improved}, where there is interaction between the confounder and the natural treatment value for the outcome, see Figure \ref{fig:sim DAG} for a graphical representation.

The data-generating mechanism of the sequential IV simulation goes as follows.

Consider two time-points, $K = 2$, and the following distributions:
\begin{itemize}
    \item $U \sim \text{Bernoulli}(0.8)$,
    \item $Z_1,Z_2 \sim \text{Bernoulli}(0.5)$,
    \item $A_1 \sim Z_1\text{Bernoulli}(0.85U + 0.1) + (1-Z_1)\text{Bernoulli}(0.3U + 0.01)$,
    \item $R_1 \sim 2I(A_1\mathcal{N}(U-1,1) + (1-A_1)\mathcal{N}(1-U,1) > 0)$,
    \item $A_2 \sim Z_2\text{Bernoulli}(0.85U + 0.05(2A_1 - 1) + 0.05) + (1-Z_2)\text{Bernoulli}(0.1U + 0.01)$,
    \item $R_2 \sim I(\mathcal{N}(1-U + (1-U)(2A_1-1)Y_1 - (1-U)(2A_2-1)Y_1 + 2(2U-1)(1-2A_2),1) > 0)$.
\end{itemize}

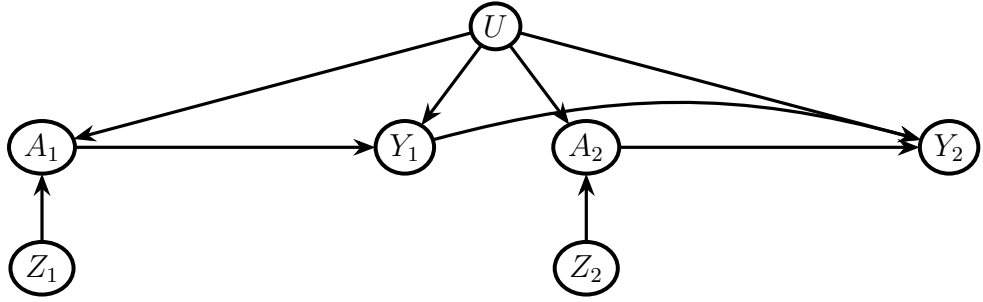
\begin{figure}
\centering
    \begin{tikzpicture}[scale = 0.8]
                \tikzset{line width=1.5pt, outer sep=0pt,
                ell/.style={draw,fill=white, inner sep=2pt,
                line width=1.5pt},
                swig vsplit={gap=5pt,
                inner line width right=0.5pt},
                swig hsplit={gap=5pt}
                };
                    \node[name=A1,ell,shape=ellipse] at (3,0){$A_1$};
                    \node[name=Y1, ell, shape=ellipse] at (9,0){$Y_1$};
                    \node[name=U, ell, shape=ellipse] at (10.5,2){$U$};
                    \node[name=Z1, ell, shape=ellipse] at (3,-2){$Z_{1}$};
                    \node[name=A2, ell, shape = ellipse] at (12,0){$A_2$};
                    \node[name=Z2, ell, shape = ellipse] at (12,-2){$Z_2$};
                    \node[name=Y2, ell, shape = ellipse] at (18,0){$Y_2$};
                \begin{scope}[>={Stealth[black]},
                  every node/.style={fill=white,circle},
                  every edge/.style={draw=black,very thick}]
                    \path[->] (U) edge (Y1);
                    \path[->] (U) edge (A1);
                    \path[->] (Z1) edge (A1);
                    \path[->] (A1) edge (Y1);
                    \path[->] (Y1) edge[bend left = 15] (Y2);
                    \path[->] (Z2) edge (A2);
                    \path[->] (A2) edge (Y2);
                    \path[->] (U) edge (A2);
                    \path[->] (U) edge (Y2);
                \end{scope}
            \end{tikzpicture}
    \caption{DAG representing the data-generating mechanism of the sequential IV simulation in Appendix \ref{sec: sequential IV data analysis}.}
    \label{fig:sim DAG}
\end{figure}
We simulated $n = 10'000$ samples from this data-generating mechanism, and evaluated the performance of the $H$-optimal regime and the optimal initiation regime $g^{\textbf{init}}$ on $n_{\text{eval}} = 1'000'000$ samples using Algorithms \ref{algo: Q-Learning estimation} and \ref{algo: ssw Q-Learning estimation}, respectively.

\begin{table}
    \centering
    \spacingset{1}
    \begin{tabular}{c | c}
         \textbf{Regime} & \textbf{Expected Value and CI}\\
         \hline
         $g^{\textbf{obs}}$& $1.546$ $(1.544, 1.548)$\\
         $\hat{g}^{\textbf{opt}}$& $1.703$ $(1.701, 1.705)$\\
         $\hat{g}^{\textbf{init}}$& $1.935$ $(1.934, 1.937)$\\
    \end{tabular}
    \caption{Simulation results of implementation of Algorithms \ref{algo: Q-Learning estimation} and \ref{algo: ssw Q-Learning estimation} on $n_{\text{eval}} = 1'000'000$ samples.}
    \label{tab: sim results}
\end{table}
The optimal initiation regime $g^{\textbf{init}}$ outperforms the optimal and observed regimes, see Table \ref{tab: sim results}.

\section{Identification of the superoptimal regime with binary treatments using Dirac's delta function}
\label{app: Dirac identification}
We show that $g^{\textbf{sup}}$ is identified when value functions are sequentially identified and treatments are binary, for example in the time-varying IV example of \citet{chen_estimating_2023}. We use Dirac delta function $\delta(x)$ to define new outcomes of interest for identification, which have the property that
\begin{align*}
    \mathbb{E}[\delta(X - a)] = f_X(a),
\end{align*}
where $f_X(a)$ is $X$'s density evaluated in $a$ \citep[p. 276]{tidemand-lichtenberg_nonlinear_2014}. 

Our arguments will hold for all $K \geq 2$, but for simplicity consider $K = 2$. Using previous results on superoptimal regimes in point treatment settings, we know that we can identify $\mathbb{E}[R_1^{a_1} | L_1, A_1]$ \citep{stensrud_optimal_2024}. It remains to be shown whether we can identify $\mathbb{E}[R_2^{a_1,a_2} | L_1, L_2^{a_1}, R_1^{a_1}, A_1 = a_1', A_2^{a_1} = a_2']$. However,
\begin{align}
    &\mathbb{E}[R_2^{a_1,a_2} | L_1, L_2^{a_1}, R_1^{a_1}, A_1 = a_1', A_2^{a_1} = a_2'] = \nonumber\\
    &\begin{cases}
       \mathbb{E}[R_2^{a_1} | L_1, L_2^{a_1}, R_1^{a_1}, A_1 = a_1', A_2^{a_1} = a_2], & \text{ if $a_2' = a_2$,}\\
       \frac{\mathbb{E}[R_2^{a_1,a_2} | L_1, L_2^{a_1}, R_1^{a_1}, A_1 = a_1'] - \mathbb{E}[R_2^{a_1} | L_1, L_2^{a_1}, R_1^{a_1}, A_1 = a_1', A_2^{a_1} = a_2] P(A_2^{a_1} = a_2 | L_1, L_2^{a_1}, R_1^{a_1}, A_1 = a_1')}{P(A_2^{a_1} = a_2' | L_1, L_2^{a_1}, R_1^{a_1}, A_1 = a_1')}, & \text{ if $a_2' \neq a_2$.}
    \end{cases} \label{eq: R_1 lemma 1}
\end{align}
Furthermore, 
\begin{align*}
    &\mathbb{E}[R_2^{a_1} | L_1, L_2^{a_1}, R_1^{a_1}, A_1 = a_1', A_2^{a_1} = a_2] \\
    &= \begin{cases}
        \mathbb{E}[R_2 | L_1, L_2, R_1, A_1 = a_1, A_2 = a_2], & \text{ if $a_1' = a_1$,}\\
        \frac{\mathbb{E}[R_2^{a_1} | L_1, L_2^{a_1}, R_1^{a_1}, A_2^{a_1} = a_2] - \mathbb{E}[R_2 | L_1, L_2, R_1, A_1 = a_1, A_2 = a_2]P(A_1 = a_1 | L_1, L_2^{a_1}, R_1^{a_1}, A_2^{a_1} = a_2)}{P(A_1 = a_1' | L_1, L_2^{a_1}, R_1^{a_1}, A_2^{a_1} = a_2)}, & \text{ if $a_1' \neq a_1$.}
    \end{cases}
\end{align*}
We also have that
\begin{align*}
    P(A_1 = a_1' | L_1, L_2^{a_1} = l_2, R_1^{a_1} = r_1, A_2^{a_1} = a_2) &= f_{A_1 | L_1, L_2^{a_1} = l_2, R_1^{a_1} = r_1, A_2^{a_1} = a_2}(a_1'),\\
    &= \frac{f_{L_2^{a_1}, R_1^{a_1} | A_1 = a_1', L_1}(l_2,r_1) f_{A_1 | L_1}(a_1')}{f_{L_2^{a_1}, R_1^{a_1} |  L_1}(l_2,r_1)},
\end{align*}
and
\begin{align*}
    f_{L_2^{a_1}, R_1^{a_1} |  L_1}(l_2,r_1) = \mathbb{E}[\delta( (L_2^{a_1}; R_1^{a_1}) - (l_2;r_1)) | L_1],
\end{align*}
which we can identify as $\delta( (L_2^{a_1}, R_1^{a_1}) - (l_2,r_1))$ can be identified when there is sequential identification, e.g. when $Z_1$ is a valid instrument for all future variables in the sequential IV setting. Then, 
\begin{align*}
    &f_{L_2^{a_1}, R_1^{a_1} | A_1 = a_1', L_1}(l_2,r_1)\\
    &= \mathbb{E}[\delta( (L_2^{a_1}; R_1^{a_1}) - (l_2;r_1)) | A_1 = a_1', L_1],\\
    &= \frac{\mathbb{E}[\delta( (L_2^{a_1}; R_1^{a_1}) - (l_2;r_1)) | L_1] - \mathbb{E}[\delta((L_2; R_1) - (l_2;r_1)) | A_1 = a_1, L_1]P(A_1 = a_1 | L_1)}{P(A_1 = a_1' | L_1)},\\
    &= \frac{f_{L_2^{a_1}, R_1^{a_1} |  L_1}(l_2,r_1) - f_{L_2, R_1 |A_1 = a_1',  L_1}(l_2,r_1)P(A_1 = a_1 | L_1)}{P(A_1 = a_1' | L_1)},
\end{align*}
which we can also identify. Similarly, we can identify $P(A_1 = a_1' | L_1, L_2^{a_1} = l_2, R_1^{a_1} = r_1, A_2^{a_1} = a_2)$ and $P(A_1 = a_1 | L_1, L_2^{a_1} = l_2, R_1^{a_1} = r_1, A_2^{a_1} = a_2)$.

Furthermore, 
\begin{align*}
    f_{R_2^{a_1} | L_1, L_2^{a_1} = l_2, R_1^{a_1} = r_1, A_2^{a_1} = a_2}(r_2) &= \frac{f_{R_2^{a_1}, L_2^{a_1}, R_1^{a_1}, A_2^{a_1} | L_1} (r_2,l_2,r_1,a_2)}{f_{L_2^{a_1}, R_1^{a_1}, A_2^{a_1} | L_1} (l_2,r_1,a_2)}\\
    &= \frac{\mathbb{E}[\delta((R_2^{a_1}; L_2^{a_1}; R_1^{a_1}; A_2^{a_1}) - (r_2;l_2;r_1;a_2)) | L_1 ]}{\mathbb{E}[\delta((L_2^{a_1}; R_1^{a_1}; A_2^{a_1}) - (l_2;r_1;a_2))| L_1 ]}
\end{align*}
The numerator and denominators are identified as $\delta((R_2^{a_1}; L_2^{a_1}; R_1^{a_1}; A_2^{a_1}) - (r_2;l_2;r_1;a_2))$ and $\delta(( L_2^{a_1}; R_1^{a_1}; A_2^{a_1}) - (l_2;r_1;a_2))$ are valid outcomes for $A_1$. Hence, 
\begin{align*}
    \mathbb{E}[R_2^{a_1} | L_1, L_2^{a_1} = l_2, R_1^{a_1} = r_1, A_2^{a_1} = a_2] = \int r_2 f_{R_2^{a_1} | L_1, L_2^{a_1} = l_2, R_1^{a_1} = r_1, A_2^{a_1} = a_2}(r_2) dr_2
\end{align*}
is identified. Thus, $\mathbb{E}[R_2^{a_1} | L_1, L_2^{a_1} = l_2, R_1^{a_1} = r_1, A_2^{a_1} = a_2, A_1 = a_1']$ is identified.

In the same way, 
\begin{align*}
    &P(A_2^{a_1} = a_2 | L_1, L_2^{a_1} = l_2, R_1^{a_1} = r_1, A_1 = a_1')\\
    &= \frac{\left(\frac{\mathbb{E}[\delta((A_2^{a_1} ; L_2^{a_1} ;R_1^{a_1}) - (a_2, l_2, r_1)) | L_1] - \mathbb{E}[\delta((A_2 ; L_2 ;R_1) - (a_2, l_2, r_2)) | L_1, A_1 = a_1] P(A_1 = a_1 | L_1)}{P(A_1 = a_1' | L_1)} \right)}{\left(\frac{\mathbb{E}[\delta(( L_2^{a_1} ;R_1^{a_1}) - ( l_2, r_1)) | L_1] - \mathbb{E}[\delta(( L_2 ;R_1) - ( l_2, r_1) )| L_1, A_1 = a_1] P(A_1 = a_1 | L_1)}{P(A_1 = a_1' | L_1)} \right)},\\
    &= \frac{\mathbb{E}[\delta((A_2^{a_1} ; L_2^{a_1} ;R_1^{a_1}) - (a_2, l_2, r_1)) | L_1]}{\mathbb{E}[\delta(( L_2^{a_1} ;R_1^{a_1}) - ( l_2, r_1)) | L_1] - \mathbb{E}[\delta(( L_2 ;R_1) - ( l_2, r_1)) | L_1, A_1 = a_1] P(A_1 = a_1 | L_1)}\\
    & - \frac{\mathbb{E}[\delta((A_2 ; L_2 ;R_1) - (a_2, l_2, r_1)) | L_1, A_1 = a_1] P(A_1 = a_1 | L_1)}{\mathbb{E}[\delta(( L_2^{a_1} ;R_1^{a_1}) - ( l_2, r_1)) | L_1] - \mathbb{E}[\delta(( L_2 ;R_1) - ( l_2, r_1) )| L_1, A_1 = a_1] P(A_1 = a_1 | L_1)},
\end{align*}
is identified. A similar argument shows $P(A_2^{a_1} = a_2' | L_1, L_2^{a_1} = l_2, R_1^{a_1} = r_1, A_1 = a_1')$ is identified.

Furthermore,
\begin{align*}
   &\mathbb{E}[R_2^{a_1,a_2} | L_1, L_2^{a_1}, R_1^{a_1}, A_1 = a_1']\\
   &= \frac{\mathbb{E}[R_2^{a_1,a_2} | L_1, L_2^{a_1}, R_1^{a_1}] - \mathbb{E}[R_2^{a_2} | L_1, L_2, R_1, A_1 = a_1]P(A_1 = a_1 | L_1, L_2^{a_1}, R_1^{a_1})}{P(A_1 = a_1' | L_1, L_2^{a_1}, R_1^{a_1})}.
\end{align*}
We can prove that $P(A_1 = a_1 | L_1, L_2^{a_1}, R_1^{a_1})$ and $P(A_1 = a_1' | L_1, L_2^{a_1}, R_1^{a_1})$ are identified using Dirac delta functions and Bayes' theorem as above. The value function $\mathbb{E}[R_2^{a_2} | L_1, L_2, R_1, A_1 = a_1]$ can be identified, for example, with IV methods. Finally, $\mathbb{E}[R_2^{a_1,a_2} | L_1, L_2^{a_1}, R_1^{a_1}]$ can be written as (or bounded by) a function of counterfactuals of $A_1$ only (and not $a_2$), by assumption. As shown above, the joint density of these counterfactuals conditional on $A_1$ and $L_1$ is identified using Dirac's delta function. Therefore, we can identify the elements in Equation \eqref{eq: R_1 lemma 1}, identifying $\mathbb{E}[R_2^{a_1,a_2} | L_1, L_2^{a_1}, R_1^{a_1}, A_1 = a_1', A_2^{a_1} = a_2']$. 

\begin{remark}
    This argument can be extended to $K > 2$ in the same manner. Crucially, this argument relies on the treatments $A_k$ being binary. We do not expect $g^{\textbf{sup}}$ to be identified, even when there are valid IVs at every time point, when treatments can take more than two values.
\end{remark}

\begin{remark}
    In practice, Dirac's delta function cannot be used as an outcome, as it is only ``nonzero" at one point, but with a nonzero expectation. Indeed, Dirac's delta defines a degenerate measure, see \citet{kallenberg1997foundations} for more details. However, we can use smooth approximations, such as Gaussian densities with variances quickly converging to $0$, or other kernel density estimation techniques, for appropriate asymptotic convergence rates of estimators, see \citet{van2000asymptotic} for an overview of results. 
\end{remark}

\section{Expression for the influence function}
\label{app: IF formula}
We give the formula for $\Psi(\mathbb{IF}(\mathbb{E}[R_k^{a_k} | H_k, \overline{A}_{k-1} = \overline{a}_{k-1}]), \overline{A}_k, H_k$), analogously to the formula derived in the single-time-point setting by \citet{stensrud_optimal_2024}:
\begin{align*} 
    &\Psi(\mathbb{IF}(\mathbb{E}[R_k^{a_k} | H_k, \overline{A}_{k-1} = \overline{a}_{k-1}]), \overline{A}_k, H_k)\\
    &:= \bigg\{ \big(\mathbb{IF}(\mathbb{E}[R_k^{a_k} | H_k, \overline{A}_{k-1} = \overline{a}_{k-1}])\\
    &- \underbrace{\frac{I(\overline{A}_k = \overline{a}_k)}{P(\overline{A}_k = \overline{a}_k \mid H_k)}(R_k - \mathbb{E}[R_k | H_k, \overline{A}_k = \overline{a}_k])}_{\mathbb{IF}(\mathbb{E}[R_k | H_k, \overline{A}_k = \overline{a}_k])}P(A_k = a_k | H_k, \overline{A}_{k-1} = \overline{a}_{k-1})\\
    &- \mathbb{E}[R_k | H_k, \overline{A}_k = \overline{a}_k]\\
    &\cdot\underbrace{\frac{I(\overline{A}_{k-1} = \overline{a}_{k-1})}{P(\overline{A}_{k-1} = \overline{a}_{k-1} \mid H_k)}(I(A_k = a_k) - P(A_k = a_k |H_k, \overline{A}_{k-1} = \overline{a}_{k-1}))}_{\mathbb{IF}(P(A_k = a_k | H_k, \overline{A}_{k-1} = \overline{a}_{k-1}))} \big)\\
    & \cdot P(A_k = a_k' | H_k, \overline{A}_{k-1} = \overline{a}_{k-1})\\
    &- \underbrace{\frac{I(\overline{A}_{k-1} = \overline{a}_{k-1})}{P(\overline{A}_{k-1} = \overline{a}_{k-1} \mid H_k)}(I(A_k = a_k') - P(A_k = a_k' | H_k, \overline{A}_{k-1} = \overline{a}_{k-1}))}_{\mathbb{IF}(P(A_k = a_k'| H_k, \overline{A}_{k-1} = \overline{a}_{k-1}))}\\
    &\cdot(\mathbb{E}[R_k^{a_k}| H_k, \overline{A}_{k - 1} = \overline{a}_{k-1}] - \mathbb{E}[R_k|H_k, \overline{A}_{k} = \overline{a}_{k}]P(A_k = a_k | H_k, \overline{A}_{k-1} = \overline{a}_{k-1}))\bigg\}\\
    & \cdot P(A_k = a_k'| H_k, \overline{A}_{k-1} = \overline{a}_{k-1})^{-2}.
\end{align*}

\section{Data-generating mechanisms}
\label{app: Data generating mechs}

\subsection{\citet{batorsky2024integrating} example}
Recall the data-generating mechanism given by \citet{batorsky2024integrating}:
\begin{itemize}
    \item $U \sim \mathcal{N}(0,1)$,
    \item $\text{age} \sim \mathcal{N}(52, 8^2)$,
    \item $\text{opioid}_1 \sim \text{Bernoulli}(0.2)$,
    \item $\text{depression}_1 \sim \text{Bernoulli}(0.3)$,
    \item $A_1 \sim \text{Bernoulli}(\text{expit}(-0.5\text{opioid}_1 - 0.2\text{depression}_1 + 2U)),$
    \item $Y_1 = 4.5 - \text{age}_{\text{std}} + 0.3\text{opioid}_1 - \text{opioid}_1 A_1 + 2\text{depression}_1 A_1 - 0.3\text{age}_{\text{std}} A_1 - 0.6 \text{age}_{\text{std}}^2 A_1 - 0.01 \text{age}_{\text{std}}^3 A_1 + 2UA_1 + \mathcal{N}(0, 0.5)$,
    \item $\text{opioid}_2 \sim \text{Bernoulli}(\text{expit}(\text{opioid}_1 - 0.5A_1))$,
    \item $\text{depression}_2 \sim \text{Bernoulli}(\text{expit}(\text{depression}_1 + 0.7A_1))$,
    \item $A_2 \sim \text{Bernoulli}(\text{expit}(-0.5\text{opioid}_2 - 0.2 \text{depression}_2 + 2U))$,
    \item $Y_2 = 4.5 - \text{age}_{\text{std}} + 0.2\text{opioid}_2 -0.1 \text{depression}_2 + A_2 - \text{opioid}_2 A_2 - 1.5\text{depression}_2 A_2 + 0.1 \text{resp} - 0.5\text{resp}A_2 + 0.3A_1 - 0.3\text{age}_{\text{std}} A_1 - 0.6 \text{age}_{\text{std}}^2 A_1 - 0.01 \text{age}_{\text{std}}^3 A_1 + 2UA_2 + \mathcal{N}(0, 1)$,
\end{itemize}
where $\text{resp} = I(Y_1 > c)$ (the cutoff value $c$ is chosen as in \citet{batorsky2024integrating}), and $\text{age}_{\text{std}}$ is the standardized age, see Figure \ref{fig:trial sim DAG} for the DAG representing the observed data-generating mechanism.

\begin{figure}
\centering
    \begin{tikzpicture}[scale = 0.8]
                \tikzset{line width=1.5pt, outer sep=0pt,
                ell/.style={draw,fill=white, inner sep=2pt,
                line width=1.5pt},
                swig vsplit={gap=5pt,
                inner line width right=0.5pt},
                swig hsplit={gap=5pt}
                };
                    \node[name=A1,ell,shape=ellipse] at (3,0){$A_1$};
                    \node[name=Y1, ell, shape=ellipse] at (9,0){$Y_1$};
                    \node[name=U, ell, shape=ellipse] at (10.5,2){$U$};
                    \node[name=L1, ell, shape=ellipse] at (6,-2){$L_{1}$};
                    \node[name=A2, ell, shape = ellipse] at (12,0){$A_2$};
                    \node[name=L2, ell, shape = ellipse] at (15,-2){$L_2$};
                    \node[name=Y2, ell, shape = ellipse] at (18,0){$Y_2$};
                \begin{scope}[>={Stealth[black]},
                  every node/.style={fill=white,circle},
                  every edge/.style={draw=black,very thick}]
                    \path[->] (U) edge (Y1);
                    \path[->] (U) edge (A1);
                    \path[->] (L1) edge (A1);
                    \path[->] (L1) edge (Y1);
                    \path[->] (A1) edge (Y1);
                    \path[->] (Y1) edge[bend left = 15] (Y2);
                    \path[->] (L1) edge (L2);
                    \path[->] (L1) edge (Y2);
                    \path[->] (A1) edge[bend left = 15] (Y2);
                    \path[->] (L2) edge (A2);
                    \path[->] (L2) edge (Y2);
                    \path[->] (A2) edge (Y2);
                    \path[->] (U) edge (A2);
                    \path[->] (U) edge (Y2);
                \end{scope}
            \end{tikzpicture}
    \caption{DAG representing the data-generating mechanism of the trial design simulation in Section \ref{sec: Sim}.}
    \label{fig:trial sim DAG}
\end{figure}
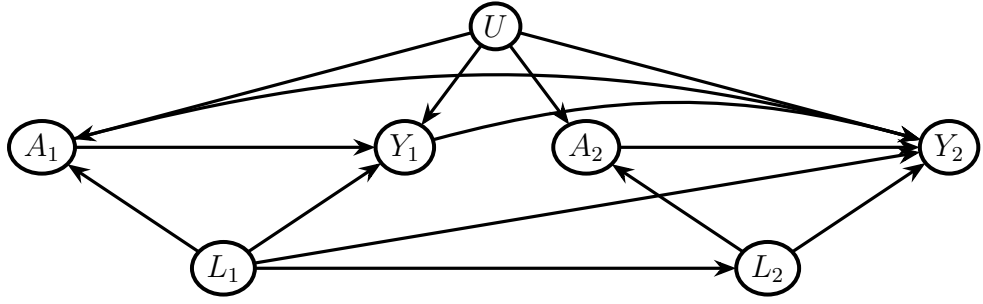

The observed regime performs relatively well as $U$ is positively correlated with $A_1$ and $A_2$ and the interaction terms in the expectations of $Y_1$ and $Y_2$ are ``$+2UA_1$" and ``$+2UA_2$", respectively: if $U > 0$, $A_1$ and $A_2$ are more likely to be equal to one, which makes $2UA_1$ and $2UA_2$ positive, and if $U < 0$, $A_1$ and $A_2$ are more likely to be equal to zero, which sets them to zero. The optimal initiation regime improves on the observed regime by deviating when the recorded history indicates that the natural treatment value is unlikely to be optimal. To obtain a data-generating mechanism where the observed regime is clearly suboptimal, we replace ``$+2UA_1$" and ``$+2UA_2$" by ``$-2U(1-2A_1)$" and ``$+2U(1-2A_2)$", respectively. If $U > 0$, it is then advantageous to pick $A_1 = 1$ and $A_2 = 0$ to maximize the interaction terms. Reciprocally, if $U < 0$, it is advantageous to pick $A_1 = 0$ and $A_2 = 1$.

\subsubsection{$K = 3$ and $K = 4$ simulations}
 We provide simulations of Data Structures \ref{algo: randomized trial obs}-\ref{algo: three choice trial } for $K = 3$ and $K = 4$.

For both simulations, we keep the data-generating mechanisms for $\overline{L}_2, \overline{A}_2, \overline{Y}_2$.

For $k \geq 2$, we simulate the stage $k$ covariates as follows:
\begin{align*}
O_k &\sim \text{Bernoulli}\big(\text{expit}(O_{k-1}-0.5A_{k-1})\big),\\
D_k &\sim \text{Bernoulli} \big(\text{expit}(D_{k-1}+0.7A_{k-1})\big),\\
A_k &\sim \text{Bernoulli}\big(\text{expit}(-0.5O_k-0.2D_k+2U)\big),
\end{align*}
where $O_2 = \text{opioid}_2$ and $D_2 = \text{depression}_2$.

The stage $k \geq 2$ outcomes are simulated as follows:
\begin{align*}
    Y_k =& 4.5- \text{age}_{\text{std}}+0.2O_k-0.1D_k +A_k-O_kA_k-1.5D_kA_k+0.1\text{resp}_{k-1}-0.5\text{resp}_{k-1}A_k\\
      &+0.3A_{k-1}-\big(0.3\text{age}_{\text{std}}+0.6\text{age}_{\text{std}}^{2}+0.01\text{age}_{\text{std}}^{3}\big)A_{k-1}
        +2U(1-2A_k)\\
        &-0.8A_kA_{k-1}+UA_kA_{k-1}-0.5A_kA_{k-2}+\mathcal N(0,1),
\end{align*}
where $\text{resp}_{k} = I(Y_k > c)$.

The results for $K = 3$ and $K = 4$ are in Tables \ref{tab: combined sim K3} and \ref{tab: combined sim K4}, respectively. As expected $g^{\textbf{init}}$ outperforms $g^{\textbf{obs}}$ and $g^{\textbf{opt}}$ in both settings.

\begin{table}
    \centering
    \spacingset{1}
    \begin{tabular}{c | c | c | c}
         \textbf{Regime} & \textbf{DS \ref{algo: randomized trial obs}} & \textbf{DS \ref{algo: recorded nat value trial}} & \textbf{DS \ref{algo: three choice trial }}\\
         \hline
         $\hat{g}^{\textbf{obs}}$ & 11.432 (11.366, 11.499) & 11.420 (11.354, 11.486) & 11.443 (11.376, 11.510)\\
         $\hat{g}^{\textbf{opt}}$ & 15.466 (15.396, 15.537) & 15.455 (15.382, 15.528) & 15.381 (15.310, 15.452)\\
         $\hat{g}^{\textbf{init}}$ & 16.041 (15.966, 16.115) & 16.171 (16.097, 16.245) & 15.951 (15.876, 16.025)\\
    \end{tabular}
    \caption{Expected values and confidence intervals of the regimes obtained by Data Structures \ref{algo: randomized trial obs}, \ref{algo: recorded nat value trial}, and \ref{algo: three choice trial } over $K=3$ time-points, on the test set with $n_{\text{eval}} = 20'000$ samples based on the modified chronic back pain example of \citet{batorsky2024integrating}. Higher values are better.}
    \label{tab: combined sim K3}
\end{table}

\begin{table}
    \centering
    \spacingset{1}
    \begin{tabular}{c | c | c | c}
         \textbf{Regime} & \textbf{DS \ref{algo: randomized trial obs}} & \textbf{DS \ref{algo: recorded nat value trial}} & \textbf{DS \ref{algo: three choice trial }}\\
         \hline
         $\hat{g}^{\textbf{obs}}$ & 14.293 (14.204, 14.381) & 14.255 (14.166, 14.343) & 14.301 (14.213, 14.389)\\
         $\hat{g}^{\textbf{opt}}$ & 20.580 (20.488, 20.673) & 20.479 (20.384, 20.574) & 20.416 (20.330, 20.501)\\
         $\hat{g}^{\textbf{init}}$ & 20.921 (20.820, 21.023) & 21.114 (21.019, 21.208) & 20.829 (20.732, 20.926)\\
    \end{tabular}
    \caption{Expected values and confidence intervals of the regimes obtained by Data Structures \ref{algo: randomized trial obs}, \ref{algo: recorded nat value trial}, and \ref{algo: three choice trial } over $K=4$ time-points, on the test set with $n_{\text{eval}} = 20'000$ samples in the modified chronic back pain example of \citet{batorsky2024integrating}. Higher values are better.}
    \label{tab: combined sim K4}
\end{table}

\section{Proofs}
\label{app: proofs}

\subsection{Proof of Proposition \ref{prop: ssw > max(obs,opt)}}
\label{app: proof of ssw > max(obs,opt)}
By definition of $\mathcal{G}^{\textbf{init}}$, $g^{\textbf{obs}}, g^{\textbf{opt}}, g^{\textbf{osh}} \in \mathcal{G}^{\textbf{init}}$. Then, the inequality in Proposition \ref{prop: ssw > max(obs,opt)} follows by the definition of $g^{\textbf{init}}$.

The fact that there exist distributions for which the inequality is strict comes from the fact that for any time-point, the superoptimal regime can strictly outperform the optimal and observed regimes \citep{stensrud_optimal_2024, laurendeau2024improved}, which the optimal initiation regime can identify if this happens at the initiation point, unlike the optimal-shifting, optimal and observed regimes.


Let $U_1 \sim \text{Bern}(0.5)$, $A_1\sim \text{Bern}(0.1 + 0.8U_1)$, $R_1 = -4A_1 + 10U_1 + 10A_1U_1 + \epsilon_{R_1}$, where $\mathbb{E}[\epsilon_{R_1}] = 0$, and $U_2 \sim \text{Bern}(0.5)$, $A_2 \sim \text{Bern}(0.1 + 0.8U_2)$, $R_2 = -4(1-A_2) + 10U_2 + 10(1-A_2)U_2 + \epsilon_{R_2}$, where $\mathbb{E}[\epsilon_{R_2}] = 0$. 

Then,
\begin{enumerate}
    \item $\mathbb{E}[R_1^{a_1}] = -4a_1 + 5 + 5a_1 = a_1 + 5$,
    \item $\mathbb{E}[R_1^{a_1} | A_1 = a_1'] = -4a_1 + (10 + 10a_1)\mathbb{E}[U_1 | A_1 = a_1'] = -4a_1 + (10 + 10a_1)\left[0.9a_1' + 0.1(1-a_1')\right]$,
    \item $\mathbb{E}[R_2^{a_2}] = -4(1-a_2) + 5 + 5(1-a_2) = 6 - a_2$,
    \item $\mathbb{E}[R_2^{a_2} | A_2 = a_2'] = -4(1-a_2) + \left[10 + 10(1-a_2)\right]\mathbb{E}[U_2 | A_2 = a_2'] = -4(1-a_2) + \left[10 + 10(1-a_2)\right]\left[0.9a_2' + 0.1(1-a_2')\right]$.
\end{enumerate}
Hence,
\begin{enumerate}
    \item $g^{\textbf{opt}}_1 = 1$ and $g^{\textbf{opt}}_2 = 0$,
    \item $g^{\textbf{init}}_1(A_1) = A_1$ and $g^{\textbf{init}}_2(A_2) = 1-A_2$,
    \item $g^{\textbf{osh}} = g^{\textbf{obs}_{2}}$, that is, $g^{\textbf{osh}}_1(A_1) = A_1$ and $g^{\textbf{osh}}_2 = 0$,
\end{enumerate}
and, writing $v(g) := \mathbb{E}[R_1^{g} + R_2^{g}]$,
\begin{align*}
    v(g^{\textbf{init}}) &= \tfrac{1}{2}(14) + \tfrac{1}{2}(1)  +  \tfrac{1}{2}(14) + \tfrac{1}{2}(1) = 15, \\
    v(g^{\textbf{osh}}) &= 7.5 + 6 = 13.5, \qquad
    v(g^{\textbf{opt}}) = 6 + 6 = 12, \qquad
    v(g^{\textbf{obs}}) = 7.5 + 3.5 = 11,
\end{align*}
so that
\begin{align*}
    \mathbb{E}[R_1^{g^{\textbf{init}}} + R_2^{g^{\textbf{init}}}] = 15 > 13.5 = \max\left(\mathbb{E}[R_1^{g^{\textbf{osh}}} + R_2^{g^{\textbf{osh}}}], \mathbb{E}[R_1^{g^{\textbf{obs}}} + R_2^{g^{\textbf{obs}}}], \mathbb{E}[R_1^{g^{\textbf{opt}}} + R_2^{g^{\textbf{opt}}}]\right),
\end{align*}
which establishes the second claim of Proposition \ref{prop: ssw > max(obs,opt)}.

\subsection{Proof of Proposition \ref{prop: usual data fusion impossibility}}
        For $k > 1$, the natural treatment value $A_k^g$ is unobserved in general, as it is not recorded in the sequentially randomized trial after interventions are made. By consistency,
        \begin{align}
            \mathbb{E}[R_{k}^{a_{k}}|H_{k}, \overline{A}_{k-1}^g = \overline{A}_{k - 1}^{g+}, A_k^g = a_k'] = \mathbb{E}[R_k^{a_k} | \overline{L}_k, \overline{R}_{k-1}, \overline{A}_{k-1}, A_k = a_k']. \label{eq: obs + trial single time-point equivalence}
        \end{align}
        We can interpret Equation \eqref{eq: obs + trial single time-point equivalence} to be a counterfactual parameter of a single intervention that sets $A_k$ to $a_k$. Hence, Equation \eqref{eq: obs + trial single time-point equivalence} can be considered as a single-time-point counterfactual conditional on the observed past. However, Lemma 1 in \citet{stensrud_optimal_2024} shows that, under positivity and consistency, identifying the expression in Equation \eqref{eq: obs + trial single time-point equivalence} for $a_k' = 1-a_k$ is equivalent to identifying
        \begin{align*}
            \mathbb{E}[R_k^{a_k} | \overline{L}_k, \overline{R}_{k-1},  \overline{A}_{k-1} = \overline{A}_{k-1}^{g+}].
        \end{align*}

        As Equation \eqref{eq: obs + trial single time-point equivalence} contains counterfactual quantities, trial data is needed for identification if we do not impose any additional assumptions on the observed data. However, in the trial data, the natural treatment values $\overline{A}_{k-1}$ are unmeasured. Hence, the quantity in Equation \eqref{eq: obs + trial single time-point equivalence} cannot be identified in general.



    
    \subsection{Proof of Proposition \ref{prop: randomized trial obs id}}
    The proof follows from the trial design and Proposition \ref{prop: initiation regime id}, but we provide an alternative proof with structured tree graphs introduced in \citet{robins1986new}.
    \subsubsection{Identification using fully randomized measured causally interpreted structured tree graph models}
\label{app: FRMCISTG models}
In a seminal paper, \citet{robins1986new} introduced Fully Randomized Measured Causally Interpreted Structured Tree Graphs (FRMCISTGs). FRMCISTGs are tree graphs, where each node represents a randomization step and each edge a treatment; intranodal edges are used to represent different natural treatment values at each time point, analogously to their use to represent different covariate levels in \citet{robins1986new}. FRMCISTGs allow for easy representation of the class of parameters identified by the trial designs in Section \ref{sec: data fusion}. In particular, regimes corresponding to source to leaf paths are identified by the trial designs. 

In Figure \ref{fig: Trial Class I FRMCISTGs}, we show the FRMCISTG for the conventional sequentially randomized trial and the staggered entry trial which follows the natural treatment value at the first time point. Figure \ref{fig: Trial Class II FRMCISTGs} presents the FRMCISTG graphs for ``adapted" staggered entry trials, that randomize conditionally on the natural treatment values at the first and first two time points, respectively. Finally, Figure \ref{fig: Trial Class III FRMCISTG} shows the FRMCISTG representation for restricted preference trials that randomize between following the natural treatment values or starting a sequentially randomized trial. This corresponds to Trial design \ref{algo: randomized trial obs} in Section \ref{sec: data fusion}. It is straightforward to see that the initiation regimes can be represented as source-to-leaf paths in the FRMCISTGs for adapted staggered entry SMART trials and restricted patient preference trials, which proves that they can be identified by these trial designs as per \citet{robins1986new}. This provides a nice identification argument for the initiation regimes and $g^{\textbf{init}}$ as a result. Moreover, one can also see that $g^{\textbf{sup}}$ will not in general be identified by the trial designs presented in Figures \ref{fig: Trial Class I FRMCISTGs}-\ref{fig: Trial Class III FRMCISTG} unless additional assumptions are made.

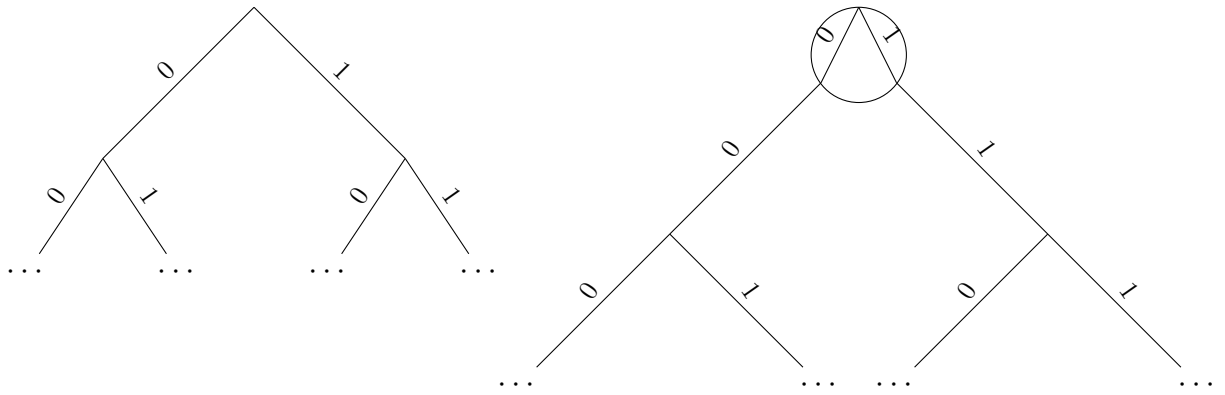
\begin{figure}
    \centering
    \begin{tikzpicture}
           \node (start) at (0,0) {};
            \node (A1=1) at (2,-2) {};
            \node (3DL) at (-3,-3.5) {$\cdots$};
            \node (3DR) at (3,-3.5) {$\cdots$};
            \node (3DML) at (-1,-3.5) {$\cdots$};
            \node (3DMR) at (1,-3.5) {$\cdots$};

            \draw[-] (0,0) -- node[rotate = 45, above, scale = .85] {$0$} (-2,-2);
            \draw[-] (0,0) -- node[rotate = 315, above, scale = .85] {$1$} (2,-2);
            \draw[-] (-2,-2) -- node[rotate = 55, above, scale = .85]{$0$} (3DL);
            \draw[-] (-2,-2) -- node[rotate = 305, above, scale = .85] {$1$} (3DML);
            \draw[-] (2,-2) -- node[rotate = 55, above, scale = .85]{$0$} (3DMR);
            \draw[-] (2,-2) -- node[rotate = 305, above, scale = .85] {$1$} (3DR);

        \node (source) at (8,0) {};
        \node (source_left) at (7.5, -1) {};
        \node (source_right) at (8.5, -1) {};
        \node (L) at (5.5, -3) {};
        \node (R) at (10.5, -3) {};
        \node (LL) at (3.5, -5) {$\cdots$};
        \node (LR) at (7.5, -5) {$\cdots$};
        \node (RL) at (8.5, -5) {$\cdots$};
        \node (RR) at (12.5,-5) {$\cdots$};

        \draw[-] (8,0) -- node[rotate = 55, above, scale = .85]{$0$} (7.5,-1);
        \draw[-] (8,0) -- node[rotate = 305, above, scale = .85] {$1$} (8.5,-1);
        \draw (8,-0.63) circle (0.63);
        \draw (7.5,-1) -- node[rotate = 55, above, scale = .85]{$0$} (5.5,-3);
        \draw (8.5, -1) -- node[rotate = 305, above, scale = .85] {$1$} (10.5,-3);
        \draw (5.5,-3) -- node[rotate = 55, above, scale = .85]{$0$} (LL);
        \draw (5.5, -3) -- node[rotate = 305, above, scale = .85] {$1$} (LR);
        \draw (10.5, -3 ) -- node[rotate = 55, above, scale = .85]{$0$} (RL);
        \draw (10.5, -3) -- node[rotate = 305, above, scale = .85] {$1$} (RR);
    \end{tikzpicture}
    \caption{FRMCISTGs for conventional sequentially randomized trial and staggered entry SMART trials, where patients enter at second time point.}
    \label{fig: Trial Class I FRMCISTGs}
\end{figure}
\begin{figure}
    \centering
    \begin{tikzpicture}[scale = 0.45, transform shape]
        \node (source) at (8,0) {};
        \node (source_left) at (7.5, -1) {};
        \node (source_right) at (8.5, -1) {};
        \node (EL) at (1,-3) {};
        \node (ELL) at (-0.5,-5) {$\cdots$};
        \node (ELR) at (2.5,-5) {$\cdots$};
        \node (L) at (5.5, -3) {};
        \node (R) at (10.5, -3) {};
        \node (LL) at (4, -5) {$\cdots$};
        \node (LR) at (7, -5) {$\cdots$};
        \node (RL) at (9, -5) {$\cdots$};
        \node (RR) at (12,-5) {$\cdots$};
        \node (ERL) at (13.5,-5) {$\cdots$};
        \node (ER) at (15,-3) {};
        \node (ERR) at (16.5,-5) {$\cdots$}; 

        \draw[-] (8,0) -- node[rotate = 55, above, scale = .85]{$0$} (7.5,-1);
        \draw[-] (8,0) -- node[rotate = 305, above, scale = .85] {$1$} (8.5,-1);
        \draw (8,-0.63) circle (0.63);
        \draw (7.5,-1) -- node[rotate = 55, above, scale = .85]{$0$} (5.5,-3);
        \draw (8.5, -1) -- node[rotate = 305, above, scale = .85] {$1$} (10.5,-3);
        \draw (5.5,-3) -- node[rotate = 55, above, scale = .85]{$0$} (LL);
        \draw (5.5, -3) -- node[rotate = 305, above, scale = .85] {$1$} (LR);
        \draw (10.5, -3 ) -- node[rotate = 55, above, scale = .85]{$0$} (RL);
        \draw (10.5, -3) -- node[rotate = 305, above, scale = .85] {$1$} (RR);

        \draw (7.5,-1) -- node[rotate = 25, above, scale = .85]{$1$} (1,-3);
        \draw (1,-3) -- node[rotate = 55, above, scale = .85]{$0$} (ELL);
        \draw (1,-3) -- node[rotate = 305, above, scale = .85] {$1$} (ELR);

        \draw (8.5,-1) -- node[rotate = 335, above, scale = .85]{$0$} (15,-3);
        \draw (15,-3) -- node[rotate = 55, above, scale = .85]{$0$} (ERL);
        \draw (15,-3) -- node[rotate = 305, above, scale = .85] {$1$} (ERR);

        \node (source2) at (26,0) {};
        \node (source_left2) at (25.5, -1) {};
        \node (source_right2) at (26.5, -1) {};
        
        \node (EL2) at (19,-3) {};
        \node (EL2_left) at (18.5,-4) {};
        \node (EL2_right) at (19.5,-4) {};
        
        \node (L2) at (23.5, -3) {};
        \node (L2_left) at (23,-4) {};
        \node (L2_right) at (24,-4) {};
        
        \node (R2) at (28.5, -3) {};
        \node (R2_left) at (28,-4) {};
        \node (R2_right) at (29,-4) {};
        
        \node (LLLL) at (17.5,-6) {$\cdots$};
        \node (LLLR) at (18.5,-6) {$\cdots$};
        \node (LLRL) at (19.5,-6) {$\cdots$};
        \node (LLRR) at (20.5,-6) {$\cdots$};

        \node (LRLL) at (22,-6) {$\cdots$};
        \node (LRLR) at (23,-6) {$\cdots$};
        \node (LRRL) at (24,-6) {$\cdots$};
        \node (LRRR) at (25,-6) {$\cdots$};

        \node (RLLL) at (27,-6) {$\cdots$};
        \node (RLLR) at (28,-6) {$\cdots$};
        \node (RLRL) at (29,-6) {$\cdots$};
        \node (RLRR) at (30,-6) {$\cdots$};

        \node (RRLL) at (31.5,-6) {$\cdots$};
        \node (RRLR) at (32.5,-6) {$\cdots$};
        \node (RRRL) at (33.5,-6) {$\cdots$};
        \node (RRRR) at (34.5,-6) {$\cdots$};

        \node (ER2) at (33,-3) {};
        \node (ER2_left) at (32.5,-4) {};
        \node (ER2_right) at (33.5,-4) {};
        

        \draw[-] (26,0) -- node[rotate = 55, above, scale = .85]{$0$} (25.5,-1);
        \draw[-] (26,0) -- node[rotate = 305, above, scale = .85] {$1$} (26.5,-1);
        \draw (26,-0.63) circle (0.63);
        \draw (25.5,-1) -- node[rotate = 55, above, scale = .85]{$0$} (23.5,-3);
        \draw (23.5,-3) -- node[rotate = 55, above, scale = .85]{$0$} (23,-4);
        \draw (23.5,-3) -- node[rotate = 305, above, scale = .85] {$1$} (24,-4);
        \draw (23.5,-3.63) circle (0.63);
        \draw (23,-4) -- node[rotate = 55, above, scale = .85]{$0$} (LRLL);
        \draw (23,-4) -- node[rotate = 305, above, scale = .85] {$1$} (LRLR);
        \draw (24,-4) -- node[rotate = 55, above, scale = .85]{$0$} (LRRL);
        \draw (24,-4) -- node[rotate = 305, above, scale = .85] {$1$} (LRRR);
        
        \draw (26.5, -1) -- node[rotate = 305, above, scale = .85] {$1$} (28.5,-3);
        \draw (28.5,-3) -- node[rotate = 55, above, scale = .85]{$0$} (28,-4);
        \draw (28.5,-3) -- node[rotate = 305, above, scale = .85] {$1$} (29,-4);
        \draw (28.5,-3.63) circle (0.63);
        \draw (28,-4) -- node[rotate = 55, above, scale = .85]{$0$} (RLLL);
        \draw (28,-4) -- node[rotate = 305, above, scale = .85] {$1$} (RLLR);
        \draw (29,-4) -- node[rotate = 55, above, scale = .85]{$0$} (RLRL);
        \draw (29,-4) -- node[rotate = 305, above, scale = .85] {$1$} (RLRR);

        \draw (25.5,-1) -- node[rotate = 25, above, scale = .85]{$1$} (19,-3);
        \draw (19,-3) -- node[rotate = 55, above, scale = .85]{$0$} (18.5,-4);
        \draw (19,-3) -- node[rotate = 305, above, scale = .85] {$1$} (19.5,-4);
        \draw (19, -3.63) circle (0.63);
        \draw (18.5,-4) -- node[rotate = 55, above, scale = .85]{$0$} (LLLL);
        \draw (18.5,-4) -- node[rotate = 305, above, scale = .85] {$1$} (LLLR);
        \draw (19.5,-4) -- node[rotate = 55, above, scale = .85]{$0$} (LLRL);
        \draw (19.5,-4) -- node[rotate = 305, above, scale = .85] {$1$} (LLRR);

        \draw (26.5,-1) -- node[rotate = 335, above, scale = .85]{$0$} (33,-3);
        \draw (33,-3) -- node[rotate = 55, above, scale = .85]{$0$} (32.5,-4);
        \draw (33,-3) -- node[rotate = 305, above, scale = .85] {$1$} (33.5,-4);
        \draw (33,-3.63) circle (0.63);
        \draw (32.5,-4) -- node[rotate = 55, above, scale = .85]{$0$} (RRLL);
        \draw (32.5,-4) -- node[rotate = 305, above, scale = .85] {$1$} (RRLR);
        \draw (33.5,-4) -- node[rotate = 55, above, scale = .85]{$0$} (RRRL);
        \draw (33.5,-4) -- node[rotate = 305, above, scale = .85] {$1$} (RRRR);
    \end{tikzpicture}
    \caption{FRMCISTGs for adapted staggered entry SMART trials, where natural treatment values are recorded at the first time point only and the first two time points, respectively.}
    \label{fig: Trial Class II FRMCISTGs}
\end{figure}
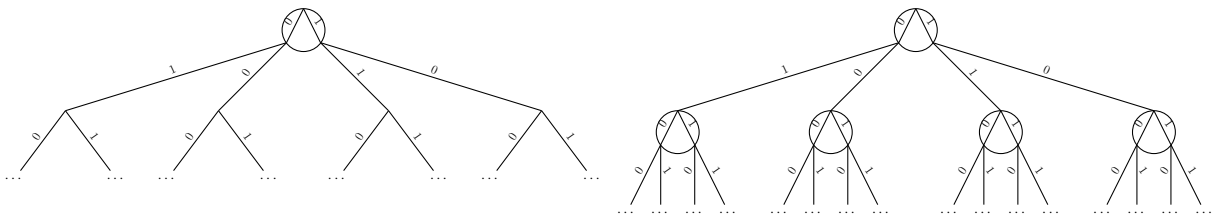

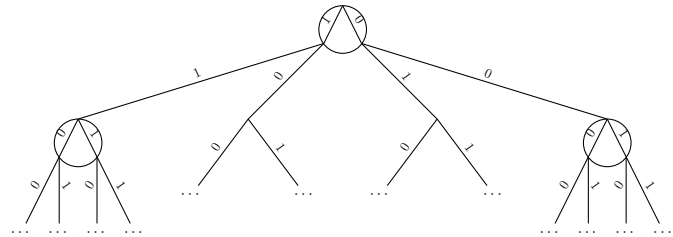
\begin{figure}
    \centering
     \begin{tikzpicture}[scale = 0.5, transform shape]
        \node (source2) at (26,0) {};
        \node (source_left2) at (25.5, -1) {};
        \node (source_right2) at (26.5, -1) {};
        
        \node (EL2) at (19,-3) {};
        \node (EL2_left) at (18.5,-4) {};
        \node (EL2_right) at (19.5,-4) {};
        
        \node (L2) at (23.5, -3) {};
        
        \node (R2) at (28.5, -3) {};
        
        \node (LL2) at (22, -5) {$\cdots$};
        \node (LR2) at (25, -5) {$\cdots$};
        \node (RL2) at (27, -5) {$\cdots$};
        \node (RR2) at (30,-5) {$\cdots$};
        \node (LLLL) at (17.5,-6) {$\cdots$};
        \node (LLLR) at (18.5,-6) {$\cdots$};
        \node (LLRL) at (19.5,-6) {$\cdots$};
        \node (LLRR) at (20.5,-6) {$\cdots$};



        \node (RRLL) at (31.5,-6) {$\cdots$};
        \node (RRLR) at (32.5,-6) {$\cdots$};
        \node (RRRL) at (33.5,-6) {$\cdots$};
        \node (RRRR) at (34.5,-6) {$\cdots$};

        \node (ER2) at (33,-3) {};
        \node (ER2_left) at (32.5,-4) {};
        \node (ER2_right) at (33.5,-4) {};
        

        \draw[-] (26,0) -- node[rotate = 55, above, scale = .85]{$1$} (25.5,-1);
        \draw[-] (26,0) -- node[rotate = 305, above, scale = .85] {$0$} (26.5,-1);
        \draw (26,-0.63) circle (0.63);
        \draw (25.5,-1) -- node[rotate = 55, above, scale = .85]{$0$} (23.5,-3);
        
        \draw (26.5, -1) -- node[rotate = 305, above, scale = .85] {$1$} (28.5,-3);
        \draw (23.5,-3) -- node[rotate = 55, above, scale = .85]{$0$} (LL2);
        \draw (23.5, -3) -- node[rotate = 305, above, scale = .85] {$1$} (LR2);
        \draw (28.5, -3 ) -- node[rotate = 55, above, scale = .85]{$0$} (RL2);
        \draw (28.5, -3) -- node[rotate = 305, above, scale = .85] {$1$} (RR2);

        \draw (25.5,-1) -- node[rotate = 25, above, scale = .85]{$1$} (19,-3);
        \draw (19,-3) -- node[rotate = 55, above, scale = .85]{$0$} (18.5,-4);
        \draw (19,-3) -- node[rotate = 305, above, scale = .85] {$1$} (19.5,-4);
        \draw (19, -3.63) circle (0.63);
        \draw (18.5,-4) -- node[rotate = 55, above, scale = .85]{$0$} (LLLL);
        \draw (18.5,-4) -- node[rotate = 305, above, scale = .85] {$1$} (LLLR);
        \draw (19.5,-4) -- node[rotate = 55, above, scale = .85]{$0$} (LLRL);
        \draw (19.5,-4) -- node[rotate = 305, above, scale = .85] {$1$} (LLRR);

        \draw (26.5,-1) -- node[rotate = 335, above, scale = .85]{$0$} (33,-3);
        \draw (33,-3) -- node[rotate = 55, above, scale = .85]{$0$} (32.5,-4);
        \draw (33,-3) -- node[rotate = 305, above, scale = .85] {$1$} (33.5,-4);
        \draw (33,-3.63) circle (0.63);
        \draw (32.5,-4) -- node[rotate = 55, above, scale = .85]{$0$} (RRLL);
        \draw (32.5,-4) -- node[rotate = 305, above, scale = .85] {$1$} (RRLR);
        \draw (33.5,-4) -- node[rotate = 55, above, scale = .85]{$0$} (RRRL);
        \draw (33.5,-4) -- node[rotate = 305, above, scale = .85] {$1$} (RRRR);
    \end{tikzpicture}
    \caption{FRMCISTGs for restricted patient preference trials, where patients are randomized between following their natural treatment values or entering a conventional sequentially randomized trial.}
    \label{fig: Trial Class III FRMCISTG}
\end{figure}

    \subsection{Proof of Proposition \ref{prop: recorded nat value trial id}}
    The proof follows from the trial design and Proposition \ref{prop: initiation regime id}.

    \subsection{Proof of Proposition \ref{prop: three choice trial id}}
    The proof follows from the trial design and Proposition \ref{prop: initiation regime id}.

    \subsection{Proof of Proposition \ref{prop: multiple treatments}}

For identifiability of the superoptimal regime, it is useful to think of each intervention as a separate counterfactual world. That is, at every time $k$, conditional on history $H_k^g$, if we intervene to set $A_k^g$ to $a_k \in \mathcal{A}_k$, then:
\begin{enumerate}
    \item The superoptimal regime uses the natural treatment value $A_k^g$ and thus needs to evaluate the effect of intervention $a_k$ for all possible values of $A_k^g$, that is, there are $|\mathcal{A}_k|$ unknowns we need to evaluate.
    \item The sequentially randomized trial data and the observed data give us information on the unknown values through the law of total expectation:
    \begin{align*}
        \mathbb{E}[R_k^{\overline{A}_{k-1}^{g+}, a_k} | H_k^g] &= \sum_{a_k' \in \mathcal{A}_k} \mathbb{E}[R_k^{\overline{A}_{k-1}^{g+}, a_k} | H_k^g, A_k^g = a_k'] P(A_k^g = a_k' | H_k^g)
    \end{align*}
    A sequentially randomized trial gives us the value of the left hand term. The observed data gives us the value of $\mathbb{E}[R_k^{\overline{A}_{k-1}^{g+}, a_k} | H_k^g, A_k^g = a_k']$ and the $P(A_k^g = a_k' | H_k^g)$ if and only if $\overline{A}_{k-1}^{g+} = \overline{A}_{k-1}$ in Trial \ref{algo: obs + trial }. However, if $|\mathcal{A}_k| > 2$, there will always be one equation for more than one unknowns and hence the value functions for $g^{\textbf{sup}}$ will not be identified in general.
\end{enumerate}

We derive the results for each design separately.
    \begin{itemize}
        \item \textbf{Data Structure \ref{algo: obs + trial }:} Suppose that Data Structure \ref{algo: obs + trial } can identify $g^{\textbf{sup}}$ without additional assumptions. Then, for any set of unmeasured confounders $\overline{U}_K$, take the data-generating mechanisms where the natural treatment values $A_k^g$ are equal to the $(H,U)$-optimal regime, that is, the $H$-optimal regime conditional on history $H_k^g$ and all unmeasured confounders $\overline{U}_k$. Then, as we assumed that $g^{\textbf{sup}}$ is identified, this implies that for any set of unmeasured confounders, Data Structure \ref{algo: obs + trial } will identify the optimal treatment for all units, as it does not know the distribution of the $A_k^g$ conditional on $\overline{U}_k$ and $H_k^g$ and it is a possibility that they are equal to the $(H,U)$-optimal regime. However, it is also possible that $A_k^g$ is equal to the worst possible treatment conditional on $H_k^g$ and $\overline{U}_k$. As we do not observe the $U_k$, and the value functions conditional on $A_k^g$ are not identified as discussed above, Data Structure \ref{algo: obs + trial } cannot distinguish between these two settings when the $A_k^g$ are counterfactual, unmeasured, covariates.
        
        \item \textbf{Data Structure \ref{algo: randomized trial obs}:} The argument in Appendix \ref{app: Dirac identification} does not hold when there are more than two treatments as there are then less equations than unknowns, that is, the same argument as for Data Structure \ref{algo: obs + trial } holds as we do not observe counterfactual natural treatment values after interventions.
        
        \item \textbf{Data Structure \ref{algo: recorded nat value trial}:} It is immediate to see that this design identifies $g^{\textbf{sup}}$ if the recorded natural treatment values are accurate.
        
        \item \textbf{Data Structure \ref{algo: three choice trial }:} This design does not identify $g^{\textbf{sup}}$ at time $k$ when the $|\mathcal{A}_k| > 2$ as the law of total expectation has $|\mathcal{A}_k| - 1 > 1$ unknowns, the natural treatment values $A_k^g$ that are not recorded when intervening.
    \end{itemize}

\subsection{Proof of Proposition \ref{prop: initiation regime id}}
    The optimal initiation regime can be decomposed in three steps:
    \begin{enumerate}
        \item It follows the observed regime until a time-point $j$ (potentially $j = K$),
        \item It performs an ``initiation" step at time-point $j + 1$ (if $j < K$),
        \item It follows the $H$-optimal regime after the initiation step, conditional on the observed history up to time $j + 1$ (if $j + 2 \leq K$).
    \end{enumerate}
    We show that $g^{\textbf{init}}_k$ is identified for all $k = 1, \ldots, K$ by considering three cases:
    \begin{enumerate}
        \item $\mathbf{k \leq j:}$ For $k \leq j$, $g_k^{\textbf{init}} = A_k$, which is observed. Hence, $g_k^{\textbf{init}}$ is identified.

        \item $\mathbf{k = j + 1:}$ As $g^{\textbf{init}}$ follows the observed regime until time-point $j$, we have by consistency that at time $k = j + 1$,
    \begin{align*}
        g_{k}^{\textbf{init}}(h_k, \overline{a}_k') &= \argmax_{a_k \in \{0,1\}} \mathbb{E}[R_k^{a_k} + \mathcal{V}_{P,k + 1}^g(H_{k+1}^g)|  \overline{L}_k = \overline{l}_k, \overline{R}_{k-1} = \overline{r}_{k-1}, \overline{A}_{k-1} = \overline{a}_{k-1}', A_k = a_k'].
    \end{align*}
    However, by assumption we can identify $\mathbb{E}[R_k^{a_k} | \overline{L}_k, \overline{R}_{k-1}, \overline{A}_{k-1}]$. Hence, by Lemma 1 in \citet{stensrud_optimal_2024}, we can identify $\mathbb{E}[R_k^{a_k} |  \overline{L}_k, \overline{R}_{k-1}, \overline{A}_{k-1}, A_k]$, which implies that $g_k^{\textbf{init}}$ is identified if $\mathcal{V}_{P,k + 1}^g(H_{k+1}^g)$ is identified.

    \item  $\mathbf{k > j + 1:}$ For $k > j + 1$, suppose that all $g^{\textbf{init}}_i$ are identified and constant for $i < k$ conditional on history $(H_i, \overline{A}_j = \overline{A}_j^{g^{\textbf{init}}+})$ from time $j$, 
    Then, 
    \begin{align*}
        g_k^{\textbf{init}}(h_k , \overline{a}_k') &= \argmax_{a_k \in \{0,1\}} \mathbb{E}[R_k^{\overline{g}_k^{\textbf{init}}} + \mathcal{V}_{P,k + 1}^g(H_{k+1}^g, \overline{A}_{k + 1}) | H_k^g = h_k, \overline{A}_k^g = \overline{a}_k'],
    \end{align*}
    which is identified by assumption, provided $\mathcal{V}_{P,k + 1}^g(H_{k+1}^g, \overline{A}_{k + 1}^g)$ is identified. 
    \end{enumerate}
    However,
    \begin{align*}
          \mathcal{V}_{P,K}^g(H_{K}, \overline{A}_K) &= \mathbb{E}[R_K^{\overline{g}_K^{\textbf{init}}} | H_K^g, \overline{A}_K^g]
    \end{align*}
    is in one of the three previous cases and thus is identified. The result then follows by the induction principle.

  \subsection{Proof of Proposition \ref{prop: estimation convergence}}
    \label{app: estimation convergence prop proof}


    We use stochastic process notation, that is, let $Pf := \mathbb{E}_P[f(O)]$ and $\mathbb{P}_n f := \frac{1}{n}\sum_{i = 1}^n f(O_i)$. Then,
    \begin{align*}
        &\hat{\mathbb{E}}(\sum_{k = 1}^K R_k^{\hat{g}}) - \mathbb{E}(\sum_{k = 1}^K R_k^g)\\
        &= \hat{\mathbb{E}}(\sum_{k = 1}^K R_k^g) - \mathbb{E}(\sum_{k = 1}^K R_k^g) + \underbrace{\hat{\mathbb{E}}(\sum_{k = 1}^K R_k^{\hat{g}} - R_k^g)}_{= \mathcal{E}^{\hat{g}}}\\
        &= \sum_{k = 1}^K (\hat{\mathbb{E}}(R_k^g) - \mathbb{E}(R_k^g)) + \mathcal{E}^{\hat{g}}\\
        &= \sum_{k = 1}^K (\mathbb{P}_n(\hat{f}_k^g - f_k^g) + \mathcal{E}^{\hat{g}}\\
        &= \sum_{k = 1}^K \mathbb{P}_n(\hat{f}_k^g - f_k^g) + (\mathbb{P}_n - P)f_k^g + \mathcal{E}^{\hat{g}}\\
        &= \sum_{k = 1}^K (\mathbb{P}_n-P)(\hat{f}_k^g - f_k^g)+ P(\hat{f}_k^g - f_k^g) + (\mathbb{P}_n - P)f_k^g + \mathcal{E}^{\hat{g}}
    \end{align*}
    However, $(\mathbb{P}_n-P)(\hat{f}_k^g - f_k^g)$ is $o_P(n^{-1/2})$ because of conditions \ref{assL2: 1} and \ref{assL2: 7} of Proposition \ref{prop: estimation convergence}, $P(\hat{f}_k^g - f_k^g)$ is $o_P(n^{-1/2})$ because of condition \ref{assL2: 6}, and $\sqrt{n}(\mathbb{P}_n - P)f_k^g$ is asymptotically normal because of the central limit theorem.
    
    Furthermore, 
    \begin{align*}
        \mathcal{E}^{\hat{g}} = \hat{\mathbb{E}}(R^{\hat{g}}) - \hat{\mathbb{E}}(R^{g})
         = \sum_{k=1}^K \Big[(\mathbb{P}_n - P)(\hat{f}_k^{\hat{g}} - \hat{f}_k^{g}) + P(\hat{f}_k^{\hat{g}} - \hat{f}_k^{g})\Big].
    \end{align*}
    Let $D_k := \max_{j \leq k} \big(\hat{g}_j(H_j^g, \overline{A}_j^{g+}) - g_j(H_j^g, \overline{A}_j^{g+})\big)^2 = I(\hat{g}_j \neq g_j \text{ for some } j \leq k)$. When $D_k = 0$, $\hat{g}$ and $g$ assign the same treatments until stage $k$, so $R_k^{\hat{g}} = R_k^g$ and $H_k^{\hat{g}} = H_k^g$ by consistency, and hence $\hat{f}_k^{\hat{g}} = \hat{f}_k^{g}$. Therefore
    \begin{align*}
        \hat{f}_k^{\hat{g}} - \hat{f}_k^{g} = (\hat{f}_k^{\hat{g}} - \hat{f}_k^{g}) D_k \text{ and }
         |\hat{f}_k^{\hat{g}} - \hat{f}_k^{g}| \leq 2C D_k
    \end{align*}
    by condition \ref{assL2: 3}. For the second term, condition \ref{assL2: 9} gives
    \begin{align*}
        |P(\hat{f}_k^{\hat{g}} - \hat{f}_k^{g})| \leq 2C\, P(D_k = 1)
        \leq 2C \sum_{j \leq k} P(\hat{g}_j \neq g_j)
        = 2C \sum_{j \leq k} \|\hat{g}_j - g_j\|_{L_2(P)}^2 = o_P(n^{-1/2}).
    \end{align*}
    
    For the first term, $\|(\hat{f}_k^{\hat{g}} - \hat{f}_k^{g}) D_k\|_{L_2(P)}^2 \leq 4C^2 P(D_k = 1) \to_P 0$, and by conditions \ref{assL2: 7} and \ref{assL2: 8} the function $(\hat{f}_k^{\hat{g}} - \hat{f}_k^{g}) D_k$ lies in a $P$-Donsker class (products of uniformly bounded Donsker classes are Donsker \citep[Example 2.10.8]{van1996weak}), so $(\mathbb{P}_n - P)((\hat{f}_k^{\hat{g}} - \hat{f}_k^{g}) D_k) = o_P(n^{-1/2})$ by \citet[Lemma 19.24]{van2000asymptotic}. Hence $\mathcal{E}^{\hat{g}} = o_P(n^{-1/2})$, and $\sqrt{n}(\hat{\mathbb{E}}(R^{\hat{g}}) - \mathbb{E}(R^{g}))$ has the same limit distribution as $\sqrt{n}(\hat{\mathbb{E}}(R^{g}) - \mathbb{E}(R^{g}))$.

    \subsection{Proof of Proposition \ref{prop: IF asymptotic normality}}
    \label{app: IF asymptotic normality prop proof}

To prove Proposition \ref{prop: IF asymptotic normality}, we introduce the following lemma for influence functions:

\begin{lemma}
Fix $k$, $\overline{a}_{k-1}$ and $a_k \neq a_k'$, and let $S_k := I(\overline{A}_{k-1} = \overline{a}_{k-1})$,
$\omega_k(H_k) := \mathbb{E}[R_k^{a_k} | H_k, \overline{A}_{k-1} = \overline{a}_{k-1}]$ and
\begin{align*}
    \psi_k := \mathbb{E}\big[S_k\, \mathbb{E}[R_k^{a_k} | H_k, \overline{A}_{k-1} = \overline{a}_{k-1}, A_k = a_k']\, I(A_k = a_k')\big]
    = \mathbb{E}\big[S_k \{\omega_k(H_k) - R_k I(A_k = a_k)\}\big],
\end{align*}
where the equality follows from Lemma \ref{lemma: lemma 1 analog}. Let $\hat{\omega}_k$ be an estimator of
$\omega_k$ built from an outcome regression $\hat{\mu}_k$ and a treatment (or trial-entry) probability
$\hat{\pi}_k$, with influence-function correction $\hat{\varphi}_k$ satisfying $P(S_k \varphi_k) = 0$ and
\begin{align*}
    \big|P\{S_k(\hat{\omega}_k + \hat{\varphi}_k - \omega_k)\}\big| \leq c\, \|\hat{\mu}_k - \mu_k\|_{L_2(P)}\, \|\hat{\pi}_k - \pi_k\|_{L_2(P)}
\end{align*}
for some $c < \infty$, and let $\tilde{\psi}_k := \mathbb{P}_n\big[S_k\{\hat{\omega}_k + \hat{\varphi}_k - R_k I(A_k = a_k)\}\big]$. If
\begin{enumerate}[label=(\roman*)]
  \item $\|\hat{\mu}_k - \mu_k\|_{L_2(P)}\, \|\hat{\pi}_k - \pi_k\|_{L_2(P)} = o_P(n^{-1/2})$, \label{assK1: i}
  \item $\|\hat{\omega}_k + \hat{\varphi}_k - \omega_k - \varphi_k\|_{L_2(P)} = o_P(1)$, \label{assK1: ii}
  \item $\hat{\omega}_k + \hat{\varphi}_k$ is $P$-Donsker (or estimated in a separate sample), \label{assK1: iii}
  \item $\hat{\pi}_k \in (\epsilon, 1-\epsilon)$ and $|R_k| \leq C$ with probability one, \label{assK1: iv}
\end{enumerate}
then $\sqrt{n}(\tilde{\psi}_k - \psi_k) \to^d N\big(0, \mathrm{Var}[S_k\{\omega_k + \varphi_k - R_k I(A_k = a_k)\}]\big)$.
\label{lemma: one-step convergence result}
\end{lemma}
\begin{proof}
Let $f_k := S_k\{\omega_k + \varphi_k - R_k I(A_k = a_k)\}$ and $\hat{f}_k$ its estimated version, so that
$P f_k = \psi_k$ and $\tilde{\psi}_k = \mathbb{P}_n \hat{f}_k$. Then
\begin{align*}
    \tilde{\psi}_k - \psi_k = (\mathbb{P}_n - P) f_k + (\mathbb{P}_n - P)(\hat{f}_k - f_k) + P(\hat{f}_k - f_k).
\end{align*}
The term $R_k I(A_k = a_k)$ is the same in $\hat{f}_k$ and $f_k$ and $P(S_k \varphi_k) = 0$, so
$P(\hat{f}_k - f_k) = P\{S_k(\hat{\omega}_k + \hat{\varphi}_k - \omega_k)\} = o_P(n^{-1/2})$ by \ref{assK1: i}. By \ref{assK1: ii}-\ref{assK1: iv} and \citet[Lemma 19.24]{van2000asymptotic}, $(\mathbb{P}_n - P)(\hat{f}_k - f_k)
= o_P(n^{-1/2})$. The first term is asymptotically normal by \ref{assK1: iv} and the central limit
theorem. 
 In Data Structure \ref{algo: randomized trial obs}, with $E_k$ the trial-entry
 indicator, $\hat{\varphi}_k = E_k I(A_k^{g+} = a_k)(R_k - \hat{\mu}_k)/(\pi_k p_k)$ and $R_k I(A_k = a_k)$
 replaced by $(1 - E_k) R_k I(A_k = a_k)/(1 - \pi_k)$, the tower property over $(E_k, A_k^{g+})$ given
 $(H_k, \overline{A}_{k-1})$ gives $P\{S_k(\hat{\omega}_k + \hat{\varphi}_k - \omega_k)\} =
 P\{S_k(\hat{\mu}_k - \mu_k)(1 - \pi_k p_k/(\pi_k p_k))\} = 0$.
\end{proof}

    Lemma \ref{lemma: one-step convergence result} gives a conditional convergence result for the one-step estimator at an ``initiation" point $k$. 

     Then,
    \begin{align*}
        &\sum_{k = 1}^K \mathbb{P}_n[\hat{\mathbb{E}}[R_k^g | H_k^g, \overline{A}_{\max \{j < k: \overline{A}_j^{g+} = \overline{A}_j\}}] + \hat{\mathbb{IF}}(\mathbb{E}[R_k^g | H_k^g, \overline{A}_{\max \{j < k: \overline{A}_j^{g+} = \overline{A}_j\}}])] - \mathbb{E}[R^g]\\
        &= \sum_{k = 1}^K (\mathbb{P}_n-P)[\hat{\mathbb{E}}[R_k^g | H_k^g, \overline{A}_{\max \{j < k: \overline{A}_j^{g+} = \overline{A}_j\}}] + \hat{\mathbb{IF}}(\mathbb{E}[R_k^g | H_k^g, \overline{A}_{\max \{j < k: \overline{A}_j^{g+} = \overline{A}_j\}}])]\\
        &+ P[\hat{\mathbb{E}}[R_k^g | H_k^g, \overline{A}_{\max \{j < k: \overline{A}_j^{g+} = \overline{A}_j\}}] + \hat{\mathbb{IF}}(\mathbb{E}[R_k^g | H_k^g, \overline{A}_{\max \{j < k: \overline{A}_j^{g+} = \overline{A}_j\}}])\\
        &- \mathbb{E}[R_k^g | H_k^g, \overline{A}_{\max \{j < k: \overline{A}_j^{g+} = \overline{A}_j\}}] - \mathbb{IF}(\mathbb{E}[R_k^g | H_k^g, \overline{A}_{\max \{j < k: \overline{A}_j^{g+} = \overline{A}_j\}}])]\\
        &= \sum_{k = 1}^K (\mathbb{P}_n-P)[\hat{\mathbb{E}}[R_k^g | H_k^g, \overline{A}_{\max \{j < k: \overline{A}_j^{g+} = \overline{A}_j\}}] + \hat{\mathbb{IF}}(\mathbb{E}[R_k^g | H_k^g, \overline{A}_{\max \{j < k: \overline{A}_j^{g+} = \overline{A}_j\}}])\\
        &- \mathbb{E}[R_k^g | H_k^g, \overline{A}_{\max \{j < k: \overline{A}_j^{g+} = \overline{A}_j\}}] - \mathbb{IF}(\mathbb{E}[R_k^g | H_k^g, \overline{A}_{\max \{j < k: \overline{A}_j^{g+} = \overline{A}_j\}}])]\\
        &+ (\mathbb{P}_n - P)[ \mathbb{E}[R_k^g | H_k^g, \overline{A}_{\max \{j < k: \overline{A}_j^{g+} = \overline{A}_j\}}] + \mathbb{IF}(\mathbb{E}[R_k^g | H_k^g, \overline{A}_{\max \{j < k: \overline{A}_j^{g+} = \overline{A}_j\}}])]\\
        &+ P[\hat{\mathbb{E}}[R_k^g | H_k^g, \overline{A}_{\max \{j < k: \overline{A}_j^{g+} = \overline{A}_j\}}] + \hat{\mathbb{IF}}(\mathbb{E}[R_k^g | H_k^g, \overline{A}_{\max \{j < k: \overline{A}_j^{g+} = \overline{A}_j\}}])\\
        &- \mathbb{E}[R_k^g | H_k^g, \overline{A}_{\max \{j < k: \overline{A}_j^{g+} = \overline{A}_j\}}] - \mathbb{IF}(\mathbb{E}[R_k^g | H_k^g, \overline{A}_{\max \{j < k: \overline{A}_j^{g+} = \overline{A}_j\}}])].
    \end{align*}
    The result is obtained using Lemma \ref{lemma: one-step convergence result}, the Donsker and boundedness properties, and the central limit theorem.

    The proof for $\hat{g}$ follows similarly as in the proof of Proposition \ref{prop: estimation convergence}.

    \subsection{Proof of Proposition \ref{prop: MDP sup better than opt}}
By the law of total expectation, 
\begin{align*}
    \mathbb{E}[R_k^g] = \mathbb{E}[\mathbb{E}[R_k^g | H_k^g]].
\end{align*}
However, by Assumption \ref{ass: Markov general}, 
\begin{align*}
    \mathbb{E}[R_k^g | H_k^g] = \mathbb{E}[R_k^{g_k} | L_k].
\end{align*}
Because $\mathbb{E}[R_k^g | H_k^g]$ is identified for any $g$ and $H_k^g$ by assumption, we have that $\mathbb{E}[R^a | L]$ is identified with probability one for any $a \in \{0,1\}$.

Therefore, by Lemma 1 in \citet{stensrud_optimal_2024}, $\mathbb{E}[R_k^a | L_k, A_k]$ is identified with probability one for any $a \in \{0,1\}$. Thus, $\Tilde{g}_k^{\textbf{sup}}$, and hence $g^{\textbf{sup}}$, are identified by definition. Furthermore, by Proposition 1 of \cite{stensrud_optimal_2024},
\begin{align*}
    \mathbb{E}[R_k^{g^{\textbf{opt}}_k} | L_k] \leq \mathbb{E}[R_k^{\Tilde{g}_k^{\textbf{sup}}} |L_k] = \mathbb{E}[R_k^{g^{\textbf{sup}}} | H_k^g].
\end{align*}

Hence,
\begin{align*}
    \mathbb{E}[\sum_{k = 1}^K R_k^{g^{\textbf{opt}}}] &= \sum_{k = 1}^K \mathbb{E}[R_k^{g^{\textbf{opt}}}]\\
    &= \sum_{k = 1}^K \mathbb{E}[\mathbb{E}[R_k^{g^{\textbf{opt}}} | H_k^g ]]\\
    &= \sum_{k = 1}^K \mathbb{E}[\mathbb{E}[R^{g^{\textbf{opt}}_k} | L = L_k ]]\\
    & \leq \sum_{k = 1}^K \mathbb{E}[\mathbb{E}[R_k^{g^{\textbf{sup}}} | H_k^g]] =  \mathbb{E}[\sum_{k = 1}^K R_k^{g^{\textbf{sup}}}].
\end{align*}

\subsection{Proof of Proposition \ref{prop: shi mdp implies forgetfulness}}
    If Assumption \ref{ass: shi mdp} holds, 
\begin{align*}
\mathbb{E}[R_k^g | H_k^g, \overline{A}_k^g]
 &= \mathbb{E}\big[\mathbb{E}[R_k^g | H_k^g, \overline{A}_k^g, U_k] \big| H_k^g, \overline{A}_k^g\big]\\
 &= \mathbb{E}\big[\mathbb{E}[R_k | H_k, \overline{A}_k = \overline{g}_k, U_k] \big| H_k^g, \overline{A}_k^g\big]\\
 &= \mathbb{E}\big[\mathbb{E}[R_k | L_k, A_k = g_k, U_k] \big| H_k^g, \overline{A}_k^g\big]\\
 &= \mathbb{E}\big[\mathbb{E}[R_k | L_k, A_k = g_k, U_k] \big| H_k^g, A_k^g\big]\\
 &= \mathbb{E}\big[\mathbb{E}[R_k^g | H_k^g, A_k^g, U_k] \big| H_k^g, A_k^g\big]
  = \mathbb{E}[R_k^g | H_k^g, A_k^g],
\end{align*}
    where the first equality follows from the law of total expectation, the second equality follows from unconfoundedness when conditioning on the $U_k$, $R_k^g \independent \overline{A}_k^g | H_k^g, U_k$, the third and fourth equality follow from Assumption \ref{ass: shi mdp}, where Assumption \ref{ass: shi mdp} is taken to hold for the counterfactual variables under $g$, the fifth equality follows from unconfoundedness when conditioning on $U_k$, and the last equality follows from the law of total expectation.

\subsection{Proof of Proposition \ref{prop: forgetfulness superopt identification}}
    First, we prove that $g^{\textbf{s}}$ is the superoptimal regime. For $g \in \mathcal{G}^{\textbf{sup}}$, the assigned treatments $\overline{A}_{k-1}^{g+}$ are a deterministic function of $(H_{k-1}^g, \overline{A}_{k-1}^g)$, so we may write $g_k = g_k(h_k, \overline{a}_{k-1}, \overline{a}_k')$, where $\overline{a}_{k-1}$ denotes the assigned and $\overline{a}_k'$ the natural treatment values. Consider the Q-functions of the class $\mathcal{G}^{\textbf{sup}}$,
    \begin{align*}
        Q_K^{\textbf{sup}}(h_K, \overline{a}_{K-1}, \overline{a}_K'; a_K) &:= \mathbb{E}[R_K^{\overline{a}_{K-1}, a_K} | H_K^{\overline{a}_{K-1}} = h_K, \overline{A}_K^{\overline{a}_{K-1}} = \overline{a}_K'],\\
        Q_k^{\textbf{sup}}(h_k, \overline{a}_{k-1}, \overline{a}_k'; a_k) &:= \mathbb{E}[R_k^{\overline{a}_{k-1}, a_k} + V_{k+1}^{\textbf{sup}}(H_{k+1}^{\overline{a}_k}, \overline{a}_k, \overline{A}_{k+1}^{\overline{a}_k}) | H_k^{\overline{a}_{k-1}} = h_k, \overline{A}_k^{\overline{a}_{k-1}} = \overline{a}_k'],
    \end{align*}
    and $V_k^{\textbf{sup}}(h_k, \overline{a}_{k-1}, \overline{a}_k') := \max_{a_k \in \{0,1\}} Q_k^{\textbf{sup}}(h_k, \overline{a}_{k-1}, \overline{a}_k'; a_k)$.
    By the standard dynamic programming argument \citep{bellman_mdp_1957, murphy2003optimal}, backward induction gives $\mathbb{E}[\sum_{j = k}^K R_j^g | H_k^g, \overline{A}_k^g] \leq V_k^{\textbf{sup}}(H_k^g, \overline{A}_{k-1}^{g+}, \overline{A}_k^g)$ for every $g \in \mathcal{G}^{\textbf{sup}}$, with equality when $g_j$ maximizes $Q_j^{\textbf{sup}}$ for all $j \geq k$. Hence $g^{\textbf{sup}}_k = \argmax_{a_k} Q_k^{\textbf{sup}}$.
    
    Analogously, let
    \begin{align*}
        Q_K^{\textbf{s}}(h_K, \overline{a}_{K-1}, a_K'; a_K) &:= \mathbb{E}[R_K^{\overline{a}_{K-1}, a_K} | H_K^{\overline{a}_{K-1}} = h_K, A_K^{\overline{a}_{K-1}} = a_K'],\\
        Q_k^{\textbf{s}}(h_k, \overline{a}_{k-1}, a_k'; a_k) &:= \mathbb{E}[R_k^{\overline{a}_{k-1}, a_k} + V_{k+1}^{\textbf{s}}(H_{k+1}^{\overline{a}_k}, \overline{a}_k, A_{k+1}^{\overline{a}_k}) | H_k^{\overline{a}_{k-1}} = h_k, A_k^{\overline{a}_{k-1}} = a_k'],
    \end{align*}
    and $V_k^{\textbf{s}} := \max_{a_k \in \{0,1\}} Q_k^{\textbf{s}}$, so that $g^{\textbf{s}}_k = \argmax_{a_k} Q_k^{\textbf{s}}$ and $V_k^{\textbf{s}}$ is the value function $\mathcal{V}_{P,k}^{\overline{a}_{k-1}}$ of $g^{\textbf{s}}$.

    We show by backward induction on $k$ that, for all $\overline{a}_{k-1}$ and $a_k$,
    \begin{equation}
        Q_k^{\textbf{sup}}(h_k, \overline{a}_{k-1}, \overline{a}_k'; a_k) = Q_k^{\textbf{s}}(h_k, \overline{a}_{k-1}, a_k'; a_k) \quad \text{a.s.,}
        \label{eq: Q sup equals Q s}
    \end{equation}
    which implies $V_k^{\textbf{sup}} = V_k^{\textbf{s}}$ and $g_k^{\textbf{sup}} = g_k^{\textbf{s}}$. For $k = K$, $Q_k^{\textbf{sup}} = Q_k^{\textbf{s}}$ by Assumption \ref{Ass: Forgetfulness} and hence $g_k^{\textbf{sup}} = g_k^{\textbf{s}}$. Now, suppose \eqref{eq: Q sup equals Q s} holds at time $k+1$. Then $V_{k+1}^{\textbf{sup}}(H_{k+1}^{\overline{a}_k}, \overline{a}_k, \overline{A}_{k+1}^{\overline{a}_k}) = V_{k+1}^{\textbf{s}}(H_{k+1}^{\overline{a}_k}, \overline{a}_k, A_{k+1}^{\overline{a}_k})$ a.s., and since $H_{k+1}^{\overline{a}_k} = (H_k^{\overline{a}_{k-1}}, R_k^{\overline{a}_k}, L_{k+1}^{\overline{a}_k})$, there is a function $\phi_k$ such that
    \begin{align*}
        R_k^{\overline{a}_k} + V_{k+1}^{\textbf{sup}}(H_{k+1}^{\overline{a}_k}, \overline{a}_k, \overline{A}_{k+1}^{\overline{a}_k}) = \phi_k(H_k^{\overline{a}_{k-1}}, R_k^{\overline{a}_k}, L_{k+1}^{\overline{a}_k}, A_{k+1}^{\overline{a}_k}) \quad \text{a.s.}
    \end{align*}
    Therefore,
    \begin{align*}
        Q_k^{\textbf{sup}}(h_k, \overline{a}_{k-1}, \overline{a}_k'; a_k) &= \mathbb{E}[\phi_k(H_k^{\overline{a}_{k-1}}, R_k^{\overline{a}_k}, L_{k+1}^{\overline{a}_k}, A_{k+1}^{\overline{a}_k}) | H_k^{\overline{a}_{k-1}} = h_k, \overline{A}_k^{\overline{a}_{k-1}} = \overline{a}_k']\\
        &= \mathbb{E}[\phi_k(H_k^{\overline{a}_{k-1}}, R_k^{\overline{a}_k}, L_{k+1}^{\overline{a}_k}, A_{k+1}^{\overline{a}_k}) | H_k^{\overline{a}_{k-1}} = h_k, A_k^{\overline{a}_{k-1}} = a_k']\\
        &= Q_k^{\textbf{s}}(h_k, \overline{a}_{k-1}, a_k'; a_k),
    \end{align*}
    where the second equality follows from Assumption \ref{Ass: Forgetfulness future}, since $(R_k^{\overline{a}_k}, L_{k+1}^{\overline{a}_k}, A_{k+1}^{\overline{a}_k})$ is independent of $\overline{A}_{k-1}^{\overline{a}_{k-2}}$ given $(H_k^{\overline{a}_{k-1}}, A_k^{\overline{a}_{k-1}})$, and the last equality follows from the definition of $Q_k^{\textbf{s}}$ and the induction hypothesis $V_{k+1}^{\textbf{sup}} = V_{k+1}^{\textbf{s}}$. This proves \eqref{eq: Q sup equals Q s} for all $k$, and thus $g^{\textbf{s}} = g^{\textbf{sup}}$.
    
    Now, we prove that $g^{\textbf{s}}$ is identified. By definition, for any $k = 1, \ldots, K$,
    \begin{align*}
        &g_k^{\textbf{s}}(h_k, \overline{A}_{k-1}^{g+}, a_k')\\
        &= \argmax_{a_k\in \{0,1\}} \mathbb{E}[R_k^{\overline{A}_{k-1}^{g+}, a_k} + V_{k+1}^{\textbf{s}}(H_{k+1}^g, \overline{A}_{k + 1}^g)|H_k^g = h_k, A_k^g = a_k'], \\
        &= \argmax_{a_k\in \{0,1\}} \mathbb{E}[R_k^{\overline{A}_{k-1}^{g+}, a_k} + V_{k+1}^{\textbf{s}}(H_{k+1}^g, \overline{A}_{k + 1}^g)|H_k^g = h_k, A_k^g = a_k', \overline{A}_{k-1}^{g} = \overline{A}_{k-1}^{g+}] \\
        &= \argmax_{a_k\in \{0,1\}} \mathbb{E}[R_k^{a_k} + V_{k+1}^{\textbf{s}}(H_{k+1}^g, \overline{A}_{k + 1}^g)|H_k = h_k, A_k = a_k', \overline{A}_{k-1} = \overline{A}_{k-1}^{g+}],
    \end{align*}
    where the second equality follows by Assumption \ref{Ass: Forgetfulness future} and the last equality follows by Assumption \ref{ass: consistency}.
    
    If $\mathbb{E}[R_k^{a_k} + \mathcal{V}_{P,k+1}^{\overline{a}_k}(H_{k+1}^{\overline{a}_k}, \overline{A}_{k + 1}^{\overline{a}_k}) | \overline{L}_k, \overline{R}_{k-1}, \overline{A}_{k-1}]$ is identified, $$\mathbb{E}[R_k^{a_k} + \mathcal{V}_{P,k+1}^{\overline{a}_k}(H_{k+1}^{\overline{a}_k}, \overline{A}_{k + 1}^{\overline{a}_k}) | \overline{L}_k,\overline{R}_{k-1}, A_k, \overline{A}_{k-1}]$$ is identified by Lemma 1 in \citet{stensrud_optimal_2024}.

\bibliographystyle{plainnat}
\spacingset{0.6}
\setlength{\bibsep}{2pt}
\bibliography{references1}

\end{document}